\documentclass{elsarticle}
\usepackage{amsmath,amsfonts,amsthm}
\usepackage{amssymb,latexsym}

\theoremstyle{plain}
\newtheorem{theorem}{Theorem}[section]
\newtheorem{lemma}[theorem]{Lemma}
\newtheorem{proposition}[theorem]{Proposition}

\theoremstyle{definition}
\newtheorem{mydefn}[theorem]{Definition}
\newtheorem{example}[theorem]{Example}
\newtheorem{observation}[theorem]{Observation}
\newtheorem{convention}[theorem]{Convention}

\usepackage{mathrsfs}
\usepackage{amsfonts}
\usepackage{graphicx}
\usepackage{slashed}

\begin{document}

\author[nuaa]{Yan Zhang}

\author[nuaa]{Zhaohui Zhu\corref{cor}}

\author[nanshen]{Jinjin Zhang}

\author[nuaa]{Yong Zhou}

\cortext[cor]{Corresponding author. Email: zhaohui.nuaa@gmail.com, bnj4892856@jlonline.com (Zhaohui Zhu).}

\address[nuaa]{College of Computer Science, Nanjing University of Aeronautics and Astronautics, Nanjing, P.R. China, 210016}
\address[nanshen]{College of Information Science, Nanjing Audit University, \\Nanjing, P.R. China, 211815}

\title{A Process Calculus with Logical Operators\tnoteref{t1}}
\tnotetext[t1]{This work received financial support of the National Natural Science of China(No. 60973045) and Fok Ying-Tung Education Foundation.}
\date{\today}
\begin{abstract}

In order to combine operational and logical styles of specifications in one unified framework, the notion of logic labelled transition systems (Logic LTS, for short) has been presented and explored by L\"{u}ttgen and Vogler in [TCS 373(1-2):19-40; Inform. \& Comput. 208:845-867].
In contrast with usual LTS, two logical constructors $\wedge$ and $\vee$ over Logic LTSs are introduced to describe logical combinations of specifications.
Hitherto such framework has been dealt with in considerable depth, however, process algebraic style way has not yet been involved and the axiomatization of constructors over Logic LTSs is absent.
This paper tries to develop L\"{u}ttgen and Vogler's work along this direction.
We will present a process calculus for Logic LTSs (CLL, for short).
The language CLL is explored in detail from two different but equivalent views.
Based on behavioral view, the notion of ready simulation is adopted to formalize the refinement relation, and the behavioral theory is developed.
Based on proof-theoretic view, a sound and ground-complete axiomatic system for CLL is provided, which captures operators in CLL through (in)equational laws.
\end{abstract}

\begin{keyword}
  Logic operators \sep Process calculus \sep Ready simulation \sep Axiomatic system \sep Logic labelled transition system
\end{keyword}

\maketitle
%

\section{Introduction}

    Over the past three decades, a lot of approaches have been proposed to formally specify and reason about reactive systems. Process algebra \cite{Handbook} and temporal logic \cite{Pnueli} are two popular paradigms of them.

    In process-algebraic paradigm, a system specification and its implementation usually are formulated in the same notation, and the underlying semantics are often given operationally.
    The notion of refinement is adopted to capture the correctness of implementation.

    In contrast, the temporal-logic paradigm adopts the language of temporal logics to formulate specifications abstractly, and implementations are described in terms of operational notations \cite{Pnueli}.
    Usually, model checking technique is used to establish that the system satisfies its specification \cite{ModelChecking}.

    Process-algebraic paradigm supports compositional reasoning (i.e., refinement of one part of a system does not depend on others), which is one of the most significant advantages of it.
    The merit of temporal-logic paradigm lies in its support for  abstract specifications, where relevant operational details may not be concerned.
    Traditionally, process-algebraic and temporal-logic formalisms are not mixed.

    In order to take advantage of both paradigms when designing systems, some theories for heterogeneous specifications have been proposed, e.g., \cite{Cleaveland00, Cleaveland02, Graf86,Kurshan94,Olderog}, which uniformly integrate refinement-based and temporal-logic specification styles.
    Among them,  Cleaveland and L{\"u}ttgen present a semantic framework for heterogenous system design based on B{\"u}chi automata and labelled transition systems augmented with an unimplementability predicate, and adopt Nicola and Hennessy's must-testing preorder \cite{Nicola83} to describe refinement relation \cite{Cleaveland00, Cleaveland02}.
    However, such refinement relation is not a precongruence \cite{Cleaveland00}.
    Hence, it does not support compositional reasoning.
    Moreover, in such framework, the operator conjunction lacks the desired property that $r$ is an implementation of the specification $p \wedge q$ if and only if $r$ implements both $p$ and $q$ \cite{Cleaveland02}.

    Recently, L{\"u}ttgen and Vogler propose the notion of Logic LTS, which combines operational and logical styles of specification in one unified framework \cite{Luttgen07,Luttgen10}.
    Roughly speaking, a Logic LTS is a labelled transition system with an inconsistency predicate on states.
    A few of constructors over Logic LTSs are introduced, which include operational constructors, such as CSP-style parallel composition and hiding, and logic constructors conjunction and disjunction.
    This framework allows one to freely mix operational operators and logic operators, while most early heterogenous specifications couple them loosely and do not allow for mixed specification \cite{Boudol92,Dam}.
    Moreover, the drawbacks in \cite{Cleaveland00,Cleaveland02} mentioned above have been remedied by using ready-tree semantics \cite{Luttgen07}.
    In order to support compositional reasoning when introducing the parallel constructor over Logic LTSs, a kind of modified ready simulation is adopted to describe the refinement relation \cite{Luttgen10}.
    Some standard temporal logic operators, such as \textquotedblleft always\textquotedblright and \textquotedblleft unless\textquotedblright, are also integrated into this framework \cite{Luttgen11}.

    \begin{figure}
    \begin{center}
        \includegraphics[scale=.6]{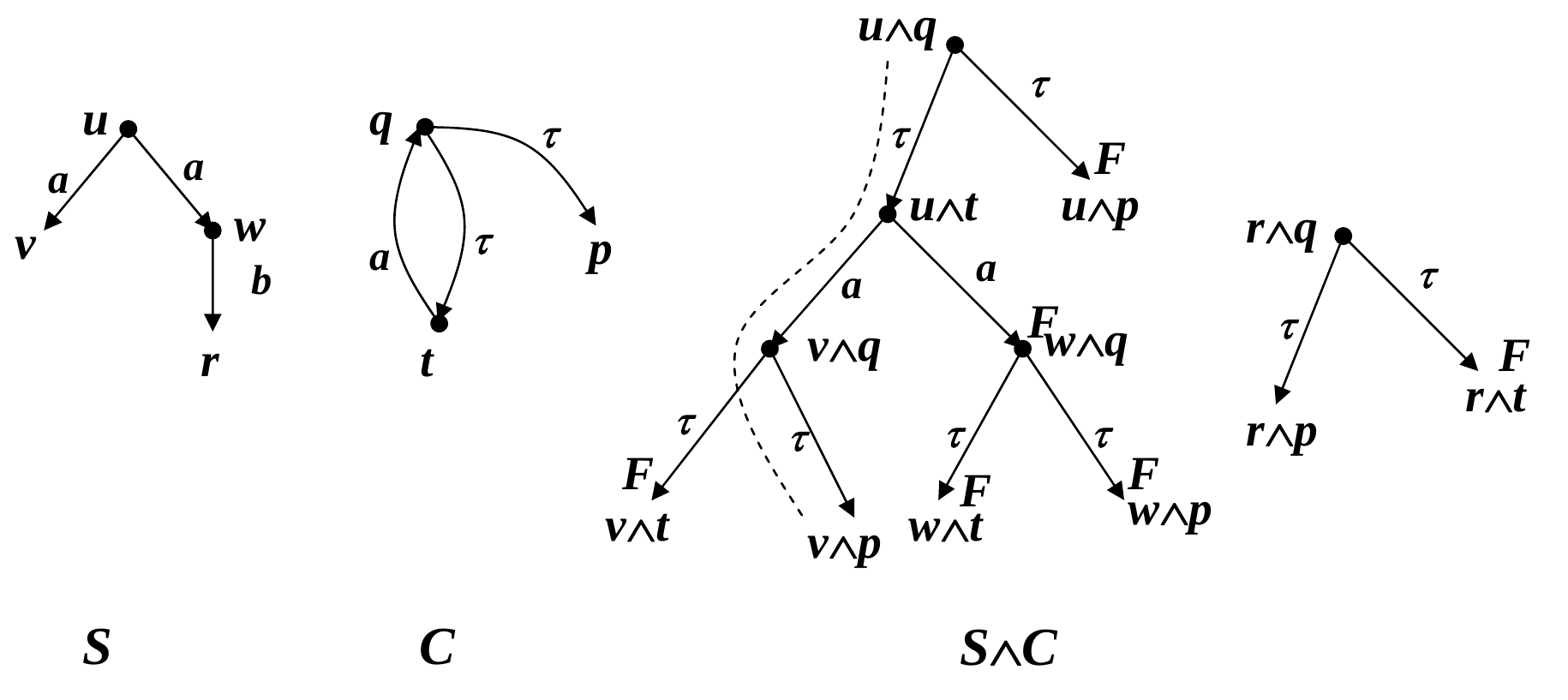}
    \end{center}
    \caption{Conjunction}\label{F:CONJUNCTION}
    \end{figure}

    Conjunction is a distinctive constructor in the framework of Logic LTSs, we give a simple example to illustrate how it works in heterogenous system design. In Figure~\ref{F:CONJUNCTION}, $S$ represents a specification which involves actions $a$ and $b$, and $C$ represents an additional constraint which requires that the action $b$ cannot occur in the implementation.
    Then, the dashed line in the Logic LTS $S \wedge C$ represents a specification which refines $S$ and satisfies the constraint $C$.
    Briefly speaking, if a state in $S$ has $b$-derivative, then the constructor conjunction turns this state into inconsistent state, i.e., unimplementable state.
    Thus, such state cannot be reached at run-time.
    The reader may refer to \cite{Luttgen10} for the concrete construction of $S \wedge C$.

    Up to now, Logic LTSs have been explored deeply, however, term-based way has not yet been involved, and the axiomatization of constructors is absent.
    This paper intends to develop L\"{u}ttgen and Vogler's work along the direction suggested by them in \cite{Luttgen10}.
    We present a process calculus for Logic LTSs (CLL, for short). In addition to prefix $\alpha.()$, external choice $\Box$ and parallel operator $\parallel_A$, CLL contains logical operators $\wedge$ and $\vee$ over process terms, which correspond to the constructors conjunction and disjunction over Logic LTSs respectively.
    Follows \cite{Luttgen10}, a variant of the usual notion of ready simulation is adopted to formalize the refinement relation.
    Moreover, an axiomatic system $AX_{CLL}$ is provided to characterize the operators in CLL in terms of (in)equational laws, and the soundness and ground-completeness w.r.t ready simulation are established.

    The remainder of this paper is organized as follows. The next section recalls some related notions.
    Section~3 introduces SOS rules of CLL.
    In section~4, based on reducing technique, the existence and uniqueness of stable transition model for CLL is demonstrated, and a few of  basic properties of the LTS associated with CLL are explored.
    Section~5  develops a behavioral theory of CLL.
    Section~6 provides an axiomatic system for CLL, and the soundness and ground-completeness are showed.
    Finally, a brief conclusion and discussion are given in Section~7.

\section{Preliminaries}

\subsection{Logic LTS}

In this subsection, we introduce some useful notations and recall the definition of Logic LTS.

Let $Act$ be the set of visible action names ranged over by letters $a$, $b$, etc., and let $Act_{\tau}$ denote $Act \cup \{\tau\}$ ranged over by $\alpha$ and $\beta$, where $\tau$ represents invisible actions.
A labelled transition system with inconsistency predicate is a quadruple $(P,Act_{\tau},\longrightarrow,F)$, where $P$ is a set of states, $\longrightarrow \subseteq P\times Act_{\tau}\times P$ is the transition relation and $F\subseteq P$ is the set of all inconsistent states.

As usual, we write $p \stackrel{\alpha}{\longrightarrow} q$ if $(p,\alpha,q)\in \longrightarrow$.
$q$ is called an $\alpha$-derivative of $p$ if $p \stackrel{\alpha}{\longrightarrow}q$.
We write $p \stackrel{\alpha}{\longrightarrow}$ (or, $p \not \stackrel{\alpha}{\longrightarrow}$) if $\exists q\in P.p\stackrel{\alpha}{\longrightarrow}q$ ($\nexists q\in P.p  \stackrel{\alpha}{\longrightarrow}q$, respectively).
$\mathcal{I}(p)$ stands for the ready set $\{\alpha \in Act_{\tau}|p \stackrel{\alpha}{\longrightarrow}\}$ of $p$.
A state $p$ is said to be stable if it cannot engage in any $\tau$-transition, i.e., $p \not\stackrel{\tau}{\longrightarrow}$.
Some useful decorated transition relations are listed below.

$p \stackrel{\alpha}{\longrightarrow}_F q$ iff $p \stackrel{\alpha}{\longrightarrow} q$ and $p,q\notin F$.

$p \stackrel{\epsilon}{\Longrightarrow}q$ iff $p (\stackrel{\tau}{\longrightarrow})^* q$, where $(\stackrel{\tau}{\longrightarrow})^* $ is the transitive reflexive closure of $\stackrel{\tau}{\longrightarrow}$.

$p \stackrel{\alpha}{\Longrightarrow}q$ iff $\exists r,s\in P.p \stackrel{\epsilon}{\Longrightarrow} r \stackrel{\alpha}{\longrightarrow}s \stackrel{\epsilon}{\Longrightarrow} q$.

$p \stackrel{\epsilon}{\Longrightarrow}_F q$ (or, $p \stackrel{\alpha}{\Longrightarrow}_F q$) iff  $p \stackrel{\epsilon}{\Longrightarrow}q$ ($p \stackrel{\alpha}{\Longrightarrow}q$, respectively) and all states along the sequence, including $p$ and $q$, are not in $F$.

$p \stackrel{\epsilon}{\Longrightarrow}_F|q$ (or, $p \stackrel{\alpha}{\Longrightarrow}_F|q$) iff $p \stackrel{\epsilon}{\Longrightarrow}_F q$ ($p \stackrel{\alpha}{\Longrightarrow}_F q$, respectively) and $q$ is stable.

\begin{mydefn}[Logic LTS \cite{Luttgen10}]\label{D:LLTS}
  An LTS $(P,Act_{\tau},\longrightarrow,F)$ is said to be a Logic LTS if, for each $p \in P$,

\noindent\textbf{(LTS1) }$p \in F$ if $\exists\alpha\in \mathcal{I}(p)\forall q\in P(p \stackrel{\alpha}{\longrightarrow}q \;\text{implies}\; q\in F)$,

\noindent\textbf{(LTS2)} $p \in F$ if $\nexists q\in P.p \stackrel{\epsilon}{\Longrightarrow}_F|q$.
\end{mydefn}

The condition (LTS1) formalizes the backward propagation of inconsistencies, and (LTS2) captures the intuition that divergence (i.e., infinite sequences of $\tau$-transitions) should be viewed as catastrophic.
For more motivation behind (LTS1) and (LTS2), the reader may refer to \cite{Luttgen10}.

\begin{mydefn}[$\tau$-pure  \cite{Luttgen10}]
An LTS $(P,Act_{\tau},\longrightarrow,F)$ is said to be {$\tau$}-pure if, for each $p \in P$, $p\stackrel{\tau}{\longrightarrow}$ implies $\nexists a\in Act.\;p\stackrel{a}{\longrightarrow}$.
\end{mydefn}

Hence, for any state $s$ in a $\tau$-pure LTS, either ${\mathcal I}(s)=\{\tau\}$ or ${\mathcal I}(s)\subseteq Act$.
Following \cite{Luttgen10}, this paper will focus on $\tau$-pure Logic LTSs.

\subsection{Transition System Specification}
Structural Operational Semantics (SOS) is proposed by G.~Plotkin in \cite{Plotkin81}, which is a logical method of giving operational semantics.
The basic idea behind SOS is to describe the behavior of a program in terms of the behavior of its components.
Thus, SOS provides a syntax oriented view on operational semantics.
Transition System Specifications (TSS's), as presented by Groote and Vaandrager in \cite{Groote92}, are formalizations of SOS.
This subsection recalls basic concepts related to TSS.
Further information on this issue may be found in \cite{Aceto01,Bol96,Groote92}.

 Let $V$ be an infinite set of variables and $\Sigma$  a signature. The set of $\Sigma $-terms over $V$, denoted by $T(\Sigma ,V)$, is the least set such that (I) $V \subseteq T(\Sigma ,V)$ and (II) if $f\in \Sigma$ and $t_1,\dots,t_n \in T(\Sigma ,V)$, then $f(t_1,\dots,t_n) \in T(\Sigma,V)$, where $n$ is the arity of $f$. $\mathbb{T} (\Sigma )$ is used to abbreviate $T(\Sigma ,V)$, $T(\Sigma ,\emptyset )$ is abbreviated by $T(\Sigma )$, elements in $T(\Sigma )$ are called closed or ground terms.

    A substitution $\sigma $ is a mapping from $V$ to $\mathbb{T}(\Sigma )$.  As usual, a substitution $\sigma $ is lifted to a mapping $\mathbb{T} (\Sigma ) \rightarrow \mathbb{T} (\Sigma )$ by $\sigma (f(t_1,...,t_n))\triangleq f(\sigma(t_1),\dots ,\sigma (t_n))$ for n-arity $f\in \Sigma$ and $t_1,\dots ,t_n \in {\mathbb T}(\Sigma )$. A substitution is said to be closed if it maps all variables to ground terms.

    A TSS is a quadruple $\mathcal{P}=(\Sigma,A,\mathbb{P},R)$, where $\Sigma$ is a signature, $A$ is a set of labels, $\mathbb P$ is a set of predicate symbols and $R$ is a set of rules.
    Positive literals are all expressions of the form $t \stackrel{a}{\longrightarrow} t'$ or $tP$, while negative literals are all expressions of the form $t \not\stackrel{a}{\longrightarrow}$ or $t\neg P$, where $t,t'\in \mathbb{T}(\Sigma )$, $a\in A$ and $P\in {\mathbb P}$.
    A literal is a positive or negative literal, and $\varphi$, $\psi$, $\chi$ are used to range over literals.
    A literal is said to be ground or closed if all terms occurring in it are ground.
    A rule $r\in R$ has the form like $\frac{H}{C}$, where $H$, the premises of the rule $r$, denoted as $prem(r)$, is a set of literals, and $C$, the conclusion of the rule $r$, denoted as $conc(r)$, is a positive literal. Furthermore, we write $pprem(r)$ for the set of positive premises of $r$ and $nprem(r)$ for the set of negative premises of $r$. A rule $r$ is said to be positive if $nprem(r)= \emptyset$. A TSS is said to be  positive if it has only positive rules. A rule is said to be an axiom if its set of premises is empty. An axiom $\frac{\emptyset}{t \stackrel{a}{\longrightarrow} t'}$ is often written as $t \stackrel{a}{\longrightarrow} t'$.
    Given a substitution $\sigma$ and a rule $r \in R$, $\sigma(r)$ is the rule obtained from $r$ by replacing each variable in $r$ by its $\sigma$-image, that is, $\sigma(r) \triangleq \frac{\{\sigma(\varphi)|\varphi \in prem(r)\}}{\sigma(conc(r))}$.
    Moreover, if $\sigma$ is closed then $\sigma(r)$ is said to be a ground instance of $r$.

\begin{mydefn}[Proof in Positive TSS]
    Let $\mathcal{P}=(\Sigma,A,\mathbb{P},R)$ be a positive TSS. A proof of a closed positive literal $\psi$ from $\mathcal{P}$ is a well-founded, upwardly branching tree, whose nodes are labelled by closed literals, such that\\
    --- the root is labelled with $\psi$,\\
    --- if $\chi$ is the label of a node $q$ and $\{\chi_i:i\in I\}$ is the set of labels of the nodes directly above $q$, then there is a rule $\{\varphi_i:i \in I\} \slash \varphi$ in $R$ and a closed substitution $\sigma$ such that $\chi=\sigma(\varphi)$ and $\chi_i=\sigma(\varphi_i)$ for each $i \in I$.

    If a proof of $\psi$ from $\mathcal{P}$ exists, then $\psi$ is said to be provable from $\mathcal{P}$, in symbols $\mathcal{P}\vdash \psi$.
\end{mydefn}

\section{Syntax and SOS Rules of CLL}

The process terms in CLL are defined by BNF below:
\[ t::= 0\;|\perp\;|\;(\alpha.t) \;|\; (t\Box t)\;|\;(t\wedge t)\;|\;(t\vee t)\;|\;(t\parallel_A t) \] where $\alpha\in Act_\tau$ and  $A\subseteq Act $.
We denote $T(\Sigma_{CLL})$ as the set of all process terms. We shall always use $t_1 \equiv t_2$ to mean that the expressions $t_1$ and $t_2$ are syntactically identical.

As usual, 0 is a process that can do nothing.
The prefix $\alpha.t$ has a single capability, expressed by $\alpha$; the process $t$ cannot proceed until $\alpha$ has been exercised.
$\Box$ is an external choice operator.
$\parallel_A $ is a CSP-style parallel operator, $t_1\parallel_A t_2$ represents a process that behaves as $t_1$ in parallel with $t_2$ under the synchronization set $A$.
$\bot$ represents an inconsistent process with empty behavior.
$\vee$ and $\wedge$ are logical operators, which are intended for describing logical combinations of processes.

We now provide SOS rules to specify the behavior of process terms formally.
These rules reflect L\"{u}ttgen and Vogler's constructions in \cite{Luttgen10} in process algebraic style.
Unless noted, let $a \in Act$, $\alpha \in Act_{\tau}$ and $A \subseteq Act$.
All SOS rules are divided into two parts: Table~\ref{Ta:OPERATIONAL_RULES} consists of operational rules $Ra_i(1\leq i\leq 15)$, and Table~\ref{Ta:PREDICATIVE_RULES} contains predicate rules $Rp_i(1\leq i \leq 13)$.

\begin{table}[ht]
\begin{center}
\fbox{$
    \begin{array}{ll}
    \displaystyle (Ra_1)\frac{-}{\alpha.x_1 \stackrel{\alpha}{\longrightarrow} x_1}  &
    \displaystyle (Ra_2)\frac{x_1 \stackrel{a}{\longrightarrow} y_1, x_2 \not \stackrel{\tau}{\longrightarrow}}{x_1 \Box x_2 \stackrel{a}{\longrightarrow} y_1}\\
    \displaystyle (Ra_3)\frac{x_1 \not\stackrel{\tau}{\longrightarrow} , x_2 \stackrel{a}{\longrightarrow} y_2 }{x_1 \Box x_2 \stackrel{a}{\longrightarrow} y_2}&
    \displaystyle (Ra_4)\frac{x_1 \stackrel{\tau}{\longrightarrow} y_1}{x_1 \Box x_2 \stackrel{\tau}{\longrightarrow} y_1 \Box x_2}\\
    \displaystyle (Ra_5)\frac{x_2 \stackrel{\tau}{\longrightarrow} y_2}{x_1 \Box x_2 \stackrel{\tau}{\longrightarrow} x_1 \Box y_2}&
    \displaystyle (Ra_6)\frac{x_1 \stackrel{a}{\longrightarrow} y_1, x_2 \stackrel{a}{\longrightarrow}y_2}{x_1 \wedge x_2 \stackrel{a}{\longrightarrow} y_1 \wedge y_2}\\
    \displaystyle (Ra_7)\frac{x_1 \stackrel{\tau}{\longrightarrow} y_1}{x_1 \wedge x_2 \stackrel{\tau}{\longrightarrow} y_1 \wedge x_2} &
    \displaystyle (Ra_8)\frac{x_2 \stackrel{\tau}{\longrightarrow} y_2}{x_1 \wedge x_2 \stackrel{\tau}{\longrightarrow} x_1 \wedge y_2} \\
    \displaystyle (Ra_9)\frac{-}{x_1 \vee x_2 \stackrel{\tau}{\longrightarrow} x_1} &
    \displaystyle (Ra_{10})\frac{-}{x_1 \vee x_2 \stackrel{\tau}{\longrightarrow} x_2} \\
    \displaystyle (Ra_{11})\frac{x_1 \stackrel{\tau}{\longrightarrow} y_1}{x_1 \parallel_A x_2 \stackrel{\tau}{\longrightarrow} y_1\parallel_A x_2} &
    \displaystyle (Ra_{12})\frac{x_2 \stackrel{\tau}{\longrightarrow} y_2}{x_1 \parallel_A x_2 \stackrel{\tau}{\longrightarrow} x_1 \parallel_A y_2} \\
    \displaystyle (Ra_{13})\frac{x_1 \stackrel{a}{\longrightarrow} y_1 , x_2 \not \stackrel{\tau}{\longrightarrow} , a\notin A}{x_1 \parallel_A x_2 \stackrel{a}{\longrightarrow} y_1 \parallel_A x_2}&
    \displaystyle (Ra_{14})\frac{x_1 \not\stackrel{\tau}{\longrightarrow} , x_2 \stackrel{a}{\longrightarrow} y_2 , a\notin A}{x_1 \parallel_A x_2 \stackrel{a}{\longrightarrow} x_1 \parallel_A y_2}\\
    \displaystyle (Ra_{15})\frac{x_1 \stackrel{a}{\longrightarrow} y_1, x_2 \stackrel{a}{\longrightarrow}y_2, a \in A}{x_1\parallel_A x_2 \stackrel{a}{\longrightarrow} y_1 \parallel y_2}& \;
    \end{array}
$}
\caption{Operational Rules\label{Ta:OPERATIONAL_RULES}}
\end{center}
\end{table}

Negative premises in rules $Ra_2$, $Ra_3$, $Ra_{13}$ and $Ra_{14}$ give $\tau$-transition precedence over visible transitions, which guarantees that the transition model of CLL is $\tau$-pure (see, Section~4).
Rules $Ra_9$ and $Ra_{10}$ illustrate that the operational aspect of $t_1\vee t_2$ is same as internal choice in usual process calculus.
The rule $Ra_6$ reflects that conjunction operator is a synchronous product for visible transitions.
The rules about other operators are usual.

\begin{table}[ht]
\begin{center}
\fbox{$
    \begin{array}{ll}
    \displaystyle (Rp_1)\frac{-}{\bot F}&
    \displaystyle (Rp_2)\frac{x_1 F}{\alpha .x_1 F}\\
    \displaystyle (Rp_3)\frac{x_1 F, x_2 F}{x_1\vee x_2 F}&
    \displaystyle (Rp_4)\frac{x_1 F}{x_1\Box x_2 F}\\
    \displaystyle (Rp_5)\frac{x_2 F}{x_1\Box x_2 F}&
    \displaystyle (Rp_6)\frac{x_1 F}{x_1\parallel_A x_2 F}\\
    \displaystyle (Rp_7)\frac{x_2 F}{x_1\parallel_A x_2 F}&
    \displaystyle (Rp_8)\frac{x_1 F}{x_1\wedge x_2 F}\\
    \displaystyle (Rp_9)\frac{x_2 F}{x_1\wedge x_2 F}&
    \displaystyle (Rp_{10})\frac{x_1 \stackrel{a}{\longrightarrow} y_1, x_2 \not\stackrel{a}{\longrightarrow}, x_1 \wedge x_2 \not\stackrel{\tau}{\longrightarrow}}{x_1 \wedge x_2 F}\\
    \displaystyle (Rp_{11})\frac{x_1 \not\stackrel{a}{\longrightarrow} , x_2 \stackrel{a}{\longrightarrow} y_2, x_1 \wedge x_2 \not\stackrel{\tau}{\longrightarrow}}{x_1 \wedge x_2 F}&
    \displaystyle (Rp_{12})\frac{x_1 \wedge x_2 \stackrel{\alpha}{\longrightarrow} y, y \neg F}{x_1 \wedge x_2 \overline{F}_{\alpha}} \\
    \displaystyle (Rp_{13})\frac{x_1 \wedge x_2 \stackrel{\alpha}{\longrightarrow} z, x_1 \wedge x_2  \neg \overline{F}_{\alpha}}{x_1 \wedge x_2 F} &\;
    \end{array}
$}
\caption{Predicate Rules}\label{Ta:PREDICATIVE_RULES}
\end{center}
\end{table}

Predicate rules in Table~\ref{Ta:PREDICATIVE_RULES} specify the properties of the predicate $F$.
In CLL, although both $0$ and $\bot$ have empty behavior, they represent different processes.
The rule $Rp_1$ says that $\bot$ is inconsistent.
Thus, the process $\bot$ cannot be implemented.
While $0$ is consistent(see Lemma~\ref{L:F_NORMAL}(5)), which is an implementable process.
The rule $Rp_3$ reflects that if both two disjunctive parts are inconsistent then so is the disjunction.
Rules $Rp_4-Rp_9$ describe the system design strategy that if one part is inconsistent, then so is the whole composition.
The rules $Rp_{10}$ and $Rp_{11}$ reveal that the conjunction is inconsistent if its conjuncts have different ready sets.
The rule below formalizes (LTS1) in Def.~\ref{D:LLTS} for process terms of the form $x_1 \wedge x_2$.
\[ \frac{x_1 \wedge x_2 \stackrel{\alpha}{\longrightarrow}, \forall y(x_1\wedge x_2 \stackrel{\alpha}{\longrightarrow}y \Rightarrow yF)}{x_1 \wedge x_2F}\]
However, the universal quantifier occurs in it explicitly.
We adopt the method presented in \cite{Verhoef95} to eliminate universal quantifier.
For this purpose, we introduce auxiliary predicate symbol $\overline{F}_{\alpha}$ for each $\alpha \in Act_{\tau}$ and rules $Rp_{12}$ and $Rp_{13}$.
Intuitively, by the rule $Rp_{12}$, $t_1 \wedge t_2 \overline{F}_{\alpha}$ states that $t_1 \wedge t_2$ has a consistent $\alpha$-derivative.
$Rp_{12}$ and $Rp_{13}$ together say that a conjunction is inconsistent if it can engage in $\alpha$-transition and all its $\alpha$-derivatives are inconsistent. A formal result concerning this will be given in Section~4 (see, Lemma~\ref{L:LLTS_I}).

Summarily, the TSS for CLL is $\mathcal{P}_{CLL} =(\Sigma_{CLL},Act_{\tau},\mathbb{P}_{CLL},R_{CLL})$, where ---$\Sigma_{CLL} = \{ \Box,\wedge,\vee,0,\bot\}\cup\{\alpha.()|\alpha \in Act_\tau\}\cup\{\parallel_A|A\subseteq Act\}$,\\
---$\mathbb{P}_{CLL}=\{F\} \cup \{\overline{F}_\alpha|\alpha \in Act_\tau\}$, and \\
---$R_{CLL}=\{Ra_1, \dots, Ra_{15}\}\cup\{Rp_1, \dots ,Rp_{13}\}$.

\section{Operational Semantics of CLL}
A natural and simple method of describing the operational nature of processes is in terms of LTSs. Given a TSS, an important problem is how to associate LTSs with process terms. For positive TSS, the answer is straightforward.
However, this problem is non-trivial for TSS containing negative premises. This section will consider the well-definedness of the TSS ${\mathcal P}_{CLL}$ (i.e., existence and uniqueness of the LTS induced by ${\mathcal P}_{CLL}$), and explore basic properties of the induced transition model.
\subsection{Basic Notations}

Let $A$ be a set of labels and $\mathbb{P}$   a set of predicate symbols. A transition model $M$ is a subset of $Tr(\Sigma ,A)\cup Pred(\Sigma ,\mathbb{P} )$, where $Tr(\Sigma ,A)=T(\Sigma)\times A\times T(\Sigma )$ and $Pred(\Sigma ,\mathbb{P} )=T(\Sigma )\times \mathbb{P}$, elements $(t,a,t^{\prime })$ and $(t,P)$ in $M$ are written as $t \stackrel{a}{\longrightarrow} t'$ and $tP$ respectively.

     A positive closed literal $\psi$ holds in $M$ or $\psi$ is valid in $M$, in symbols $M\models \psi$, if $\psi \in M$. A negative closed literal  $t \not \stackrel{a}{\longrightarrow}$ (or, $t \neg P$) holds in $M$, in symbols $M \models t \not \stackrel{a}{\longrightarrow}$ ($M \models t \neg P$, respectively), if there is no $t'$ such that $t \stackrel{a}{\longrightarrow} t' \in M$($tP \notin M$, respectively). For a set of closed literals $\Psi$, we write $M \models \Psi$ iff $M \models \psi$ for each $\psi \in \Psi$.

\begin{mydefn}
Let $\mathcal{P}=(\Sigma,A,\mathbb{P},R)$ be a TSS and $M$ a transition model.
    $M$ is said to be a model of $\mathcal{P}$ if, for each $r\in R$ and $\sigma:V \longrightarrow T(\Sigma)$ such that $M \models prem(\sigma(r))$, we have $M \models conc(\sigma(r))$.
    $M$ is said to be supported by $\mathcal{P}$ if, for each $\psi \in M$, there exists $r \in R$ and $\sigma:V \longrightarrow T(\Sigma)$ such that $M \models prem(\sigma(r))$ and $conc(\sigma(r))=\psi$.
    $M$ is said to be a supported model of $\mathcal{P}$ if $M$ is supported by $\mathcal{P}$ and $M$ is a model of $\mathcal{P}$.
\end{mydefn}

It is well known that every positive TSS $\mathcal P$ has a least transition model, which exactly consists of provable transitions of $\mathcal P$ \cite{Bol96}.
However, since it is not immediately clear what can be considered a \textquotedblleft proof\textquotedblright for a negative formula, it is much less trivial to associate an LTS with a TSS containing negative premises \cite{Groote93}.
The first generic answer to this question is formulated in \cite{Bloom95,Groote93}, in which the above notion of supported model is introduced.
However, this notion doesn't always work well. Several alternatives have been proposed, and a good overview on this issue is provided in \cite{Glabbeek04}.
In the following, we recall the notions of stratification and stable transition model, which play an important role in this field.

\begin{mydefn}[Stratification \cite{Aceto01,Bol96}]\label{D:STRATIFICATION}
     Let $\mathcal{P}=(\Sigma,A,\mathbb{P},R)$ be a TSS and $\alpha$ an ordinal number. A function $S:Tr(\Sigma,A)\cup Pred(\Sigma, \mathbb{P}) \longrightarrow \alpha$ is said to be a stratification of $\mathcal{P}$ if, for every rule $r \in R$ and every substitution $\sigma :V \longrightarrow T(\Sigma)$, the following conditions hold.\\
    (1) $S(\psi)\leq S(conc( \sigma(r)))$ for each $\psi \in pprem(\sigma (r))$,\\
    (2) $S(tP)<S(conc(\sigma(r)))$ for each $t \neg P \in nprem(\sigma(r))$, and\\
    (3) $S(t \stackrel{a}{\longrightarrow} t') < S(conc(\sigma(r)))$ for each $t' \in T(\Sigma)$ and $t \not\stackrel{a}{\longrightarrow} \in nprem(\sigma(r))$.

    A TSS is said to be stratified iff there exists a stratification function for it.
\end{mydefn}

In the following, given a set $R$ of rules, we denote $R_{ground}$ the set of all ground instances of rules in $R$. Similarly, given a rule $r$, $r_{ground}$ is the set of all ground instances of it.

\begin{mydefn}[Stable Transition Model  \cite{Bol96,Gelfond88}]\label{D:STABLE}
    Let $\mathcal{P}=(\Sigma,A,\mathbb{P},R)$ be a TSS and $M$ a transition model. $M$ is said to be a stable transition model for $\mathcal P$ if
    \[M=M_{Strip({\mathcal P},M)},\]
    where
    $Strip({\mathcal P},M)$ is the TSS $(\Sigma,A,\mathbb{P},Strip(R,M))$ with
    \[Strip(R,M)\triangleq\left\{\frac{pprem(r)}{conc(r)}|\quad r\in R_{ground}\;\text{and}\;M\models nprem(r)\right\},
    \]
    and $M_{Strip({\mathcal P},M)}$ is the least transition model of the positive TSS $Strip({\mathcal P},M)$.
\end{mydefn}

The above notion provides a reasonable semantics for TSS's with negative premises.
TSS's that do not have a unique stable model should be ruled out and considered pathological \cite{Bol96}.
As is well known, stable models are supported models, and each stratified TSS $\mathcal P$ has a unique stable model \cite{Bol96}, moreover, such stable model does not depend on particular stratification function \cite{Groote93}. Unfortunately, the TSS ${\mathcal P}_{CLL}$ in Section~3 cannot be stratified.

\begin{observation}
 ${\mathcal P}_{CLL}$ cannot be stratified.
\end{observation}
\begin{proof}
    Assume that $S$ is a stratification of ${\mathcal P}_{CLL}$. Let $t_1,t_2 \in T(\Sigma_{CLL})$. By Def.~\ref{D:STRATIFICATION}, the rule $Rp_{13}$ entails $S(t_1\wedge t_2\overline{F}_{\alpha})<S(t_1\wedge t_2 F)$. However, the rule $Rp_{12}$ requires that $S(t_1\wedge t_2\overline{F}_{\alpha})>S(t_1\wedge t_2 F)$ holds. Hence, the stratification of ${\mathcal P}_{CLL}$ does not exist.
\end{proof}

In \cite{Bol96}, the notion of positive after reduction (also called complete in \cite{Glabbeek04}) is introduced as a criterion for well-definedness of the semantics and is shown to be more general than stratification.
We shall adopt such reducing technique to obtain a stable transition model of ${\mathcal P}_{CLL}$.
Related concepts and results are recalled below.

\begin{mydefn}\label{D:REDUCTION}
    Let $\mathcal{P}=(\Sigma,A,\mathbb{P},R)$ be a TSS and $M_{true}$, $M_{pos}$ be transition models. The reduction TSS of $\mathcal P$ w.r.t $M_{true}$ and $M_{pos}$ is defined as
    \[Reduce({\mathcal P},M_{true},M_{pos})\triangleq(\Sigma,A,{\mathbb P},Reduce(R,M_{true},M_{pos})),\]
    where $Reduce(R,M_{true},M_{pos})$ is the set of all rules $r'$ such that there exists $r\in R_{ground}$ satisfying
    \begin{enumerate}
      \item $M_{true}\models nprem(r)$,
      \item $M_{pos}\models pprem(r)$, and
      \item $r'=\frac{\{\psi \in pprem(r)|M_{true}\not\models\psi\}\cup\{\psi\in nprem(r)| M_{pos} \not \models \psi\}} {conc(r)}$.
    \end{enumerate}
    In such case, we will say that $r'$ originates from $r$. Clearly,
    $Reduce({\mathcal P},M_{true},M_{pos})$ is ground, that is, its all rules are ground.
\end{mydefn}

\begin{mydefn}
    Let $\mathcal{P}=(\Sigma,A,\mathbb{P},R)$ be a ground TSS.\\
    ---$True({\mathcal P})\triangleq(\Sigma,A,\mathbb{P},True(R))$, where $True(R)\triangleq\{r\in R|nprem(r)=\emptyset\}$.\\
    ---$Pos({\mathcal P})\triangleq(\Sigma,A,\mathbb{P},Pos(R))$, where
    \[Pos(R)\triangleq\left\{\frac{pprem(r)}{conc(r)}| r\in R\right\}.
    \]
\end{mydefn}

\begin{mydefn}
    Let $\mathcal{P}=(\Sigma,A,\mathbb{P},R)$ be a TSS. For every ordinal number $\alpha$, the $\alpha$-reduction of $\mathcal P$, in symbols $Red^{\alpha}(\mathcal P)$, is recursively defined below\\
    ---$Red^0({\mathcal P})\triangleq(\Sigma,A,\mathbb{P},R_{ground})$, \\
    ---$Red^{\alpha}({\mathcal P})\triangleq Reduce({\mathcal P},\bigcup_{\beta<\alpha}M_{True(Red^{\beta}({\mathcal P}))},\bigcap_{\beta<\alpha}M_{Pos(Red^{\beta}({\mathcal P}))})$.
\end{mydefn}
In the above, $M_{True(Red^{\beta}({\mathcal P}))}$ (or, $M_{Pos(Red^{\beta}({\mathcal P}))}$) is the least transition
    model of the positive TSS $True(Red^{\beta}({\mathcal P}))$ ($Pos(Red^{\beta}({\mathcal P}))$, respectively).
A useful result about reducing technique is cited below.

\begin{theorem}\label{C:stable}
    Let $\mathcal{P}=(\Sigma,A,\mathbb{P},R)$ be a TSS. Suppose that $S:Tr(\Sigma,A)\cup Pred(\Sigma, \mathbb{P}) \longrightarrow \alpha$ is a stratification of $Red^{\beta}({\mathcal P})$ for some ordinal $\alpha$ and $\beta$. Then, the stable transition model of $Red^{\beta}({\mathcal P})$ is the unique stable transition model of $\mathcal P$.
\end{theorem}
\begin{proof}
  See Theorem 6.1 and Corollary 6.5 in \cite{Bol96}.
\end{proof}

\subsection{Stable Transition Model of ${\mathcal P}_{CLL}$}

    This subsection will adopt the reducing technique recalled above to show that ${\mathcal P}_{CLL}$ has a unique stable model. It turns out that 1-reduction is enough to realize our aim.

    Firstly, we reduce ${\mathcal P}_{CLL}$. By the definition of $\alpha$-reduction, we have
    \[Red^0({\mathcal P}_{CLL})=(\Sigma_{CLL},Act_{\tau},{\mathbb P}_{CLL},{R_{CLL}}_{ground})\]
    and
    \[Red^1({\mathcal P}_{CLL})=(\Sigma_{CLL},Act_{\tau},{\mathbb P}_{CLL},Reduce(R_{CLL},M_T,M_P)),\]
    where

\noindent    \textbf{(1)} $M_T$ and $M_P$ are the least models of the positive TSS's  $True(Red^0({\mathcal P}_{CLL}))$ and $Pos(Red^0({\mathcal P}_{CLL}))$ below respectively,
\[True(Red^0({\mathcal P}_{CLL}))= \underset{i\in I_1}\bigcup {Ra_i}_{ground} \cup \underset{i\in I_2}\bigcup  {Rp_i}_{ground},\]
where $I_1 = \{1,4,5,6,7,8,9,10,11,12,15\}$ and $I_2 = \{1,2,3,4,5,6,7,8,9\}$, and
\begin{multline*}
    Pos(Red^0({\mathcal P}_{CLL}))= \\True(Red^0({\mathcal P}_{CLL}))\cup
    \underset{i \in K_1}\bigcup \left(\frac{pprem(Ra_i)}{conc(Ra_i)}\right)_{ground}\cup
        \underset{i \in K_2}\bigcup \left(\frac{pprem(Rp_i)}{conc(Rp_i)}\right)_{ground},
\end{multline*}
where $K_1 =\{ 2,3,13,14\}$ and $K_2 = \{ 10,11,12,13\}$.

\noindent  \textbf{(2)} $Reduce(R_{CLL},M_T,M_P)$ is the set of all rules $r'$ such that there exists $r\in {R_{CLL}}_{ground}$ satisfying
    \begin{enumerate}
      \item [2.1.] $M_T\models nprem(r)$,
      \item [2.2.] $M_P\models pprem(r)$, and
      \item [2.3.] $r'=\frac{\{\psi \in pprem(r)|M_T\not\models\psi\}\cup\{\psi\in nprem(r)|M_P\not\models\psi\}}{conc(r)}$.
    \end{enumerate}

Next, we will show that $Red^1({\mathcal P}_{CLL})$ is stratified.
For this purpose, a few preliminary definitions are needed.
As usual, the degree of process terms is defined below.
\begin{enumerate}
  \item $|0|=|\bot|\triangleq1$,
  \item $|t_1 \odot t_2|\triangleq|t_1|+|t_2|+1$ for each $\odot \in \{\wedge,\Box,\vee,\parallel_A\}$,
  \item $|\alpha .t|\triangleq|t|+1$ for each $\alpha \in Act_{\tau}$.
\end{enumerate}

Further, the function $S_{{\mathcal P}_{CLL}}$ from  $Tr(\Sigma_{CLL},Act_{\tau})\;\cup\; Pred(\Sigma_{CLL},{\mathbb P}_{CLL})$ to $\mathbb{N}$ is defined as
\begin{enumerate}
  \item $S_{{\mathcal P}_{CLL}}(t\stackrel{\alpha}{\longrightarrow} t')\triangleq|t|$,
  \item $S_{{\mathcal P}_{CLL}}(tF)\triangleq|t|\times 2+1$, and
  \item $S_{{\mathcal P}_{CLL}}(t{\overline F}_{\alpha})\triangleq|t|\times 2$ for each $\alpha \in Act_{\tau}$.
\end{enumerate}

In order to show that $S_{{\mathcal P}_{CLL}}$ is a stratification of $Red^1({\mathcal P}_{CLL})$, we shall first give the result below.

\begin{lemma}\label{L:POS}
    If $t\stackrel{\alpha}{\longrightarrow} t'\in M_P$ then $|t|>|t'|$.
\end{lemma}
\begin{proof}
Let $t\stackrel{\alpha}{\longrightarrow} t' \in M_P$.
 We proceed by induction on the depth of the proof tree of $Pos(Red^0({\mathcal P}_{CLL})) \vdash t\stackrel{\alpha}{\longrightarrow} t'$.

For the induction basis, the only rule applied in the proof tree is an axiom.
Clearly, such axiom originates from $Ra_1$, $Ra_9$ or $Ra_{10}$.
It is a simple matter to check one by one.
For instance, if $t\stackrel{\alpha}{\longrightarrow} t'$ is inferred by the ground rule $\frac{-}{\beta.t_1\stackrel{\beta}{\longrightarrow}t_1}$in ${Ra_1}_{ground}$, then $t \equiv \beta.t_1$, $\alpha = \beta$ and $t' \equiv t_1$. Obviously, $|\alpha.t_1|>|t_1|$.

The induction step proceeds by distinguishing fifteen cases based on the format of the last rule applied in the proof tree.
It is routine, and we deal with two cases as examples and leave the remainder to the reader.\\

\noindent Case 1 $\frac{t_1 \stackrel{a}{\longrightarrow} t_1', t_2 \stackrel{a}{\longrightarrow}t_2'}{t_1 \wedge t_2 \stackrel{a}{\longrightarrow} t_1' \wedge t_2'}$.

Hence, we get $t \equiv t_1 \wedge t_2$, $\alpha = a \in Act$, $t' \equiv t_1'\wedge t_2'$, $t_1\stackrel{a}{\longrightarrow} t_1' \in M_P$ and $t_2\stackrel{a}{\longrightarrow} t_2' \in M_P$. By induction hypothesis (IH, for short), $|t_1|>|t_1'|$ and $|t_2|>|t_2'|$.
Then, it immediately follows that $|t_1 \wedge t_2|>|t_1' \wedge t_2'|$.\\

\noindent Case 2 $\frac{t_1 \stackrel{a}{\longrightarrow} t_1'}{t_1 \square t_2 \stackrel{a}{\longrightarrow} t_1'}$.

Clearly, this rule originates from one in ${Ra_2'}_{ground}$, where $Ra_2'=\frac{x_1 \stackrel{a}{\longrightarrow} y_1}{x_1 \Box x_2 \stackrel{a}{\longrightarrow} y_1}$.
Thus, we have $\alpha = a \in Act$, $t \equiv t_1 \Box t_2$, $t' \equiv t_1'$ and $t_1\stackrel{a}{\longrightarrow} t_1' \in M_P$.
By IH, we get $|t_1|>|t_1'|$.
So, $|t_1 \Box t_2|>|t_1'|$.
\end{proof}

\begin{lemma}\label{L:Stratification}
    $S_{{\mathcal P}_{CLL}}$ is a stratification of ${Red^1}({\mathcal P}_{CLL})$.
\end{lemma}
\begin{proof}

    We want to prove that, for each rule in $ Reduce({R_{CLL}}, M_T, M_P)$, $S_{{\mathcal P}_{CLL}}$ satisfies the items (1), (2) and (3) in Def.~\ref{D:STRATIFICATION}.
    Let $r \in Reduce({R_{CLL}}, M_T, M_P)$.
    We distinguish cases based on the origin of $r$. Here, we only consider three cases as illustrations, the remainder is similar and omitted.\\

\noindent Case 1 $r$ originates from a rule in ${Ra_2}_{ground}$.

    Then, $prem(r) \subseteq \{t_1 \stackrel{a}{\longrightarrow} t, t_2 \not \stackrel{\tau}{\longrightarrow}\}$ and $conc(r)=t_1 \Box t_2 \stackrel{a}{\longrightarrow} t$ for some $t_1,t_2,t \in T(\Sigma_{CLL})$ and $a \in Act$.
    Clearly, it suffices to consider the case where $prem(r) = \{t_1 \stackrel{a}{\longrightarrow} t, t_2 \not \stackrel{\tau}{\longrightarrow}\}$.
    It immediately follows from the definition of $S_{{\mathcal P}_{CLL}}$ that $ S_{{\mathcal P}_{CLL}}(t_1 \Box t_2 \stackrel{a}{\longrightarrow} t)> S_{{\mathcal P}_{CLL}}(t_1 \stackrel{a}{\longrightarrow} t )$ and, for any $s \in T(\Sigma_{CLL})$, $ S_{{\mathcal P}_{CLL}}(t_1 \Box t_2 \stackrel{a}{\longrightarrow} t)> S_{{\mathcal P}_{CLL}}(t_2 \stackrel{\tau}{\longrightarrow} s )$, as desired.\\

\noindent Case 2 $r$ originates from a rule in ${Rp_{12}}_{ground}$.

    Similar to Case 1, it is enough to deal with the case where $r$ has the form $\frac{t_1 \wedge t_2 \stackrel{\alpha}{\longrightarrow} t_3, t_3 \neg F}{t_1 \wedge t_2 \overline{F}_{\alpha}}$ with $M_P\models t_1 \wedge t_2 \stackrel{\alpha}{\longrightarrow} t_3$.
    By Lemma~\ref{L:POS}, $|t_1 \wedge t_2 |>|t_3|$.
    Then, by the definition of $S_{{\mathcal P}_{CLL}}$, we have $S_{{\mathcal P}_{CLL}}(t_1 \wedge t_2 \overline{F}_{\alpha}) = 2\times |t_1 \wedge t_2| > 2\times |t_3|+1 = S_{{\mathcal P}_{CLL}}(t_3F)$
    and $S_{{\mathcal P}_{CLL}}({t_1 \wedge t_2}{\overline{F}}_{\alpha}) > S_{{\mathcal P}_{CLL}}(t_1 \wedge t_2 \stackrel{\alpha}{\longrightarrow}t_3)$.   \\

\noindent Case 3 $r$ originates from a rule in ${Rp_{13}}_{ground}$.

    Similarly, we treat the rule with the form $\frac{t_1 \wedge t_2\stackrel{\alpha}{\longrightarrow} t_3, t_1 \wedge t_2 \neg \overline{F}_{\alpha}}{t_1 \wedge t_2 F}$.
    It immediately follows that $S_{{\mathcal P}_{CLL}}({t_1 \wedge t_2}F) > S_{{\mathcal P}_{CLL}}(t_1 \wedge t_2 \stackrel{\alpha}{\longrightarrow}t_3)$ and $S_{{\mathcal P}_{CLL}}({t_1 \wedge t_2}F) > S_{{\mathcal P}_{CLL}}(t_1 \wedge t_2 {\overline{F}}_{\alpha})$, as desired.
\end{proof}

We now arrive at the main result of this subsection.

\begin{theorem}
    ${\mathcal P}_{CLL}$ has a unique stable transition model.
\end{theorem}
\begin{proof}
Follows from Theorem~\ref{C:stable} and Lemma~\ref{L:Stratification}.
\end{proof}

\noindent\textbf{Notation}
Henceforward the unique stable transition model of ${\mathcal P}_{CLL}$ is denoted by $M_{CLL}$.  \\

The LTS associated with CLL is defined below.

\begin{mydefn}
    The LTS associated with CLL, in symbols $LTS(CLL)$, is the quadruple
    $(T(\Sigma_{CLL}),Act_{\tau},\longrightarrow_{CLL},F_{CLL})$, where \\
    --- $t \stackrel{\alpha}{\longrightarrow}_{CLL} t'$ iff $t\stackrel{\alpha}{\longrightarrow} t' \in M_{CLL}$,\\
    --- $t\in F_{CLL}$ iff $tF \in M_{CLL}$.
\end{mydefn}

Since $M_{CLL}$ is a stable transition model, which exactly consists of provable transitions of the positive TSS $Strip({\mathcal P}_{CLL},M_{CLL})$, the result below follows.

\begin{theorem}
    Let $t,t_1,t_2 \in T(\Sigma_{CLL})$ and $\alpha \in Act_{\tau}$.
  \begin{enumerate}
     \item The following are equivalent:

            1.1. $t_1 \stackrel{\alpha}{\longrightarrow}_{CLL}t_2$,

            1.2. $t_1 \stackrel{\alpha}{\longrightarrow}t_2 \in M_{CLL}$,

            1.3. $t_1 \stackrel{\alpha}{\longrightarrow}t_2 \in M_{Strip({\mathcal P}_{CLL},M_{CLL})}$,

            1.4. $Strip({\mathcal P}_{CLL},M_{CLL})\vdash t_1 \stackrel{\alpha}{\longrightarrow}t_2$.

     \item The following are equivalent:

            2.1. $t\in F_{CLL}$,

            2.2. $tF \in M_{CLL}$,

            2.3. $tF \in M_{Strip({\mathcal P}_{CLL},M_{CLL})}$,

            2.4. $Strip({\mathcal P}_{CLL},M_{CLL})\vdash tF$.

   \end{enumerate}
\end{theorem}
\begin{proof}
  Straightforward.
\end{proof}

The above theorem is trivial but useful. It provides a way to prove the properties of $LTS(CLL)$ and $M_{CLL}$. That is, we can demonstrate some conclusions by proceeding induction on the depth of proof trees in the positive TSS $Strip({\mathcal P}_{CLL},M_{CLL})$.  In the remainder of this paper, we will apply this theorem without any reference.

\subsection{Basic Properties of $LTS({CLL})$}

This subsection will provide a number of simple properties of $LTS({CLL})$. In particular, we will show that $LTS({CLL})$ is a Logic LTS and it is $\tau$-pure.
Some useful notations are listed below.

$t \stackrel{\alpha}{\longrightarrow}_F s$ iff $t \stackrel{\alpha}{\longrightarrow}_{CLL} s$ with $t,s\notin F_{CLL}$.

$t \stackrel{\epsilon}{\Longrightarrow}_{CLL}s$ iff $t (\stackrel{\tau}{\longrightarrow}_{CLL})^* s$.

$t \stackrel{\alpha}{\Longrightarrow}_{CLL}s $ iff $\exists r,p\in T(\Sigma_{CLL}).t \stackrel{\epsilon}{\Longrightarrow}_{CLL} r \stackrel{\alpha}{\longrightarrow}_{CLL}p \stackrel{\epsilon}{\Longrightarrow}_{CLL} s$.

$t \stackrel{\epsilon}{\Longrightarrow}_F s$ (or, $t \stackrel{\alpha}{\Longrightarrow}_F s$) iff  $t \stackrel{\epsilon}{\Longrightarrow}_{CLL}s$ ($t \stackrel{\alpha}{\Longrightarrow}_{CLL}s$, respectively) and all process terms  occurring in this sequence, including $t$ and $s$, are not in $F_{CLL}$.

$t \stackrel{\epsilon}{\Longrightarrow}_F|s$ (or, $t \stackrel{\alpha}{\Longrightarrow}_F|s$) iff $t \stackrel{\epsilon}{\Longrightarrow}_F s$ ($t \stackrel{\alpha}{\Longrightarrow}_F s$, respectively) and $s$ is stable.

A few simple properties of transition relation $\longrightarrow_{CLL}$ are listed in the next three lemmas, which will be frequently used in subsequent sections.

\begin{lemma}\label{L:Basic_I}
Let $t,t_1,t_2\in T(\Sigma_{CLL})$ and $\alpha,\beta \in Act_{\tau}$. Then
    \begin{enumerate}
      \item $\alpha.t \stackrel{\beta}{\longrightarrow}_{{CLL}} r$ iff $\alpha = \beta$ and $t \equiv r$. Hence, $\alpha.t \not \stackrel{\beta}{\longrightarrow}_{{CLL}}$ for any $\beta  \neq \alpha$.
      \item $t_1 \vee t_2 \stackrel{\alpha}{\longrightarrow}_{{CLL}} t$ iff $\alpha = \tau$, and either $t \equiv t_1$ or $t \equiv t_2$.
      \item $0 \not\stackrel{\alpha}{\longrightarrow}_{{CLL}}$ and $\bot \not\stackrel{\alpha}{\longrightarrow}_{{CLL}}$.
      \item If $t \stackrel{\alpha}{\longrightarrow}_{{CLL}}s$ then $|t|>|s|$. Hence, there is no infinite transition sequence in $LTS(CLL)$.
    \end{enumerate}
\end{lemma}
\begin{proof}
  \noindent \textbf{(1)} (Left implies Right) Assume $\alpha.t \stackrel{\beta}{\longrightarrow}_{{CLL}} t'$. Then, \[Strip({\mathcal P}_{CLL},M_{{CLL}})\vdash \alpha.t \stackrel{\beta}{\longrightarrow} t'.\]
  Further, since the axiom $\alpha.t \stackrel{\alpha}{\longrightarrow}t$ is the only rule whose conclusion has the head $\alpha.t$, we get
   $\alpha = \beta$ and $t \equiv t'$.\\

  \noindent (Right implies Left) Since $\frac{-}{\alpha.t \stackrel{\alpha}{\longrightarrow} t}$ is an axiom in $Strip(R_{CLL}, M_{CLL})$, we immediately get $\alpha.t \stackrel{\alpha}{\longrightarrow} t \in M_{CLL}$. Hence, $\alpha.t \stackrel{\alpha}{\longrightarrow}_{CLL} t$. \\

  \noindent \textbf{(2)} Similar to (1). \\

  \noindent \textbf{(3)} Assume that $0 \stackrel{\alpha}{\longrightarrow}_{{CLL}}t'$ for some $t'\in T(\Sigma_{CLL})$ and $\alpha \in Act_{\tau}$.
  Since $M_{CLL}$ is a supported transition model, there exists a rule $r \in {R_{CLL}}_{ground}$ such that $M_{CLL}\models prem(r)$ and $conc(r) \equiv 0 \stackrel{\alpha}{\longrightarrow}t'$. However, there is no such ground rule, a contradiction.
  Similarly, $\bot \not\stackrel{\alpha}{\longrightarrow}_{{CLL}}$ holds for each $\alpha \in Act_{\tau}$.\\

  \noindent \textbf{(4)} By Lemma~\ref{L:POS}, it follows from the fact that
   $M_{CLL} \subseteq M_P$.
\end{proof}

The properties in the above lemma hold for any kind of transitions (visible or invisible). The next lemma contains some simple properties which hold only for visible transitions.

\begin{lemma}\label{L:Basic_II}
Let $t_1,t_2\in T(\Sigma_{CLL})$ and $a,b \in Act$. Then
    \begin{enumerate}
      \item  $t_1 \Box t_2 \stackrel{a}{\longrightarrow}_{{CLL}} t_3$ iff  either $t_1 \stackrel{a}{\longrightarrow}_{{CLL}} t_3$ and $t_2 \not\stackrel{\tau}{\longrightarrow}_{CLL}$, or $t_2 \stackrel{a}{\longrightarrow}_{{CLL}} t_3$ and $t_1 \not\stackrel{\tau}{\longrightarrow}_{CLL}$.
          In particular, $a.t_1 \Box b.t_2 \stackrel{\alpha}{\longrightarrow}_{{CLL}} t_3$ iff  either $\alpha = a$ and $t_3 \equiv t_1$, or $\alpha = b$ and $t_3 \equiv t_2$.
      \item $t_1\wedge t_2 \stackrel{a}{\longrightarrow}_{CLL}t_3$ iff $t_1 \stackrel{a}{\longrightarrow}_{CLL} t_1'$, $t_2 \stackrel{a}{\longrightarrow}_{CLL}t_2'$ and $t_3\equiv t_1'\wedge t_2'$ for some $t_1',t_2'$.
    \end{enumerate}
\end{lemma}
\begin{proof}
  \noindent \textbf{(1)} (Left implies Right) Suppose $t_1 \Box t_2 \stackrel{a}{\longrightarrow}_{{CLL}} t_3$.
  Clearly, the last rule applied in the proof tree of  $Strip({\mathcal P}_{CLL},M_{{CLL}})\vdash t_1 \Box t_2 \stackrel{a}{\longrightarrow} t_3$ is of the form
  \[\text{either}\frac{t_1 \stackrel{a}{\longrightarrow} s}{t_1 \Box t_2 \stackrel{a}{\longrightarrow} s}\;\text{with}\; M_{{CLL}} \models t_2 \not\stackrel{\tau}{\longrightarrow}
  \;\text{or}\; \frac{t_2 \stackrel{a}{\longrightarrow} s}{t_1 \Box t_2 \stackrel{a}{\longrightarrow} s} \; \text{with} \; M_{{CLL}} \models t_1 \not\stackrel{\tau}{\longrightarrow}.\]
  Then, we get either $t_1 \stackrel{a}{\longrightarrow}_{CLL} s$, $t_2 \not\stackrel{\tau}{\longrightarrow}_{CLL}$ and $t_3 \equiv s$, or $t_2 \stackrel{a}{\longrightarrow}_{CLL} s$, $t_1 \not\stackrel{\tau}{\longrightarrow}_{CLL}$ and $t_3 \equiv s$.\\

  \noindent (Right implies Left) W.l.o.g, suppose $t_1 \stackrel{a}{\longrightarrow}_{{CLL}} t_3$ and $t_2 \not\stackrel{\tau}{\longrightarrow}_{CLL}$.
  Since $M_{CLL}\models t_2 \not\stackrel{\tau}{\longrightarrow}$, we get $ \frac{t_1 \stackrel{a}{\longrightarrow} t_3}{t_1 \Box t_2 \stackrel{a}{\longrightarrow} t_3} \in Strip(R_{CLL},M_{CLL})$.
  So, $t_1 \Box t_2 \stackrel{a}{\longrightarrow} t_3 \in M_{CLL}$ immediately follows from $t_1 \stackrel{a}{\longrightarrow} t_3 \in M_{CLL}$.\\

  \noindent \textbf{(2)} We shall prove that the left implies the right, the converse is trivial and omitted.
   Suppose $t_1\wedge t_2 \stackrel{a}{\longrightarrow}_{CLL} t_3$.
   Since $a \in Act$, the last rule applied in the proof tree of $Strip({\mathcal P}_{CLL},M_{CLL}) \vdash t_1\wedge t_2 \stackrel{a}{\longrightarrow} t_3 $ is of the form $\frac{t_1 \stackrel{a}{\longrightarrow} t_1', t_2 \stackrel{a}{\longrightarrow} t_2'}{t_1\wedge t_2 \stackrel{a}{\longrightarrow} t_1'\wedge t_2'}$ for some $t_1'$, $t_2'$.
   Thus, it follows that $t_3 \equiv t_1'\wedge t_2'$, $t_1 \stackrel{a}{\longrightarrow} t_1' \in M_{CLL}$ and $t_2 \stackrel{a}{\longrightarrow} t_2' \in M_{CLL}$.
\end{proof}

The lemma below shows that the operators $\parallel_A$, $\wedge$ and $\Box$ are static combinators w.r.t $\tau$-transitions, in other words, the structure that they represent is undisturbed by $\tau$-transitions.

\begin{lemma}\label{L:STABILIZATION}
Let $\odot \in \{\parallel_A,\wedge,\Box\}$ and $t_1,t_2\in T(\Sigma_{CLL})$.
Then, $t_1\odot t_2 \stackrel{\tau}{\longrightarrow}_{CLL}t_3$ iff there exists $s$ such that either $t_1 \stackrel{\tau}{\longrightarrow}_{CLL}s$ and $t_3\equiv s\odot t_2$, or $t_2 \stackrel{\tau}{\longrightarrow}_{CLL}s$ and $t_3\equiv t_1\odot s$.
Hence, $t_1$ and $t_2$ are stable iff $t_1 \odot t_2$ is stable.
\end{lemma}
\begin{proof}
(Left implies Right) Proceed by distinguishing cases based on the form of the last rule applied in the proof tree of $Strip({\mathcal P}_{CLL},M_{CLL})\vdash t_1\odot t_2 \stackrel{\tau}{\longrightarrow}_{CLL}t_3$.\\

\noindent (Right implies Left) It immediately follows from the fact that both $\frac{t_1 \stackrel{\tau}{\longrightarrow}t_1'}{t_1\odot t_2 \stackrel{\tau}{\longrightarrow} t_1' \odot t_2}$ and $\frac{t_2 \stackrel{\tau}{\longrightarrow}t_2'}{t_1\odot t_2 \stackrel{\tau}{\longrightarrow} t_1 \odot t_2'}$ are rules in $Strip(R_{CLL},M_{CLL})$.
\end{proof}

The next lemma provides some basic properties of $F_{CLL}$.

\begin{lemma}\label{L:F_NORMAL}
  Let $t_1,t_2\in T(\Sigma_{CLL})$. Then
  \begin{enumerate}
    \item  $t_1,t_2 \in F_{CLL}$ iff $t_1 \vee t_2 \in F_{CLL}$.
    \item $\alpha.t_1 \in F_{CLL}$ iff $t_1 \in F_{CLL}$.
    \item Either $t_1 \in F_{CLL}$ or $t_2 \in F_{CLL}$ iff $t_1 \odot t_2 \in F_{CLL}$  for each $\odot \in \{\Box, \parallel_A\}$.
    \item Either $t_1 \in F_{CLL}$ or $t_2 \in F_{CLL}$ implies $t_1 \wedge t_2 \in F_{CLL}$.
    \item $0 \notin F_{{CLL}}$ and $\bot \in F_{CLL}$.
  \end{enumerate}
\end{lemma}
\begin{proof}
We prove items (1) and (5), the others are similar and omitted.

\noindent \textbf{(1)} Assume that $t_1,t_2 \in F_{CLL}$. So, both $t_1 F$ and $t_2 F$ are in $M_{CLL}$. Further, it follows from $ \frac{t_1F,t_2F}{t_1 \vee t_2 F}\in Strip(R_{CLL},M_{CLL})$ that $t_1 \vee t_2 \in F_{CLL}$.

Conversely, assume that $t_1 \vee t_2 \in F_{CLL}$. Then, the last rule applied in the proof tree of $Strip({\mathcal P}_{CLL},M_{CLL}) \vdash t_1 \vee t_2F$ is $ \frac{t_1F,t_2F}{t_1 \vee t_2 F}$. So, we get $t_1 \in F_{CLL}$ and $t_2  \in F_{CLL}$.\\

\noindent \textbf{(5)} $\bot \in F_{CLL}$ immediately follows from $ \frac{-}{\bot F}\in Strip(R_{CLL},M_{CLL}) $. Next, we prove $0 \notin F_{CLL}$. Assume that $0F \in M_{{CLL}}$. Since $M_{{CLL}}$ is a supported transition model, there exists a rule $r \in Strip(R_{CLL},M_{CLL})$ such that $M_{{CLL}} \models prem(r)$ and $conc(r)=0F$. However, there is no rule which has the conclusion $0F$, a contradiction.
\end{proof}

An immediate consequence of the above lemma is that the operators $\alpha.()$, $\vee$, $\Box$ and $\parallel_A$ preserve consistency. More precisely, $\alpha.s$ and $s \odot t$ are not in $F_{CLL}$ if $s,t\notin F_{CLL}$, where $\odot$ is any binary operator except $\wedge$. Similar to conjunction in usual logic systems, the operator $\wedge$ doesn't preserve consistency. That is, the converse of (4) in the above lemma fails. For instance, consider the process terms $a.0$ and $b.0$. Clearly, by (2) and (5) in the above lemma, $a.0 \notin F_{CLL}$ and $b.0 \notin F_{CLL}$.
However, by Lemma~\ref{L:Basic_I}(1) and \ref{L:STABILIZATION}, we have $M_{CLL} \models \{b.0 \not\stackrel{a}{\longrightarrow}, a.0\wedge b.0 \not\stackrel{\tau}{\longrightarrow}\}$, hence $\frac{a.0 \stackrel{a}{\longrightarrow}0}{a.0 \wedge b.0 F} \in Strip(R_{CLL},M_{CLL})$.
Then, it is easy to see that $a.0\wedge b.0 \in F_{CLL}$.
In fact, the above lemma doesn't completely capture the nature of $F_{CLL}$, and further properties will be revealed in the subsequent.

\begin{lemma}\label{L:CON_LLTS}
    Let $t_1,t_2\in T(\Sigma_{CLL})$ and $\alpha \in Act_{\tau}$. If $M_{CLL}\models t_1\wedge t_2 \neg\overline{F}_{\alpha}$ then
    $t_3\in F_{CLL}$ for each $t_3$ such that $t_1\wedge t_2 \stackrel{\alpha}{\longrightarrow}_{CLL}t_3$.
\end{lemma}
\begin{proof}
   Assume  that $t_1\wedge t_2 \stackrel{\alpha}{\longrightarrow}_{CLL}t_3$ with $t_3 \notin F_{CLL}$.
   So, $t_3F \notin M_{CLL}$, which implies $ \frac{t_1\wedge t_2 \stackrel{\alpha}{\longrightarrow}t_3}{t_1\wedge t_2 \overline{F}_{\alpha}}\in Strip(R_{CLL},M_{CLL})$. Further, $t_1\wedge t_2 \overline{F}_{\alpha} \in M_{CLL}$ comes from $t_1\wedge t_2 \stackrel{\alpha}{\longrightarrow} t_3 \in M_{CLL}$, which contradicts $M_{CLL}\models t_1\wedge t_2 \neg\overline{F}_{\alpha}$.
\end{proof}

As mentioned in Section~2, the notion of $\tau$-pure is introduced in \cite{Luttgen07, Luttgen10}, which is a technique constraint for Logic LTSs. Follows \cite{Luttgen07, Luttgen10}, this paper insists such constraint. SOS rules in Table \ref{Ta:OPERATIONAL_RULES} have reflected it, while the result below will formally show that the LTS associated with CLL is indeed $\tau$-pure.

\begin{theorem}\label{L:TAU_PURE}
    $LTS({CLL})$ is $\tau$-pure.
\end{theorem}
\begin{proof}
It suffices to prove that for each $t\in T(\Sigma_{CLL})$,
  \[t \stackrel{\tau}{\longrightarrow}_{{CLL}}  \;\text{implies}\; \nexists a\in Act.\;t\stackrel{a}{\longrightarrow}_{{CLL}}.\]
We prove it by induction on the structure of $t$.

\noindent $\bullet$ $t \equiv 0$ or $t \equiv \bot$.

        By Lemma~\ref{L:Basic_I}(3), it holds trivially.

  \noindent $\bullet$ $t \equiv \alpha.t_1$.

          Assume that $t \stackrel{\tau}{\longrightarrow}_{{CLL}}$ and $t \stackrel{a}{\longrightarrow}_{{CLL}}$ for some $a\in Act$.
         So, by Lemma~\ref{L:Basic_I}(1), we have $\alpha = \tau = a$, a contradiction.

  \noindent $\bullet$ $t \equiv t_1 \vee t_2$.

        Immediately follows from Lemma~\ref{L:Basic_I}(2).

  \noindent $\bullet$ $t \equiv t_1 \Box t_2$.

        Assume that $t_1 \Box t_2 \stackrel{\tau}{\longrightarrow}_{{CLL}}$.
        Then, by Lemma~\ref{L:STABILIZATION}, we get $t_i \stackrel{\tau}{\longrightarrow}_{CLL}$ for some $i\in \{1,2\}$. Further, by IH and Lemma~\ref{L:Basic_II}(1), it follows that $t_1 \Box t_2 \not\stackrel{a}{\longrightarrow}_{CLL}$ for each $a\in Act$.

  \noindent $\bullet$ $t \equiv t_1 \wedge t_2$.

        Similar to $t_1 \Box t_2$, but using Lemma~\ref{L:Basic_II}(2) instead of Lemma~\ref{L:Basic_II}(1).

  \noindent $\bullet$ $t \equiv t_1 \parallel_A t_2$.

        Assume that $t_1 \parallel_A t_2 \stackrel{\tau}{\longrightarrow}_{{CLL}}$.
         Then, by Lemma~\ref{L:STABILIZATION}, we get $t_i \stackrel{\tau}{\longrightarrow}_{CLL}$ for some $i\in \{1,2\}$.
         By IH, we have $t_i \not\stackrel{a}{\longrightarrow}_{CLL}$ for each $a \in Act$.
         Assume that $t_1 \parallel_A t_2 \stackrel{a}{\longrightarrow} s \in M_{{CLL}}$ for some $a\in Act$ and $s \in T(\Sigma_{CLL})$.
         Then, the last rule applied in the proof tree of $Strip({\mathcal P}_{CLL},M_{{CLL}}) \vdash t_1 \parallel_A t_2 \stackrel{a}{\longrightarrow}s $ has one of formats below.
         \begin{enumerate}
           \item $\frac{t_1 \stackrel{a}{\longrightarrow} r}{t_1 \parallel_A t_2 \stackrel{a}{\longrightarrow} r \parallel_A t_2}$ with $a\notin A$ and $M_{{CLL}} \models t_2 \not \stackrel{\tau}{\longrightarrow}$,
           \item $\frac{t_2 \stackrel{a}{\longrightarrow} r}{t_1 \parallel_A t_2 \stackrel{a}{\longrightarrow} t_1 \parallel_A r}$ with $a\notin A$ and $M_{{CLL}} \models t_1 \not \stackrel{\tau}{\longrightarrow}$,
           \item $\frac{t_1 \stackrel{a}{\longrightarrow} r,t_2 \stackrel{a}{\longrightarrow} q}{t_1 \parallel_A t_2 \stackrel{a}{\longrightarrow} r \parallel_A q }$ with $a\in A$.
         \end{enumerate}
        It is trivial to check that each case  above leads to a contradiction, as desired.
\end{proof}

In the following, we shall prove that $LTS(CLL)$ is a Logic LTS. We proceed by proving that both (LTS1) and (LTS2) hold in $LTS(CLL)$.

\begin{lemma}\label{L:LLTS_I}
   $LTS(CLL)$ satisfies (LTS1).
\end{lemma}
\begin{proof}
 It is enough to show that for each $t\in T(\Sigma_{CLL})$
    \[\exists \alpha \in {\mathcal I}(t)\;\forall s (t \stackrel{\alpha}{\longrightarrow}_{{CLL}} s \;\text{implies}\; s \in F_{{CLL}})\; \text{implies}\; t \in F_{{CLL}}. \]
    We prove it by induction on the structure of $t$.

  \noindent $\bullet$ $t \equiv 0$, $\bot$, $\beta.t_1$ or $t_1 \vee t_2$.

        Immediately follows from Lemma~\ref{L:F_NORMAL}(1)(2) and Lemma~\ref{L:Basic_I}(1)(2)(3).
        Notice that, for $t\equiv 0$ or $\bot$, it holds trivially due to ${\mathcal I}(t)= \emptyset$.

  \noindent $\bullet$ $t \equiv t_1 \Box t_2$.

        Assume that $\forall t'(t \stackrel{\alpha}{\longrightarrow}_{{CLL}} t'\;\text{implies}\; t' \in F_{{CLL}})$ for some $\alpha \in {\mathcal I}(t)$. Since $\alpha \in {\mathcal I}(t)$, we get $t_1\Box t_2 \stackrel{\alpha}{\longrightarrow}_{{CLL}}t'$ for some $t'$.
        Consider two cases below.\\

       \noindent Case 1 $\alpha \in Act$.

        By Lemma~\ref{L:Basic_II}(1), we have
        \[\text{either}\; \alpha \in {\mathcal I}(t_2)\; \text{and}\; t_1 \not\stackrel{\tau}{\longrightarrow}_{CLL}, \;\text{or}\; \alpha \in {\mathcal I}(t_1) \;\text{and}\; t_2 \not\stackrel{\tau}{\longrightarrow}_{CLL}.\]
        W.l.o.g, we consider the first alternative.
        In such case, it follows from Lemma~\ref{L:Basic_II}(1) that
        \[\forall t_2'(t_2 \stackrel{\alpha}{\longrightarrow}_{{CLL}} t_2'\;\text{implies}\; t_1 \Box t_2 \stackrel{\alpha}{\longrightarrow}_{CLL}t_2').\]
        Thus, by the assumption, it holds that
        \[\forall t_2'(t_2 \stackrel{\alpha}{\longrightarrow}_{{CLL}} t_2'\;\text{implies}\; t_2' \in F_{CLL}).\]
        So, by IH, $t_2 \in F_{CLL}$.
        Then, $t_1 \Box t_2  \in F_{{CLL}}$ immediately follows from Lemma~\ref{L:F_NORMAL}(3).\\

       \noindent Case 2 $\alpha = \tau$.

        So, by Lemma~\ref{L:STABILIZATION}, it follows that either $t_1 \stackrel{\tau}{\longrightarrow}_{CLL}$ or $t_2 \stackrel{\tau}{\longrightarrow}_{CLL}$.
        W.l.o.g, we consider the former.
        If $t_2 \in F_{{CLL}}$ then $t_1\Box t_2 \in F_{{CLL}}$ comes from Lemma~\ref{L:F_NORMAL}(3) at once.
        In the following, we deal with the case $t_2 \notin F_{{CLL}}$.
        Firstly, it follows from Lemma~\ref{L:STABILIZATION} that
        \[\forall t_1' (t_1  \stackrel{\tau}{\longrightarrow}_{CLL} t_1' \;\text{implies}\; t_1 \Box t_2 \stackrel{\tau}{\longrightarrow}_{CLL} t_1'\Box t_2).\]
        Then, by the assumption, we get
        \[\forall t_1' (t_1  \stackrel{\tau}{\longrightarrow}_{CLL} t_1' \;\text{implies}\; t_1' \Box t_2 \in F_{CLL}).\]
        Further, by $t_2 \notin F_{CLL}$ and Lemma~\ref{L:F_NORMAL}(3), it holds that
         \[\forall t_1' (t_1  \stackrel{\tau}{\longrightarrow}_{CLL} t_1' \;\text{implies}\; t_1' \in F_{CLL}).\]
        So, by IH, $t_1 \in F_{{CLL}}$. Therefore, by Lemma~\ref{L:F_NORMAL}(3), it follows that $t_1\Box t_2 \in F_{{CLL}}$, as desired.

  \noindent $\bullet$ $t \equiv t_1 \wedge t_2$.

        Assume that $\alpha \in {\mathcal I}(t)$ and $s \in F_{CLL}$ for each $s$ such that $t \stackrel{\alpha}{\longrightarrow}_{CLL}s$.
        Thus, $ t_1\wedge t_2 \stackrel{\alpha}{\longrightarrow} t_3 \in M_{{CLL}}$ for some $t_3$.
        If $t_1 \wedge t_2 \overline{F}_{\alpha} \notin M_{{CLL}}$(i.e., $M_{{CLL}} \models t_1 \wedge t_2 \neg \overline{F}_{\alpha} $) then $ \frac{t_1\wedge t_2 \stackrel{\alpha}{\longrightarrow} t_3 }{t_1\wedge t_2F}\in Strip(R_{CLL},M_{{CLL}}) $, which, with the helping of $ t_1\wedge t_2 \stackrel{\alpha}{\longrightarrow} t_3 \in M_{{CLL}}$, implies $t_1\wedge t_2F \in M_{{CLL}}$. Therefore, in order to complete the proof, it is enough to show that $t_1 \wedge t_2 \overline{F}_{\alpha} \notin M_{{CLL}}$.

        Suppose that $t_1 \wedge t_2 \overline{F}_{\alpha} \in M_{{CLL}}$. Clearly, the last rule applied in the proof tree of $Strip({\mathcal P}_{CLL},M_{{CLL}})\vdash t_1 \wedge t_2 \overline{F}_{\alpha}$  has the form below
        \[\frac{t_1 \wedge t_2 \stackrel{\alpha}{\longrightarrow} t_4}{t_1 \wedge t_2 \overline{F}_{\alpha}} \; \text{with} \;  M_{{CLL}} \models t_4 \neg F.\]
        So, $t_1 \wedge t_2 \stackrel{\alpha}{\longrightarrow} t_4 \in M_{{CLL}}$. Then, by the assumption, we get $t_4 \in F_{{CLL}}$, which contradicts $M_{{CLL}} \models t_4 \neg F$.

  \noindent $\bullet$ $t \equiv t_1 \parallel_A t_2$.

        Assume that $\forall s(t\stackrel{\alpha}{\longrightarrow}_{CLL}s\;
        \text{implies}\;s\in F_{CLL})$ for some $\alpha \in {\mathcal I}(t)$.
        Since $\alpha \in {\mathcal I}(t)$, we have $t_1\parallel_A t_2 \stackrel{\alpha}{\longrightarrow}_{{CLL}}t'$ for some $t'$.
         If $\alpha = \tau$ then the proof is similar to one of $t_1 \Box t_2 $, we omit it.
         In the following, we consider the case where $\alpha \in Act$. In such case, the last rule applied in the proof tree of $Strip({\mathcal P}_{CLL},M_{{CLL}}) \vdash t_1 \parallel_A t_2 \stackrel{\alpha}{\longrightarrow} t'$ has one of the following three formats.\\

        \noindent Case 1 $\frac{t_1 \stackrel{\alpha}{\longrightarrow} t_1'}{t_1 \parallel_A t_2 \stackrel{\alpha}{\longrightarrow} t_1' \parallel_A t_2} (\alpha\notin A)$ with $M_{{CLL}} \models t_2 \not \stackrel{\tau}{\longrightarrow}$.

        Then, $ t' \equiv t_1' \parallel_A t_2$ and $t_1 \stackrel{\alpha}{\longrightarrow} t_1' \in M_{{CLL}}$.
        If $t_2 \in F_{{CLL}}$ then, by Lemma~\ref{L:F_NORMAL}(3), $t_1\parallel_A t_2 \in F_{{CLL}}$, as desired.
        Next, we consider another case where $t_2 \notin F_{{CLL}}$.
        Since $M_{{CLL}} \models t_2 \not \stackrel{\tau}{\longrightarrow}$ and $\alpha \notin A$, we get
        \[\frac{t_1 \stackrel{\alpha}{\longrightarrow} t_1''}{t_1 \parallel_A t_2 \stackrel{\alpha}{\longrightarrow} t_1'' \parallel_A t_2} \in Strip(R_{CLL},M_{{CLL}})\;\text{for each}\; t_1''.\]
        Then, it follows that
        \[\forall t_1'' (t_1  \stackrel{\alpha}{\longrightarrow}_{CLL} t_1'' \;\text{implies}\; t_1 \parallel_A t_2 \stackrel{\alpha}{\longrightarrow}_{CLL} t_1''\parallel_A t_2).\]
        Moreover, by the assumption, we have
         \[\forall t_1'' (t_1  \stackrel{\alpha}{\longrightarrow}_{CLL} t_1'' \;\text{implies}\; t_1''\parallel_A t_2 \in F_{CLL}).\]
         Further, since $t_2\notin F_{CLL}$, by Lemma~\ref{L:F_NORMAL}(3), it holds that
         \[\forall t_1'' (t_1  \stackrel{\alpha}{\longrightarrow}_{CLL} t_1'' \;\text{implies}\; t_1'' \in F_{CLL}).\]
         So, by IH, $t_1 \in F_{CLL}$.
         Then, by Lemma~\ref{L:F_NORMAL}(3), $t_1 \parallel_A t_2 \in F_{{CLL}}$ immediately follows.\\

       \noindent Case 2 $\frac{t_2 \stackrel{\alpha}{\longrightarrow} t_2'}{t_1 \parallel_A t_2 \stackrel{\alpha}{\longrightarrow} t_1\parallel_A t_2'}(\alpha\notin A)$ with $M_{{CLL}} \models t_1 \not \stackrel{\tau}{\longrightarrow}$.

        Similar to Case 1.\\

       \noindent Case 3 $\frac{t_1 \stackrel{\alpha}{\longrightarrow} t_1',t_2 \stackrel{\alpha}{\longrightarrow} t_2'}{t_1 \parallel_A t_2 \stackrel{\alpha}{\longrightarrow} t_1' \parallel_A t_2'}( \alpha\in A) $.

        In such case, we get $t_1 \stackrel{\alpha}{\longrightarrow}_{CLL} t_1'$ and $t_2 \stackrel{\alpha}{\longrightarrow}_{CLL} t_2'$. Thus, $\alpha \in {\mathcal I}(t_1)$ and $\alpha \in {\mathcal I}(t_2)$. In order to complete the proof, it is enough to show that either $t_1 \in F_{CLL}$ or $t_2 \in F_{CLL}$.
        Assume that $t_1 \notin F_{{CLL}}$ and $t_2 \notin F_{{CLL}}$.
        Then, by IH, we get $t_1 \stackrel{\alpha}{\longrightarrow} s_1 \in  M_{{CLL}}$ with $s_1 \notin F_{{CLL}}$ and $t_2 \stackrel{\alpha}{\longrightarrow} s_2 \in  M_{{CLL}}$ with $s_2 \notin F_{{CLL}}$ for some $s_1$ and $s_2$.
        Since $\alpha \in A$, it follows that \[\frac{t_1 \stackrel{\alpha}{\longrightarrow} s_1,t_2 \stackrel{\alpha}{\longrightarrow} s_2}{t_1 \parallel_A t_2 \stackrel{\alpha}{\longrightarrow} s_1 \parallel_A s_2 } \in Strip(R_{CLL},M_{{CLL}}).\]
         Hence, $t_1 \parallel_A t_2 \stackrel{\alpha}{\longrightarrow}_{CLL} s_1 \parallel_A s_2 $.
         Further, by assumption, we have $s_1 \parallel_A s_2 \in F_{{CLL}}$.
         By Lemma~\ref{L:F_NORMAL}(3), this contradicts $s_1 \notin F_{{CLL}}$ and $s_2 \notin F_{{CLL}}$.
\end{proof}

In fact, (LTS1) can be strengthened so as to provide a complete character of non-stable processes in the set $F_{CLL}$ in terms of $\tau$-transitions.
Formally, we have the result below.

\begin{lemma}\label{L:F_TAU_I}
 If $\tau \in {\mathcal I}(t)$ then
  \[t \in F_{CLL} \;\text{iff}\;\; \forall s(t\stackrel{\tau}{\longrightarrow}_{CLL}s\;\text{implies}\; s\in F_{CLL}).\]
\end{lemma}
\begin{proof}
It immediately follows from Lemma~\ref{L:LLTS_I} that the right implies the left.
Another implication can be proved by induction on the structure of $t$, we leave it to the reader.
\end{proof}

As an immediate consequence, we also have, for each $t \in T(\Sigma_{CLL})$,
\[\forall t'(t \stackrel{\epsilon}{\Longrightarrow}_{{CLL}}|t' \;\text{implies}\; t' \in F_{{CLL}}) \;\text{iff}\; t \in F_{{CLL}}.\]
Although this result characterizes all processes in $F_{CLL}$, it isn't more interesting than Lemma~\ref{L:F_TAU_I}.
In fact, for stable process terms, this result is trivial.
For non-stable process terms, since there is no infinite $\tau$-transition sequence in $M_{CLL}$, it isn't difficult to see that this result is implied by Lemma~\ref{L:F_TAU_I}.

\begin{lemma}\label{L:LLTS_II}
    $LTS(CLL)$ satisfies (LTS2).
\end{lemma}
\begin{proof}
    It suffices to show that, for each $t$, if $t\notin F_{CLL}$ then $t \stackrel{\epsilon}{\Longrightarrow}_F|t'$ for some $t'$.
   We prove it by induction on the degree of $t$.
   Assume that it holds for all $s$ with $|s|<|t|$.

   If $t$ is stable, then $t \stackrel{\epsilon}{\Longrightarrow}_F|t$ follows from $t\notin F_{CLL}$.
   Otherwise, since $t\notin F_{CLL}$ and $\tau \in {\mathcal I}(t)$, by Lemma~\ref{L:LLTS_I}, we have $t \stackrel{\tau}{\longrightarrow}_F s$ for some $s$.
   By Lemma~\ref{L:Basic_I}(4), we get $|s|<|t|$. Then, by IH, we obtain  $s \stackrel{\epsilon}{\Longrightarrow}_F|t'$ for some $t'$. Hence, $t \stackrel{\epsilon}{\Longrightarrow}_F|t'$.
\end{proof}

It is now a short step to

\begin{theorem}\label{L:LLTS}
    $LTS({CLL})$ is a Logic LTS.
\end{theorem}
\begin{proof}
    Immediately follows from Lemma~\ref{L:LLTS_I} and \ref{L:LLTS_II}.
\end{proof}

In the next section, the relation $\stackrel{\epsilon}{\Longrightarrow}_F|$ will play an important role in developing the behavioral theory of CLL. The remainder of this section is devoted to basic properties of it.

Lemma~\ref{L:STABILIZATION} asserts that, for operators $\parallel_A$, $\wedge$ and $\Box$, the structure that they represent is preserved under $\tau$-transitions.
A possible conjecture is that such property may be generalized to the circumstance where the stability and consistency of process terms are involved, that is, these operators are also static combinators w.r.t the transition relation  $\stackrel{\epsilon}{\Longrightarrow}_F|$.
It turns out that this conjecture almost holds except that we need to add a moderate condition when considering the operator $\wedge$. Formally, we have the result below.

\begin{lemma}\label{L:TAU_I}
  Let $\odot \in \{\Box,\parallel_A,\wedge\}$ and $t_1,t_2 \in T(\Sigma_{CLL})$.
    \begin{enumerate}
      \item If $t_1 \stackrel{\epsilon}{\Longrightarrow}_F| t_1'$ and $t_2 \stackrel{\epsilon}{\Longrightarrow}_F|t_2'$, then $t_1 \odot t_2 \stackrel{\epsilon}{\Longrightarrow}_F| t_1' \odot t_2'$  for $\odot \in \{\Box,\parallel_A\}$,
           moreover, $t_1\wedge t_2 \stackrel{\epsilon}{\Longrightarrow}_F| t_1' \wedge t_2'$ if $t_1' \wedge t_2' \notin F_{{CLL}}$.
      \item If $t_1 \odot t_2 \stackrel{\epsilon}{\Longrightarrow}_F| t_3$ then $t_1 \stackrel{\epsilon}{\Longrightarrow}_F| t_1'$, $t_2 \stackrel{\epsilon}{\Longrightarrow}_F|t_2'$ and $t_3 \equiv t_1' \odot t_2'$ for some $t_1', t_2' \in T(\Sigma_{CLL})$.
    \end{enumerate}
\end{lemma}
\begin{proof}
\noindent\textbf{(1)} We consider the case where $\odot = \wedge$, the others can be treated similarly and omitted.

Suppose $t_1 \stackrel{\epsilon}{\Longrightarrow}_F| t_1'$, $t_2 \stackrel{\epsilon}{\Longrightarrow}_F|t_2'$ and $t_1' \wedge t_2' \notin F_{{CLL}}$.
So, we have $t_1\equiv S_0 \stackrel{\tau}{\longrightarrow}_F,\dots,\stackrel{\tau}{\longrightarrow}_F S_n\equiv t_1'$ for some $S_0,\dots,S_n$ with $n \geq 0$, and $t_2\equiv T_0 \stackrel{\tau}{\longrightarrow}_F,\dots,\stackrel{\tau}{\longrightarrow}_F T_m\equiv t_2'$ for some $T_0,\dots,T_m$ with $m \geq0$.
    By Lemma~\ref{L:STABILIZATION}, it is easy to see that
    \begin{multline*}
        t_1 \wedge t_2 \equiv S_0 \wedge T_0 \stackrel{\tau}{\longrightarrow}_{{CLL}}S_1 \wedge T_0, \dots, \stackrel{\tau}{\longrightarrow}_{{CLL}} S_n\wedge T_0 \\
    \stackrel{\tau}{\longrightarrow}_{{CLL}} S_n \wedge T_1 \stackrel{\tau}{\longrightarrow}_{{CLL}} S_n \wedge T_2,\dots, \stackrel{\tau}{\longrightarrow}_{{CLL}} S_n\wedge T_m \equiv t_1'\wedge t_2'. \tag{\ref{L:TAU_I}.1}
    \end{multline*}
    If $n=m=0$ then both $t_1$ and $t_2$ are stable and $t_i\equiv t_i'$ for $i=1,2$.
    So, $t_1\wedge t_2 \stackrel{\epsilon}{\Longrightarrow}_F|t_1'\wedge t_2'$ holds trivially.
    Next, we consider the case either $n \neq 0$ or $m \neq 0$.
    Then, by Lemma~\ref{L:STABILIZATION}, $\tau \in {\mathcal I}(p)$ for each process term $p$ except $S_n \wedge T_m$ along the sequence (\ref{L:TAU_I}.1).
    Further, since $t_1'\wedge t_2' \notin F_{CLL}$, by Lemma~\ref{L:F_TAU_I}, it is easy to know that  $p \notin F_{{CLL}}$ for each $p$ occurring in (\ref{L:TAU_I}.1).
    Moreover, since $t_1'$ and $t_2'$ are stable, by Lemma~\ref{L:STABILIZATION}, so is $t_1'\wedge t_2'$.
    Therefore, $t_1\wedge t_2 \stackrel{\epsilon}{\Longrightarrow}_F| t_1' \wedge t_2'$.\\

\noindent \textbf{(2)}
    We shall treat the case where $\odot = \wedge$, the remaining proofs are similar and omitted.

    Suppose $t_1 \wedge t_2 \stackrel{\epsilon}{\Longrightarrow}_F| t_3$. Then, for some $n \geq 0$ and $S_i(i\leq n)$, we have
    \[t_1 \wedge t_2 \equiv S_0 \stackrel{\tau}{\longrightarrow}_F S_1,\dots,\stackrel{\tau}{\longrightarrow}_F S_n \equiv t_3 \]
     In the following, we consider the non-trivial case $n>0$. In such case, by Lemma~\ref{L:STABILIZATION}, there exist $t_{1i},t_{2i}(i\leq n)$ such that
     \begin{enumerate}
       \item $S_i \equiv t_{1i} \wedge t_{2i}$  for each $i\leq n$ ,
       \item either $t_{1i} \stackrel{\tau}{\longrightarrow}_{CLL} t_{1i+1}$ and $t_{2i+1}\equiv t_{2i}$, or  $t_{2i} \stackrel{\tau}{\longrightarrow}_{CLL} t_{2i+1}$ and $t_{1i+1}\equiv t_{1i}$  for each $i< n$.
     \end{enumerate}

     For each $i\leq n$, since $S_i\notin F_{CLL}$, by Lemma~\ref{L:F_NORMAL}(4), we have $t_{1i},t_{2i} \notin F_{CLL}$. Moreover, by Lemma~\ref{L:STABILIZATION}, since $t_3 \equiv t_{1n} \wedge t_{2n}$ is stable, so are $t_{1n}$ and $t_{2n}$. Finally, we can obtain a $\tau$-transition sequence from $t_1$ to $t_{1n}$ by removing duplicate process terms in the sequence $\{t_{1i}\}_{i\leq n}$ \footnote{Notice that this $\tau$-transition sequence is empty if $t_1 \equiv t_{1n}$, i.e., $t_1$ is stable.}.
     So, $t_1 \stackrel{\epsilon}{\Longrightarrow}_F| t_{1n}$.
     Similarly, $t_2 \stackrel{\epsilon}{\Longrightarrow}_F| t_{2n}$ holds.
\end{proof}

By the way, if the consistency is ignored, similar to the result above, it is not hard to show that, for each $\odot \in \{\Box,\parallel_A,\wedge\}$,

\begin{enumerate}
  \item   If $t_1 \stackrel{\epsilon}{\Longrightarrow}_{{CLL}}| t_1'$ and $t_2 \stackrel{\epsilon}{\Longrightarrow}_{{CLL}} |t_2'$, then $t_1 \odot t_2 \stackrel{\epsilon}{\Longrightarrow}_{{CLL}}| t_1' \odot t_2'$.
  \item  If $t_1 \odot t_2 \stackrel{\epsilon}{\Longrightarrow}_{{CLL}}| t_3$ then $t_1 \stackrel{\epsilon}{\Longrightarrow}_{{CLL}}| t_1'$, $t_2 \stackrel{\epsilon}{\Longrightarrow}_{{CLL}} |t_2'$ and $t_3 \equiv t_1' \odot t_2'$ for some $t_1', t_2' \in T(\Sigma_{CLL})$.
\end{enumerate}

We conclude this section with a useful result, which will be used in Section~5 when we deal with distributive laws.

\begin{lemma}\label{L:DIS}
    Let $\odot\in\{\Box,\wedge,\parallel_A\}$. If $t_1 \odot( t_2 \vee t_3) \stackrel{\epsilon}{\Longrightarrow}_F|t_4$  then there is $t_1'$ and a sequence $t_1 \odot( t_2 \vee t_3) \equiv T_0 \stackrel{\tau}{\longrightarrow}_F,\dots,\stackrel{\tau}{\longrightarrow}_F T_n \equiv t_4$ ($n\geq 1$) such that $t_1\stackrel{\epsilon}{\Longrightarrow}_F t_1'$, $T_j\equiv t_1' \odot( t_2 \vee t_3)$ and $T_{j+1}\equiv t_1' \odot t_k $ for some $j<n$ and $k\in \{2,3\}$.
\end{lemma}
\begin{proof}
   We consider the case where $\odot = \wedge$, the others are handled in a similar way.
   Let $t_1 \wedge ( t_2 \vee t_3) (\stackrel{\tau}{\longrightarrow}_F)^n|t_4(n\geq 0)$.
   Since $t_2 \vee t_3 \stackrel{\tau}{\longrightarrow}_{CLL}$, we get $t_1 \wedge ( t_2 \vee t_3) \stackrel{\tau}{\longrightarrow}_{CLL}$.
   So, $n\geq 1$. We proceed by induction on $n$.

   For the inductive basis $n=1$, since $t_4$ is stable, by Lemma~\ref{L:STABILIZATION} and \ref{L:Basic_I}(2), the last rule applied in the proof tree of $Strip({\mathcal P}_{CLL},M_{CLL}) \vdash t_1 \wedge ( t_2 \vee t_3) \stackrel{\tau}{\longrightarrow}t_4$ is
   \[\text{either}\; \frac{t_2 \vee t_3 \stackrel{\tau}{\longrightarrow}t_2}{t_1 \wedge ( t_2 \vee t_3) \stackrel{\tau}{\longrightarrow}t_1 \wedge t_2}
    \;\text{or}\; \frac{t_2 \vee t_3 \stackrel{\tau}{\longrightarrow}t_3}{t_1 \wedge ( t_2 \vee t_3) \stackrel{\tau}{\longrightarrow}t_1 \wedge t_3}.\]
   W.l.o.g, we consider the first alternative.
   In such case, $t_4 \equiv t_1\wedge t_2$ and we get the sequence $T_0 \equiv t_1 \wedge( t_2 \vee t_3) \stackrel{\tau}{\longrightarrow}_F|T_1 \equiv t_1\wedge t_2$.
   Hence, by Lemma~\ref{L:STABILIZATION}, $t_1$ is stable.
   Moreover, since $t_1 \wedge ( t_2 \vee t_3) \notin F_{CLL}$, by Lemma~\ref{L:F_NORMAL}(4), we  have $t_1 \notin F_{CLL}$.
   So, $t_1 \stackrel{\epsilon}{\Longrightarrow}_F|t_1$.  Then, $t_1' \triangleq t_1$ and $j\triangleq0$ are what we need.

   For the inductive step, assume $t_1 \wedge ( t_2 \vee t_3) (\stackrel{\tau}{\longrightarrow}_F)^{k+1}|t_4$.
   So, $t_1 \wedge ( t_2 \vee t_3)\stackrel{\tau}{\longrightarrow}_F t_5 (\stackrel{\tau}{\longrightarrow}_F)^{k}|t_4$ for some $t_5$.
   We distinguish two cases based on the form of the last rule applied in the proof tree of $Strip({\mathcal P}_{CLL},M_{CLL}) \vdash t_1 \wedge ( t_2 \vee t_3) \stackrel{\tau}{\longrightarrow}t_5$.\\

   \noindent Case 1 $\frac{t_1 \stackrel{\tau}{\longrightarrow}t_1'}{t_1 \wedge ( t_2 \vee t_3) \stackrel{\tau}{\longrightarrow}t_1'\wedge ( t_2 \vee t_3) }$.

    So, $t_1 \stackrel{\tau}{\longrightarrow}t_1' \in M_{CLL}$ and $t_5 \equiv t_1'\wedge ( t_2 \vee t_3)$.
    By IH, there exists a sequence
   \[t_1' \wedge ( t_2 \vee t_3) \equiv T_0 \stackrel{\tau}{\longrightarrow}_F,\dots,\stackrel{\tau}{\longrightarrow}_F T_m \equiv t_4 (m \geq 1)\]
    and for some $t_1''$ and $j<m$, we have
    \begin{enumerate}
      \item  $t_1' \stackrel{\epsilon}{\Longrightarrow}_F t_1''$,
      \item $T_j\equiv t_1'' \wedge ( t_2 \vee t_3)$ and
      \item either $T_{j+1}\equiv t_1'' \wedge t_2 $ or $T_{j+1}\equiv t_1'' \wedge  t_3$.
    \end{enumerate}
   Moreover, by Lemma~\ref{L:F_NORMAL}(4), $t_1\notin F_{CLL}$ and $t_1' \notin F_{CLL}$ follow from $t_1 \wedge (t_2 \vee t_3) \notin F_{CLL}$ and $t_1' \wedge (t_2 \vee t_3) \notin F_{CLL}$, respectively.
   Then, we have the sequence below
   \[t_1 \wedge ( t_2 \vee t_3) \stackrel{\tau}{\longrightarrow}_F t_1' \wedge ( t_2 \vee t_3) \equiv T_0 \stackrel{\tau}{\longrightarrow}_F,\dots,\stackrel{\tau}{\longrightarrow}_F T_m \equiv t_4\]
   and $t_1 \stackrel{\tau}{\longrightarrow}_F  t_1' \stackrel{\epsilon}{\Longrightarrow}_F t_1''$, $T_j\equiv t_1'' \wedge ( t_2 \vee t_3)$ and either $T_{j+1}\equiv t_1'' \wedge t_2 $ or $T_{j+1}\equiv t_1'' \wedge  t_3$, as desired.\\

   \noindent  Case 2 $\frac{t_2 \vee t_3 \stackrel{\tau}{\longrightarrow}t_6}{t_1 \wedge ( t_2 \vee t_3) \stackrel{\tau}{\longrightarrow}t_1\wedge t_6}$.

    So, $t_2 \vee t_3 \stackrel{\tau}{\longrightarrow}t_6 \in M_{CLL}$ and $t_5 \equiv t_1\wedge t_6$.
    Moreover, by Lemma~\ref{L:Basic_I}(2), either $t_6\equiv t_2$ or $t_6 \equiv t_3$. Clearly, we have the sequence $t_1 \wedge (t_2 \vee t_3) \stackrel{\tau}{\longrightarrow}_F t_1 \wedge t_6 (\stackrel{\tau}{\longrightarrow}_F)^{k}|t_4$, and, $t_1' \triangleq t_1$ and $j\triangleq0$ are what we need.
\end{proof}

\section{Behavioral Theory of CLL}

This section will develop the behavioral theory of CLL. Follows \cite{Luttgen10}, the notion of ready simulation below is adopted to formalize the refinement relation among process terms, which is a modified version of the usual notion of ready simulation (see, e.g., \cite{Glabbeek01}).
We will only care about the properties of ready simulation which are needed in establishing the soundness of the axiomatic system $AX_{CLL}$ in the next section.

\begin{mydefn}[Ready simulation \cite{Luttgen10}]\label{D:READYSIMULATION_TERMS}
A relation ${\mathcal R} \subseteq T(\Sigma_{CLL})\times T(\Sigma_{CLL})$ is a stable ready simulation relation, if for any $(t,s) \in {\mathcal R}$ and $a \in Act $\\
\textbf{(RS1)} $t$, $s$ stable;\\
\textbf{(RS2)} $t \notin F_{{CLL}}$ implies $s \notin F_{{CLL}}$;\\
\textbf{(RS3)} $t \stackrel{a}{\Longrightarrow}_F|t'$ implies $\exists s'.s \stackrel{a}{\Longrightarrow}_F|s'\; \textrm{and}\;(t',s') \in {\mathcal R}$;\\
\textbf{(RS4)} $t\notin F_{{CLL}}$ implies ${\mathcal I}(t)={\mathcal I}(s)$.
\end{mydefn}

 We say that $t$ is stable ready simulated by $s$, in symbols $t \underset{\thicksim}{\sqsubset}_{RS} s$, if there exists a stable ready simulation relation $\mathcal R$ with $(t,s) \in {\mathcal R}$.
 Further, $t$ is said to be ready simulated by $s$, written $t\sqsubseteq_{RS}s$, if \[\forall t'(t\stackrel{\epsilon}{\Longrightarrow}_F| t' \;\text{implies}\; \exists s'(s \stackrel{\epsilon}{\Longrightarrow}_F| s'\; \text{and}\;t' \underset{\thicksim}{\sqsubset}_{RS} s')).\]
 It is easy to see that both $\underset{\thicksim}{\sqsubset}_{RS}$ and $\sqsubseteq_{RS}$ are pre-order (i.e., reflexive and transitive). The equivalence relations induced by them are denoted by $\approx_{RS}$ and $=_{RS}$, respectively.
 That is, \[\approx_{RS} \;\triangleq\; \underset{\thicksim}{\sqsubset}_{RS} \cap (\underset{\thicksim}{\sqsubset}_{RS})^{-1}\; \text{and}\; =_{RS} \;\triangleq\; \sqsubseteq_{RS} \cap (\sqsubseteq_{RS})^{-1}.\] The identity relation over stable process terms is denoted by $Id_S$.
 It is obvious that both $Id_S$ and $\underset{\thicksim}{\sqsubset}_{RS}$ are stable ready simulation relations.

 The remainder of this section is taken up with proving (in)equational laws concerning $=_{RS}$ ($\sqsubseteq_{RS}$, respectively).
 These laws capture inherent properties of composition operators, e.g., commutativity, associativity and zero element.
 A number of laws obtained by L\"{u}ttgen and Vogler in \cite{Luttgen10}, which are needed in the next section, will also be rephrased in process-algebraic style.
 In particular, we will show that the ready simulation is a precongruence, that is, it is preserved by all algebraic contexts.

We begin with the laws about $\bot$ and $\tau$ occurring in prefix construction.

\begin{proposition}[Prefix]\label{S:PREFIX}\hfill
    \begin{enumerate}
      \item $a.\bot =_{RS} \bot$.
      \item $\tau.t =_{RS} t$.
    \end{enumerate}
\end{proposition}
\begin{proof}
\noindent \textbf{(1) }  Since $\bot  \in F_{{CLL}}$, by Lemma~\ref{L:F_NORMAL}(2), we have $a.\bot  \in F_{{CLL}}$. So, $a.\bot =_{RS} \bot$ holds trivially.\\

\noindent \textbf{(2)} Firstly, we prove $\tau.t \sqsubseteq_{RS} t$. Suppose $\tau.t \stackrel{\epsilon}{\Longrightarrow}_F| t'$. So, $\tau.t \stackrel{\tau}{\longrightarrow}_F t \stackrel{\epsilon}{\Longrightarrow}_F| t'$.
Then, we have $t \stackrel{\epsilon}{\Longrightarrow}_F| t'$ and $t'\underset{\thicksim}{\sqsubset}_{RS} t'$.\\

Secondly, we prove $t \sqsubseteq_{RS} \tau.t$.
Suppose $t \stackrel{\epsilon}{\Longrightarrow}_F| t'$.
Then, $\tau.t \stackrel{\tau}{\longrightarrow}_{{CLL}} t \stackrel{\epsilon}{\Longrightarrow}_F| t'$.
Moreover, by Lemma~\ref{L:F_NORMAL}(2), $\tau.t \notin F_{{CLL}}$ follows from  $t \notin F_{{CLL}}$.
Thus, $\tau.t \stackrel{\epsilon}{\Longrightarrow}_F| t'$ and $t'\underset{\thicksim}{\sqsubset}_{RS} t'$, as desired.
\end{proof}

The second set of (in)equations focuses on the properties of the combinator $\vee$.

\begin{proposition}[Disjunction]\label{S:DISJUNCTION}\hfill
    \begin{enumerate}
      \item $t_1 \vee t_2 =_{RS} t_2 \vee t_1$.
      \item $(t_1 \vee t_2) \vee t_3 =_{RS} t_1 \vee (t_2 \vee t_3)$.
      \item $t \vee t =_{RS} t$.
      \item $t \vee \bot =_{RS} t$.
      \item $t_1 \sqsubseteq_{RS} t_1 \vee t_2$.
    \end{enumerate}
\end{proposition}
\begin{proof}
\noindent \textbf{(1)} It is enough to prove $t_1 \vee t_2 \sqsubseteq_{RS} t_2 \vee t_1$.
    Suppose $t_1 \vee t_2 \stackrel{\epsilon}{\Longrightarrow}_F| t'$.
    It follows from Lemma~\ref{L:Basic_I}(2) that $t_1 \vee t_2 \stackrel{\tau}{\longrightarrow}_F t_i \stackrel{\epsilon}{\Longrightarrow}_F| t'$ for some $i\in \{1,2\}$.
    W.l.o.g, assume that $i=1$.
    Since $t_1 \vee t_2 \notin F_{{CLL}}$, by Lemma~\ref{L:F_NORMAL}(1), we get $t_2 \vee t_1 \notin F_{{CLL}}$.
    Further, it follows from $t_2 \vee t_1 \stackrel{\tau}{\longrightarrow}_{{CLL}}t_1 \stackrel{\epsilon}{\Longrightarrow}_F| t'$ that $t_2 \vee t_1 \stackrel{\epsilon}{\Longrightarrow}_F| t'$ and $t'\underset{\thicksim}{\sqsubset}_{RS} t'$. Thus,  $t_1 \vee t_2 \sqsubseteq_{RS} t_2 \vee t_1$.\\

\noindent \textbf{(2)}, \textbf{(3)} Similar to (1).\\

\noindent \textbf{(4)} Firstly, we prove $t \vee \bot \sqsubseteq_{RS} t$.
Suppose $t \vee \bot \stackrel{\epsilon}{\Longrightarrow}_F| t'$.
By Lemma~\ref{L:Basic_I}(2), we have $t \vee \bot \stackrel{\tau}{\longrightarrow}_F t'' \stackrel{\epsilon}{\Longrightarrow}_F| t'$ with either $t''\equiv t$ or $t'' \equiv \bot$. Further, since $\bot \in F_{{CLL}}$ and $t'' \notin F_{CLL}$, we get $t'' \equiv t$.
    Thus, $t \stackrel{\epsilon}{\Longrightarrow}_F| t'$ and $t'\underset{\thicksim}{\sqsubset}_{RS} t'$.

Secondly, we prove $t \sqsubseteq_{RS} t \vee \bot$.
    Suppose  $t \stackrel{\epsilon}{\Longrightarrow}_F| t'$.
    By Lemma~\ref{L:F_NORMAL}(1), $t \vee \bot \notin F_{{CLL}}$ follows from $t \notin F_{{CLL}}$.
    Further, since $t \vee \bot \stackrel{\tau}{\longrightarrow}_{{CLL}} t \stackrel{\epsilon}{\Longrightarrow}_F| t'$,
    we obtain $t \vee \bot \stackrel{\epsilon}{\Longrightarrow}_F| t'$ and $t'\underset{\thicksim}{\sqsubset}_{RS} t'$.\\

\noindent \textbf{(5)} Straightforward.
\end{proof}

In the following, we will consider (in)equational laws about the operators $\Box$, $\parallel_A$ and $\wedge$.
For convenience, we adopt the convention below:

\begin{convention}\label{C:STABLE_TERM}
  When treating (in)equations $E(t_1,\dots,t_n)\approx_{RS}E'(t_1,\dots,t_n)$ (or, $E(t_1,\dots,t_n) \underset{\thicksim}{\sqsubset}_{RS} E'(t_1,\dots,t_n)$) in Prop. \ref{L:EC2}, \ref{L:CON_ID_I}, \ref{S:CONJUNCTION} and \ref{S:PARALLEL}, we assume that $t_i$ ranges over stable process terms for $1 \leq i \leq n$.
\end{convention}

The next group of equations is concerned with the operator $\Box$.

\begin{proposition}[External Choice]\hfill\label{L:EC2}
    \begin{enumerate}
      \item $t_1 \Box t_2 \approx_{RS} t_2 \Box t_1$.
      \item $t_1 \Box t_2 =_{RS} t_2 \Box t_1$.
      \item $(t_1 \Box t_2) \Box t_3 \approx_{RS} t_1 \Box (t_2 \Box t_3)$.
      \item $(t_1 \Box t_2) \Box t_3 =_{RS} t_1 \Box (t_2 \Box t_3)$.
      \item $t \Box t \approx_{RS} t$.
      \item $t \Box t =_{RS} t$.
      \item $t \Box \bot =_{RS} \bot$.
      \item $t \Box 0 \approx_{RS} t$.
      \item $t \Box 0 =_{RS} t$.
    \end{enumerate}
\end{proposition}
\begin{proof}
    We will prove (1), (2) and (7) one by one, the other parts of the proposition are handled in a similar way and omitted.

   \noindent \textbf{(1)}
    It is enough to prove $t_1 \Box t_2 \underset{\thicksim}{\sqsubset}_{RS} t_2 \Box t_1$. Put
    \[{\mathcal R}\triangleq\{(t_1 \Box t_2 , t_2 \Box t_1)\} \cup Id_S .\]

    We want to show that the relation $\mathcal R$ is a stable ready simulation relation. It suffices to prove that each pair in $\mathcal R$ satisfies (RS1)-(RS4). It is trivial for pairs in $Id_S$. In the following, we consider the pair $(t_1 \Box t_2 , t_2 \Box t_1)$. It immediately follows from Convention \ref{C:STABLE_TERM}, Lemma~\ref{L:STABILIZATION}, \ref{L:F_NORMAL}(3) and \ref{L:Basic_II}(1) that this pair satisfies (RS1), (RS2) and (RS4).

    \textbf{(RS3)} Suppose $t_1 \Box t_2 \stackrel{a}{\Longrightarrow}_F| t'$.
    Since $t_1 \Box t_2$ is stable, $t_1 \Box t_2 \stackrel{a}{\longrightarrow}_F t''\stackrel{\epsilon}{\Longrightarrow}_F| t'$ for some $t''$.
    Then, by Lemma~\ref{L:Basic_II}(1),  we get $t_2 \Box t_1 \stackrel{a}{\longrightarrow}_{CLL} t''$.
    Moreover, by Lemma~\ref{L:F_NORMAL}(3), $t_2 \Box t_1 \notin F_{CLL}$ comes from $t_1 \Box t_2 \notin F_{CLL}$.
    So, $t_2 \Box t_1 \stackrel{a}{\longrightarrow}_F t''\stackrel{\epsilon}{\Longrightarrow}_F| t'$ and $(t',t')\in {\mathcal R}$.\\

    \noindent\textbf{(2)} We shall prove $t_1 \Box t_2 \sqsubseteq_{RS} t_2 \Box t_1$.
    Suppose $t_1 \Box t_2 \stackrel{\epsilon}{\Longrightarrow}_F| t'$.
    It follows from Lemma~\ref{L:TAU_I}(2) that
    \[t_1 \stackrel{\epsilon}{\Longrightarrow}_F| t_1',  t_2 \stackrel{\epsilon}{\Longrightarrow}_F| t_2'\;\text{and}\; t'\equiv t_1' \Box t_2'\; \text{for some}\; t_1',t_2'.\]
    So, by Lemma~\ref{L:TAU_I}(1), we get $t_2 \Box t_1 \stackrel{\epsilon}{\Longrightarrow}_F| t_2' \Box t_1'$. Further, by item (1) in this lemma, we obtain $t_1 \Box t_2 \sqsubseteq_{RS} t_2 \Box t_1$.\\

    \noindent \textbf{(7)} Since $\bot \in F_{{CLL}}$, by Lemma~\ref{L:F_NORMAL}(3), we have  $t \Box \bot \in F_{{CLL}}$ for each $t\in T(\Sigma_{CLL})$. Then, $t \Box \bot =_{RS} \bot$ holds trivially.
\end{proof}

In order to treat (in)equations concerning the operator $\wedge$, the next two lemmas are needed.

\begin{lemma}\label{L:RS_CON}
If $t_1 \underset{\thicksim}{\sqsubset}_{RS} t_2$, $t_1 \underset{\thicksim}{\sqsubset}_{RS} t_3$ and $t_1 \notin F_{{CLL}}$, then $t_2 \wedge t_3 \notin F_{{CLL}}$.
\end{lemma}
\begin{proof}
    We prove it by induction on the degree of $t_1$.
    Suppose that it holds for all $s$ such that $|s|<|t_1|$.
    Assume that $t_1 \notin F_{{CLL}}$, $t_1 \underset{\thicksim}{\sqsubset}_{RS} t_2$ and $t_1 \underset{\thicksim}{\sqsubset}_{RS} t_3$.
    Thus, it follows at once that $t_2 \notin F_{{CLL}}$, $t_3 \notin F_{{CLL}}$ and ${\mathcal I}(t_1)={\mathcal I}(t_2)={\mathcal I}(t_3)$.
    Assume that $t_2 \wedge t_3 \in F_{{CLL}}$.
    We distinguish cases based on the form of the last rule applied in the proof tree of $Strip({\mathcal P}_{CLL},M_{{CLL}}) \vdash t_2 \wedge t_3F$.\\

\noindent Case 1 $\frac{t_2F}{t_2 \wedge t_3  F}$.

            Then, $ t_2F \in M_{{CLL}}$, which contradicts $t_2 \notin F_{{CLL}}$.\\

\noindent Case 2 $\frac{t_3F}{t_2 \wedge t_3  F}$.

            Analogous to Case 1.\\

\noindent Case 3 $\frac{t_2 \stackrel{a}{\longrightarrow} t_2'}{t_2 \wedge t_3  F}$ with $M_{{CLL}} \models \{t_3 \not\stackrel{a}{\longrightarrow} , t_2 \wedge t_3 \not\stackrel{\tau}{\longrightarrow}\}$.

            So, $ t_2 \stackrel{a}{\longrightarrow}_{CLL} t_2'$ and $ t_3 \not\stackrel{a}{\longrightarrow}_{CLL}$, which contradicts ${\mathcal I}(t_2)={\mathcal I}(t_3)$.\\

\noindent Case 4  $\frac{t_3 \stackrel{a}{\longrightarrow} t_3'}{t_2 \wedge t_3  F}$ with $M_{{CLL}} \models \{t_2 \not\stackrel{a}{\longrightarrow} , t_2 \wedge t_3 \not\stackrel{\tau}{\longrightarrow}\}$.

          Similar to Case 3.\\

\noindent Case 5 $\frac{t_2 \wedge t_3 \stackrel{\alpha}{\longrightarrow} t_4}{t_2 \wedge t_3  F}$ with $M_{{CLL}} \models t_2 \wedge t_3 \neg \overline{F}_{\alpha}$.

        Thus, $t_2 \wedge t_3 \stackrel{\alpha}{\longrightarrow} t_4\in M_{{CLL}}$.
        Since both $t_2$ and $t_3$ are stable, by Lemma~\ref{L:STABILIZATION}, so is $t_2 \wedge t_3$.
        Hence, $\alpha \neq \tau$.
        Then, by Lemma~\ref{L:Basic_II}(2), $\alpha \in {\mathcal I}(t_1)$ due to ${\mathcal I}(t_1)={\mathcal I}(t_2)={\mathcal I}(t_3)$.
        Further, by Lemma~\ref{L:LLTS_I}, it follows from $t_1\notin F_{CLL}$ that $t_1 \stackrel{\alpha}{\longrightarrow}_F t_1'$ for some $t_1'$.
        Then, by Lemma~\ref{L:LLTS_II}, we have $t_1' \stackrel{\epsilon}{\Longrightarrow}_F| t_1''$ for some $t_1''$.
        So, $t_1 \stackrel{\alpha}{\Longrightarrow}_F| t_1''$.
        Hence, it immediately follows from   $t_1 \underset{\thicksim}{\sqsubset}_{RS} t_2$ and $t_1 \underset{\thicksim}{\sqsubset}_{RS} t_3$ that
        \[t_2 \stackrel{\alpha}{\Longrightarrow}_F| t_2''\;\text{and}\; t_1'' \underset{\thicksim}{\sqsubset}_{RS} t_2''\;\text{for some}\; t_2'', \tag{\ref{L:RS_CON}.1}\]
        \[t_3 \stackrel{\alpha}{\Longrightarrow}_F| t_3''\;\text{and}\; t_1'' \underset{\thicksim}{\sqsubset}_{RS} t_3''\;\text{for some}\; t_3''. \tag{\ref{L:RS_CON}.2}\]
        Since $t_2$ and $t_3$ are stable, it follows from (\ref{L:RS_CON}.1) and (\ref{L:RS_CON}.2) that, for $i=2,3$, $t_i \stackrel{\alpha}{\longrightarrow}_F t_i' \stackrel{\epsilon}{\Longrightarrow}_F|t_i''$ for some $t_i'$.
        So, by Lemma~\ref{L:Basic_II}(2), $t_2 \wedge t_3  \stackrel{\alpha}{\longrightarrow}_{CLL} t_2'\wedge t_3' $.
        Further, by Lemma~\ref{L:CON_LLTS} and $M_{{CLL}} \models t_2 \wedge t_3 \neg \overline{F}_{\alpha}$, we have
        \[t_2' \wedge t_3' \in F_{{CLL}}.\tag{\ref{L:RS_CON}.3}\]
        On the other hand, since $t_1'' \notin F_{{CLL}}$ and $|t_1''|<|t_1|$, by IH, (\ref{L:RS_CON}.1) and (\ref{L:RS_CON}.2), we obtain $t_2'' \wedge t_3'' \notin F_{{CLL}}$.
        Further,  by Lemma~\ref{L:TAU_I}(1), it follows from  $t_2' \stackrel{\epsilon}{\Longrightarrow}_F|t_2''$ and $t_3' \stackrel{\epsilon}{\Longrightarrow}_F|t_3''$  that $t_2' \wedge t_3' \stackrel{\epsilon}{\Longrightarrow}_F|t_2'' \wedge t_3''$, which contradicts (\ref{L:RS_CON}.3).
\end{proof}

\begin{lemma}\label{L:CON_ID_I}\hfill
     \begin{enumerate}
      \item $t_1  \wedge t_2 \underset{\thicksim}{\sqsubset}_{RS} t_i$ for $i=1,2$.
      \item $t_1 \wedge t_2 \sqsubseteq_{RS} t_i$ for $i=1,2$.
      \item If $t_1 \underset{\thicksim}{\sqsubset}_{RS} t_2 $ and $t_1 \underset{\thicksim}{\sqsubset}_{RS} t_3$, then $t_1  \underset{\thicksim}{\sqsubset}_{RS} t_2 \wedge t_3$.
      \item If $t_1 \sqsubseteq_{RS} t_2 $ and $t_1 \sqsubseteq_{RS} t_3$, then $t_1  \sqsubseteq_{RS} t_2 \wedge t_3$.
    \end{enumerate}
\end{lemma}
\begin{proof}
\textbf{(1)} We will prove $t_1  \wedge t_2 \underset{\thicksim}{\sqsubset}_{RS} t_1$. The proof of $t_1  \wedge t_2 \underset{\thicksim}{\sqsubset}_{RS} t_2$ is similar and omitted.
    Put
     \[ {\mathcal R}\triangleq\{(s \wedge t , s)| \;s\;\text{and}\;t \;\text{are stable}\}.\]
     By Convention~\ref{C:STABLE_TERM}, it is enough to prove that $\mathcal R$ is a stable ready simulation relation. Let $(t_1 \wedge t_2 , t_1) \in {\mathcal R}$.
    By Lemma~\ref{L:STABILIZATION} and \ref{L:F_NORMAL}(4), it immediately follows that the pair $(t_1 \wedge t_2 , t_1)$ satisfies (RS1) and (RS2). We will prove (RS3) and (RS4) below.

    \textbf{(RS3)} Suppose $t_1 \wedge t_2 \stackrel{a}{\Longrightarrow}_F| t_{12}$. Since $t_1 \wedge t_2$ is stable, we get $t_1 \wedge t_2 \stackrel{a}{\longrightarrow}_F t_{12}' \stackrel{\epsilon}{\Longrightarrow}_F|t_{12}$ for some $t_{12}'$.
    So, it follows from Lemma~\ref{L:F_NORMAL}(4) and \ref{L:Basic_II}(2) that
    \[t_1  \stackrel{a}{\longrightarrow}_F t_1', t_2 \stackrel{a}{\longrightarrow}_F  t_2'\;\text{and}\; t_{12}' \equiv t_1' \wedge t_2'\;\text{for some}\;t_1',t_2'.\]
    Since $t_1' \wedge t_2' \stackrel{\epsilon}{\Longrightarrow}_F|t_{12}$, by Lemma~\ref{L:TAU_I}(2), we obtain $t_1' \stackrel{\epsilon}{\Longrightarrow}_F|t_1''$, $t_2' \stackrel{\epsilon}{\Longrightarrow}_F|t_2''$ and  $t_{12} \equiv t_1'' \wedge t_2''$ for some $t_1'',t_2''$. Hence, $t_1 \stackrel{a}{\longrightarrow}_F t_1' \stackrel{\epsilon}{\Longrightarrow}_F|t_1''$ and $(t_{12},t_1'') \in {\mathcal R}$.

    \textbf{(RS4)} Suppose $t_1 \wedge t_2 \notin F_{{CLL}}$. We want to prove ${\mathcal I}(t_1 \wedge t_2) = {\mathcal I}(t_1)$.
    It is straightforward that ${\mathcal I}(t_1 \wedge t_2) \subseteq {\mathcal I}(t_1)$.
    Next, we show ${\mathcal I}(t_1 \wedge t_2) \supseteq {\mathcal I}(t_1)$.
    Let $a \in {\mathcal I}(t_1)$.
    Thus, $t_1 \stackrel{a}{\longrightarrow}_{{CLL}}t_1'$ for some $t_1'$.
    Assume that $a \notin {\mathcal I}(t_2)$. So, $ \frac{t_1 \stackrel{a}{\longrightarrow}t_1'}{t_1 \wedge t_2F} \in Strip(R_{CLL},M_{{CLL}})$ comes from $M_{{CLL}} \models \{t_1 \wedge t_2 \not\stackrel{\tau}{\longrightarrow},t_2 \not\stackrel{a}{\longrightarrow}\}$.
    Then, $t_1 \wedge t_2F \in M_{{CLL}}$ immediately follows, a contradiction. \\

\noindent \textbf{(2)}
  Suppose $t_1 \wedge t_2 \stackrel{\epsilon}{\Longrightarrow}_F| t_{12}$. By Lemma~\ref{L:TAU_I}(2), there exist $t_1',t_2'$ such that  $t_1 \stackrel{\epsilon}{\Longrightarrow}_F| t_1'$, $t_2 \stackrel{\epsilon}{\Longrightarrow}_F| t_2'$ and $t_{12} \equiv t_1' \wedge t_2'$.
  By item (1) in this lemma, we have $ t_1'  \wedge t_2' \underset{\thicksim}{\sqsubset}_{RS} t_1'$. Hence, $t_1 \wedge t_2 \sqsubseteq_{RS} t_1$ holds. Similarly, $t_1 \wedge t_2 \sqsubseteq_{RS} t_2$ holds.\\

\noindent \textbf{(3)} Put
        \[{\mathcal R}\triangleq\{(s, r \wedge t )| s \underset{\thicksim}{\sqsubset}_{RS} r\;\text{and}\; s \underset{\thicksim}{\sqsubset}_{RS} t\}.\]
We intend to show that $\mathcal R$ is a stable ready simulation relation. Let $(t_1, t_2 \wedge t_3 ) \in {\mathcal R}$.
    By Lemma~\ref{L:STABILIZATION} and \ref{L:RS_CON}, it is easy to see that $(t_1, t_2 \wedge t_3 )$ satisfies both (RS1) and (RS2).

   \textbf{(RS3)} Suppose $t_1 \stackrel{a}{\Longrightarrow}_F| t_1'$. Since $t_1$ is stable, $t_1 \stackrel{a}{\longrightarrow}_F t_1'' \stackrel{\epsilon}{\Longrightarrow}_F| t_1'$ for some $t_1''$.
   It follows from $t_1 \underset{\thicksim}{\sqsubset}_{RS} t_2$ that $t_2 \stackrel{a}{\Longrightarrow}_F| t_2'$ and $t_1' \underset{\thicksim}{\sqsubset}_{RS} t_2'$ for some $t_2'$.
    Similarly, $t_3 \stackrel{a}{\Longrightarrow}_F| t_3'$ and $t_1' \underset{\thicksim}{\sqsubset}_{RS} t_3'$ for some $t_3'$.
     Further, by Lemma~\ref{L:RS_CON}, we get $t_2'\wedge t_3' \notin F_{{CLL}}$ because of $t_1' \notin F_{{CLL}}$.
   Moreover, since $t_2$ and $t_3$ are stable, $t_2 \stackrel{a}{\longrightarrow}_F t_2'' \stackrel{\epsilon}{\Longrightarrow}_F| t_2'$ and $t_3 \stackrel{a}{\longrightarrow}_F t_3'' \stackrel{\epsilon}{\Longrightarrow}_F| t_3'$ for some $t_2'',t_3''$.
   Then, it follows from Lemma~\ref{L:TAU_I}(1) that \[t_2'' \wedge t_3'' \stackrel{\epsilon}{\Longrightarrow}_F| t_2' \wedge t_3'.\tag{\ref{L:CON_ID_I}.3.1}\]
   Since  $t_2 \stackrel{a}{\longrightarrow}_{{CLL}} t_2''$ and $t_3 \stackrel{a}{\longrightarrow}_{{CLL}} t_3''$, by Lemma~\ref{L:Basic_II}(2), we get \[t_2 \wedge t_3 \stackrel{a}{\longrightarrow}_{{CLL}} t_2'' \wedge t_3''.\tag{\ref{L:CON_ID_I}.3.2}\]
    By Lemma~\ref{L:RS_CON}, $t_2 \wedge t_3 \notin F_{{CLL}}$ due to $t_1 \notin F_{{CLL}}$.
    Then, by (\ref{L:CON_ID_I}.3.1) and (\ref{L:CON_ID_I}.3.2), we obtain $t_2 \wedge t_3 \stackrel{a}{\longrightarrow}_F t_2''\wedge t_3'' \stackrel{\epsilon}{\Longrightarrow}_F| t_2' \wedge t_3'$ and $(t_1', t_2' \wedge t_3') \in {\mathcal R}$ .

   \textbf{(RS4)} Suppose $t_1 \notin F_{{CLL}}$.
    It follows from $t_1 \underset{\thicksim}{\sqsubset}_{RS} t_2$ and $t_1 \underset{\thicksim}{\sqsubset}_{RS} t_3$ that ${\mathcal I}(t_1) = {\mathcal I}(t_2) = {\mathcal I}(t_3)$.
    Further, since $t_1$, $t_2$ and $t_3$ are stable, we get ${\mathcal I}(t_i)\subseteq Act$ for $i=1,2,3$.
    Hence, by Lemma~\ref{L:Basic_II}(2), we have  ${\mathcal I}(t_1) = {\mathcal I}(t_2 \wedge t_3)$.\\

\noindent \textbf{(4)} Suppose $t_1 \stackrel{\epsilon}{\Longrightarrow}_F| t_1'$.
    It is enough to find $s$ such that $t_1' \underset{\thicksim}{\sqsubset}_{RS} s$ and $t_2\wedge t_3 \stackrel{\epsilon}{\Longrightarrow}_F|s$.
    Clearly, it follows from $t_1 \sqsubseteq_{RS} t_2 $ that $t_2 \stackrel{\epsilon}{\Longrightarrow}_F| t_2'$ and $t_1' \underset{\thicksim}{\sqsubset}_{RS} t_2'$ for some $t_2'$. Similarly, $t_3 \stackrel{\epsilon}{\Longrightarrow}_F| t_3'$ and $t_1' \underset{\thicksim}{\sqsubset}_{RS} t_3'$ for some $t_3'$.
    We will check that $t_2' \wedge t_3'$ is exactly what we need.
    By item~(3) in this lemma, we get $t_1'  \underset{\thicksim}{\sqsubset}_{RS} t_2' \wedge t_3'$.
    Further, $ t_2' \wedge t_3' \notin F_{{CLL}}$ comes from $t_1' \notin F_{{CLL}}$. Then,  by Lemma~\ref{L:TAU_I}(1), we obtain $t_2 \wedge t_3 \stackrel{\epsilon}{\Longrightarrow}_F| t_2' \wedge t_3'$, as desired.
\end{proof}

As an immediate consequence of items (2) and (4) in the above lemma, the property below is given, which has been obtained in \cite{Luttgen10}.
\[t_1 \sqsubseteq_{RS} t_2 \wedge t_3 \;\text{iff}\;t_1 \sqsubseteq_{RS} t_2\;\text{and}\; t_1 \sqsubseteq_{RS} t_3 .\]

As pointed out by L\"{u}ttgen and Vogler in \cite{Luttgen07,Luttgen10}, this is a fundamental property of ready simulation in the presence of logic operators. Intuitively, it says that $t_1$ is an implementation  of the specification $t_2\wedge t_3$ if and only if $t_1$ implements both $t_2$ and $t_3$.

We now are ready to prove some basic properties of the operator $\wedge$.

\begin{proposition}[Conjunction]\label{S:CONJUNCTION}\hfill
    \begin{enumerate}
      \item $t_1  \wedge t_2 \thickapprox_{RS} t_2 \wedge t_1$.
      \item $t_1  \wedge t_2 =_{RS} t_2 \wedge t_1$.
      \item $(t_1 \wedge t_2) \wedge t_3 \thickapprox_{RS} t_1 \wedge (t_2 \wedge t_3)$.
      \item $(t_1 \wedge t_2) \wedge t_3 =_{RS} t_1 \wedge (t_2 \wedge t_3)$.
      \item $t \wedge t\thickapprox_{RS} t$.
      \item $t \wedge t =_{RS} t$.
      \item $t \wedge \bot =_{RS} \bot$.
    \end{enumerate}
\end{proposition}
\begin{proof}
  \noindent \textbf{(1)} By Lemma~\ref{L:CON_ID_I}(1), $t_1  \wedge t_2 \underset{\thicksim}{\sqsubset}_{RS} t_1$ and $t_1  \wedge t_2 \underset{\thicksim}{\sqsubset}_{RS} t_2$.
  Then, by Lemma~\ref{L:CON_ID_I}(3), $t_1  \wedge t_2 \underset{\thicksim}{\sqsubset}_{RS} t_2 \wedge t_1$. Similarly, $t_2  \wedge t_1 \underset{\thicksim}{\sqsubset}_{RS} t_1 \wedge t_2$.\\

\noindent \textbf{(2)} Similar to (1), it follows from Lemma~\ref{L:CON_ID_I}(2) and \ref{L:CON_ID_I}(4).\\

\noindent \textbf{(3)} By Lemma~\ref{L:CON_ID_I}(1), $(t_1  \wedge t_2) \wedge t_3 \underset{\thicksim}{\sqsubset}_{RS} t_i $ for $i=1,2,3$.
Then, it immediately follows from Lemma~\ref{L:CON_ID_I}(3) that $(t_1  \wedge t_2) \wedge t_3 \underset{\thicksim}{\sqsubset}_{RS} t_1 \wedge (t_2 \wedge t_3)$.
Similarly, $t_1  \wedge (t_2 \wedge t_3) \underset{\thicksim}{\sqsubset}_{RS} (t_1 \wedge t_2) \wedge t_3$. \\

\noindent \textbf{(4)}
      Similar to (3), but using Lemma~\ref{L:CON_ID_I}(2)(4) instead of Lemma~\ref{L:CON_ID_I}(1)(3).\\

\noindent \textbf{(5)} Immediately follows from Lemma~\ref{L:CON_ID_I}(1)(3).\\

\noindent \textbf{(6)} Immediate consequence of Lemma~\ref{L:CON_ID_I}(2)(4).\\

\noindent \textbf{(7)} By Lemma~\ref{L:F_NORMAL}(4), $t \wedge \bot \in F_{{CLL}}$ follows from $\bot \in F_{{CLL}}$. Then, it holds trivially that $t \wedge \bot =_{RS} \bot$.
\end{proof}

The next lemma records some simple properties of the operator $\parallel_A$.

\begin{proposition}[Parallel]\label{S:PARALLEL}\hfill
    \begin{enumerate}
      \item $t_1 \parallel_A t_2 \thickapprox_{RS} t_2 \parallel_A t_1$.
      \item $t_1 \parallel_A t_2 =_{RS} t_2 \parallel_A t_1$.
      \item $t \parallel_A \bot =_{RS} \bot$.
      \item $0 \parallel_A 0 =_{RS} 0$.
    \end{enumerate}
\end{proposition}
\begin{proof}
  Analogous to Prop.~\ref{L:EC2}.
\end{proof}

In the following, we will show that $\sqsubseteq_{RS}$ is a precongruence.
We first prove that the operators $\Box$, $\parallel_A$ and $\wedge$ are monotonic w.r.t $\underset{\thicksim}{\sqsubset}_{RS}$.

\begin{lemma}[Monotonic w.r.t $\underset{\thicksim}{\sqsubset}_{RS}$]\label{L:CONGRUENCE_STABLE}
If $t_1 \underset{\thicksim}{\sqsubset}_{RS} t_2$ and $t_3$ is stable, then
\begin{enumerate}
  \item $t_1 \Box t_3 \underset{\thicksim}{\sqsubset}_{RS} t_2 \Box t_3$,
  \item $t_1 \wedge t_3 \underset{\thicksim}{\sqsubset}_{RS} t_2\wedge t_3$,
  \item $t_1\parallel_A t_3 \underset{\thicksim}{\sqsubset}_{RS} t_2\parallel_A t_3$.
\end{enumerate}
\end{lemma}
\begin{proof}
\textbf{(1)} Assume that $t_1 \underset{\thicksim}{\sqsubset}_{RS} t_2$ and $t_3$ is stable. Put
\[{\mathcal R}\triangleq \{(t_1 \Box t_3 , t_2 \Box t_3 )\} \cup \underset{\thicksim}{\sqsubset}_{RS}.\]
It suffices to show that $\mathcal R$ is a stable ready simulation relation.
Clearly, by Lemma~\ref{L:STABILIZATION}, both $t_1 \Box t_2$ and $t_2 \Box t_3$ are stable, that is, (RS1) holds. In the following, we show that the pair $(t_1\Box t_3, t_2\Box t_3)$ satisfies (RS2)-(RS4).

\textbf{(RS2)} Suppose  $t_2 \Box t_3 \in F_{{CLL}}$.
    By Lemma~\ref{L:F_NORMAL}(3), $t_2 \in F_{{CLL}}$ or $t_3 \in F_{{CLL}}$.
    For the first alternative, it follows from $t_1 \underset{\thicksim}{\sqsubset}_{RS} t_2$ that $t_1 \in F_{{CLL}}$. So, $t_1 \Box t_3  \in F_{{CLL}}$.
    For the second alternative, by Lemma~\ref{L:F_NORMAL}(3), $t_1 \Box t_3  \in F_{{CLL}}$ immediately follows.

\textbf{(RS3)} Suppose $t_1 \Box t_3 \stackrel{a}{\Longrightarrow}_F|t_{13}$.
    Since $t_1\Box t_3$ is stable, it follows that
    \[t_1 \Box t_3 \stackrel{a}{\longrightarrow}_F t_{13}' \stackrel{\epsilon}{\Longrightarrow}_F|t_{13}\;\text{for some}\; t_{13}'.\tag{\ref{L:CONGRUENCE_STABLE}.1.1}\]
    Since $t_1 \Box t_3 \notin F_{{CLL}}$, by (RS2), we get
    \[t_2 \Box t_3 \notin F_{{CLL}}.\tag{\ref{L:CONGRUENCE_STABLE}.1.2}\]
     The argument now splits into two cases depending on  the form of the last rule applied in the proof tree of $Strip({\mathcal P}_{CLL},M_{{CLL}}) \vdash t_1 \Box t_3 \stackrel{a}{\longrightarrow}t_{13}'$.\\

\noindent  Case 1 $\frac{t_1 \stackrel{a}{\longrightarrow} t_1'}{t_1 \Box t_3 \stackrel{a}{\longrightarrow} t_1'}$ with $M_{{CLL}} \models t_3 \not\stackrel{\tau}{\longrightarrow}$.

        So, $t_{13}' \equiv t_1'$ and $t_1 \stackrel{a}{\longrightarrow} t_1' \in M_{{CLL}}$.
        Then, by Lemma~\ref{L:F_NORMAL}(3) and (\ref{L:CONGRUENCE_STABLE}.1.1), we have $t_1 \stackrel{a}{\longrightarrow}_F t_1' \stackrel{\epsilon}{\Longrightarrow}_F|t_{13}$.
        It follows from $t_1 \underset{\thicksim}{\sqsubset}_{RS} t_2$ that $t_2  \stackrel{a}{\Longrightarrow}_F|t_2'$ and $t_{13} \underset{\thicksim}{\sqsubset}_{RS} t_2'$ for some $t_2'$.
        Moreover, since $t_2$ is stable, we get
        \[t_2 \stackrel{a}{\longrightarrow}_F t_2'' \stackrel{\epsilon}{\Longrightarrow}_F|t_2'\;\text{for some}\; t_2''. \tag{\ref{L:CONGRUENCE_STABLE}.1.3}\]
        Further, it follows from $M_{{CLL}} \models t_3 \not\stackrel{\tau}{\longrightarrow}$ that $\frac{t_2 \stackrel{a}{\longrightarrow} t_2''}{t_2 \Box t_3 \stackrel{a}{\longrightarrow}t_2''} \in Strip(R_{CLL},M_{{CLL}})$.
        Thus,
        \[t_2 \Box t_3 \stackrel{a}{\longrightarrow}t_2'' \in M_{{CLL}}.\tag{\ref{L:CONGRUENCE_STABLE}.1.4}\]
        Hence, by (\ref{L:CONGRUENCE_STABLE}.1.2), (\ref{L:CONGRUENCE_STABLE}.1.3) and (\ref{L:CONGRUENCE_STABLE}.1.4), we have  $t_2 \Box t_3 \stackrel{a}{\longrightarrow}_F t_2'' \stackrel{\epsilon}{\Longrightarrow}_F|t_2'$ and $(t_{13},t_2') \in {\mathcal R}$.\\

\noindent  Case 2 $\frac{t_3 \stackrel{a}{\longrightarrow} t_3'}{t_1 \Box t_3 \stackrel{a}{\longrightarrow} t_3'}$ with $M_{{CLL}} \models t_1 \not\stackrel{\tau}{\longrightarrow}$.

    Hence, $t_{13}' \equiv t_3'$ and $t_3 \stackrel{a}{\longrightarrow} t_3' \in M_{{CLL}}$.
   Since $t_2$ is stable, it follows that $\frac{t_3 \stackrel{a}{\longrightarrow} t_3'}{t_2 \Box t_3 \stackrel{a}{\longrightarrow}t_3'} \in Strip(R_{CLL},M_{{CLL}})$.
   So, $t_2 \Box t_3 \stackrel{a}{\longrightarrow} t_3' \in M_{{CLL}} $.
   Then, by (\ref{L:CONGRUENCE_STABLE}.1.1) and  (\ref{L:CONGRUENCE_STABLE}.1.2), we have $t_2 \Box t_3 \stackrel{a}{\longrightarrow}_F t_3' \stackrel{\epsilon}{\Longrightarrow}_F|t_{13}$ and $(t_{13},t_{13}) \in {\mathcal R}$,  as desired.

\textbf{(RS4)} Suppose $t_1 \Box t_3 \notin F_{{CLL}}$.
 So, by Lemma~\ref{L:F_NORMAL}(3), $t_1 \notin F_{{CLL}}$.
 Hence, ${\mathcal I}(t_1)={\mathcal I}(t_2) $ due to $t_1 \underset{\thicksim}{\sqsubset}_{RS} t_2$.
Further, by Lemma~\ref{L:Basic_II}(1), we get ${\mathcal I}(t_1 \Box t_3)={\mathcal I}(t_2 \Box t_3) $.\\

\noindent \textbf{(2)}
    $  t_1 \wedge t_3 \underset{\sim}{\sqsubset}_{RS} t_1 \;\text{and}\; t_1 \wedge t_3 \underset{\sim}{\sqsubset}_{RS} t_3 \; \;\;\;\;\;  \text{(by Lemma~\ref{L:CON_ID_I}(1))} $

    $  \Rightarrow t_1 \wedge t_3 \underset{\sim}{\sqsubset}_{RS} t_2\;\text{ and}\;  t_1 \wedge t_3 \underset{\sim}{\sqsubset}_{RS} t_3\;  \text{(by}\; t_1 \underset{\sim}{\sqsubset}_{RS} t_2\;\text{ and}\; \underset{\sim}{\sqsubset}_{RS}\;\text{is transitive)} $

    $  \Rightarrow t_1 \wedge t_3 \underset{\sim}{\sqsubset}_{RS} t_2 \wedge t_3\;  \qquad \qquad \qquad\;\; \text{(by Lemma~\ref{L:CON_ID_I}(3))} $\\

\noindent \textbf{(3)}
        Put \[{\mathcal R}\triangleq \{(s \parallel_A t, r \parallel_A t)\;|\; s \underset{\thicksim}{\sqsubset}_{RS} r\;\text{and}\;t\;\text{is stable}\}.\]
We want to show that $\mathcal R$ is a stable ready simulation relation. Let $(t_1 \parallel_A t_3, t_2\parallel_A t_3) \in {\mathcal R}$.

\textbf{(RS1)} Since $t_1$, $t_2$ and $t_3$ are stable, by Lemma~\ref{L:STABILIZATION}, so are $t_1 \parallel_A t_3$ and $t_2\parallel_A t_3$.

\textbf{(RS2)} Suppose $ t_2\parallel_A t_3  \in F_{{CLL}}$.
    So, by Lemma~\ref{L:F_NORMAL}(3), $t_2 \in F_{{CLL}}$ or $t_3\in F_{{CLL}}$.
    Then, by  $t_1 \underset{\thicksim}{\sqsubset}_{RS} t_2$ and Lemma~\ref{L:F_NORMAL}(3), it immediately follows that $t_1 \parallel_A t_3  \in F_{{CLL}}$.

\textbf{(RS3)} Suppose $ t_1 \parallel_A t_3  \stackrel{a}{\Longrightarrow}_F|t_4$.
    Since $t_1 \parallel_A t_3 $ is stable, $t_1 \parallel_A t_3 \stackrel{a}{\longrightarrow}_F t_5 \stackrel{\epsilon}{\Longrightarrow}_F|t_4$ for some $t_5$.
    We consider three cases based on the form of the last rule applied in the proof tree of $Strip({\mathcal P}_{CLL},M_{{CLL}}) \vdash t_1 \parallel_A t_3 \stackrel{a}{\longrightarrow} t_5$.\\

 \noindent   Case 1 $\frac{t_1 \stackrel{a}{\longrightarrow} t_1'}{t_1 \parallel_A t_3 \stackrel{a}{\longrightarrow} t_1'\parallel_A t_3}$ with $a\notin A$ and $M_{{CLL}} \models t_3 \not\stackrel{\tau}{\longrightarrow}$.

    So, $t_5 \equiv t_1'\parallel_A t_3$ and $t_1 \stackrel{a}{\longrightarrow} t_1' \in M_{{CLL}}$.
    By Lemma~\ref{L:TAU_I}(2), it follows from $t_1'\parallel_A t_3 \equiv t_5 \stackrel{\epsilon}{\Longrightarrow}_F|t_4$ that     $t_1' \stackrel{\epsilon}{\Longrightarrow}_F|t_1''$, $t_3 \stackrel{\epsilon}{\Longrightarrow}_F|t_3'$ and $t_4 \equiv t_1''\parallel_A t_3'$ for some $t_1'',t_3'$.
    Since $t_3$ is stable,  $t_3' \equiv t_3$.
    It is easy to see that $t_1,t_3,t_1' \notin F_{{CLL}}$.
    So, $t_1 \stackrel{a}{\longrightarrow}_F t_1' \stackrel{\epsilon}{\Longrightarrow}_F|t_1''$.
    Further, since $t_1 \underset{\thicksim}{\sqsubset}_{RS} t_2$, we obtain $t_2 \notin  F_{{CLL}}$,  $t_2 \stackrel{a}{\Longrightarrow}_F|t_2''$  and $t_1'' \underset{\thicksim}{\sqsubset}_{RS} t_2''$ for some $t_2''$.
    Since $t_2$ is stable, there exists $t_2'$ such that $t_2 \stackrel{a}{\longrightarrow}_F t_2' \stackrel{\epsilon}{\Longrightarrow}_F|t_2''$,
    further, by Lemma~\ref{L:TAU_I}(1) and $t_3 \stackrel{\epsilon}{\Longrightarrow}_F|t_3$, it follows that
    \[t_2' \parallel_A t_3\stackrel{\epsilon}{\Longrightarrow}_F|t_2'' \parallel_A t_3.\tag{\ref{L:CONGRUENCE_STABLE}.3.1}\]
    Moreover, since $t_3$ is stable and $a\notin A$, we obtain
    \[\frac{t_2 \stackrel{a}{\longrightarrow} t_2'}{t_2 \parallel_A t_3 \stackrel{a}{\longrightarrow} t_2'\parallel_A t_3} \in Strip(R_{CLL},M_{{CLL}}).\]
    Hence
    \[t_2 \parallel_A t_3 \stackrel{a}{\longrightarrow} t_2'\parallel_A t_3 \in M_{{CLL}}. \tag{\ref{L:CONGRUENCE_STABLE}.3.2}\]
    Further, since $t_2 \notin  F_{{CLL}}$ and $ t_3 \notin  F_{{CLL}}$, by Lemma~\ref{L:F_NORMAL}(3), we get $t_2 \parallel_A t_3 \notin  F_{{CLL}}$.
    Thus, it follows from (\ref{L:CONGRUENCE_STABLE}.3.1) and (\ref{L:CONGRUENCE_STABLE}.3.2) that
    $t_2 \parallel_A t_3 \stackrel{a}{\Longrightarrow}_F| t_2'' \parallel_A t_3$.
    Moreover, $(t_1 ''\parallel_A t_3,t_2''\parallel_A t_3) \in {\mathcal R}$ due to $t_1'' \underset{\thicksim}{\sqsubset}_{RS} t_2''$.\\

\noindent Case 2 $\frac{t_3 \stackrel{a}{\longrightarrow} t_3'}{t_1 \parallel_A t_3 \stackrel{a}{\longrightarrow} t_1\parallel_A t_3'}$ with $M_{{CLL}} \models t_1 \not\stackrel{\tau}{\longrightarrow}$ and $a\notin A$.

 Similar to Case 1.\\

\noindent Case 3 $\frac{t_1 \stackrel{a}{\longrightarrow} t_1',t_3 \stackrel{a}{\longrightarrow} t_3'}{t_1 \parallel_A t_3 \stackrel{a}{\longrightarrow} t_1'\parallel_A t_3'}$ with $a\in A$.

We invite the reader to check it.\\

\textbf{(RS4)}  Suppose $t_1 \parallel_A t_3 \notin F_{{CLL}}$.
    We will prove ${\mathcal I}(t_1 \parallel_A t_3 ) \subseteq {\mathcal I}(t_2 \parallel_A t_3 ) $.
    Assume that $a \in {\mathcal I}(t_1 \parallel_A t_3 )$. Then $t_1\parallel_A t_3 \stackrel{a}{\longrightarrow}_{{CLL}} t_4$ for some $t_4$. In the following, we consider three cases based on the form of the last rule applied in the proof tree of  $Strip({\mathcal P}_{CLL},M_{{CLL}}) \vdash t_1\parallel_A t_3 \stackrel{a}{\longrightarrow} t_4$.\\

\noindent    Case 1 $\frac{t_1 \stackrel{a}{\longrightarrow} t_1'}{t_1 \parallel_A t_3 \stackrel{a}{\longrightarrow} t_1'\parallel_A t_3}$ with $a\notin A$ and $M_{{CLL}} \models t_3 \not\stackrel{\tau}{\longrightarrow}$.

    So, $t_4 \equiv t_1'\parallel_A t_3$ and $t_1 \stackrel{a}{\longrightarrow} t_1' \in M_{{CLL}}$. Then, $ a \in {\mathcal I}(t_1 )$.
    Since $t_1 \parallel_A t_3 \notin F_{{CLL}}$, by Lemma~\ref{L:F_NORMAL}(3), we have $t_1 \notin F_{{CLL}}$.
    So, ${\mathcal I}(t_1 ) = {\mathcal I}(t_2 ) $ comes from $t_1 \underset{\thicksim}{\sqsubset}_{RS} t_2$.
    Then, $t_2 \stackrel{a}{\longrightarrow}_{{CLL}} t_2'$ for some $t_2'$.
    Since $a\notin A$ and $M_{{CLL}} \models t_3 \not\stackrel{\tau}{\longrightarrow}$, we get $\frac{t_2 \stackrel{a}{\longrightarrow} t_2'}{t_2 \parallel_A t_3 \stackrel{a}{\longrightarrow} t_2'\parallel_A t_3} \in Strip(R_{CLL},M_{{CLL}})$.
    So, $t_2 \parallel_A t_3 \stackrel{a}{\longrightarrow} t_2'\parallel_A t_3 \in M_{{CLL}}$.
    Hence, $a \in {\mathcal I}(t_2 \parallel_A t_3 )$.\\

\noindent    Case 2 $\frac{t_3 \stackrel{a}{\longrightarrow} t_3'}{t_1 \parallel_A t_3 \stackrel{a}{\longrightarrow}t_1 \parallel_A t_3'}$ with $a\notin A$ and $M_{{CLL}} \models t_1 \not\stackrel{\tau}{\longrightarrow}$.

Similar to Case 1.\\

\noindent    Case 3 $\frac{t_1 \stackrel{a}{\longrightarrow} t_1',t_3 \stackrel{a}{\longrightarrow} t_3'}{t_1 \parallel_A t_3 \stackrel{a}{\longrightarrow} t_1' \parallel_A t_3'}$ with $a \in A$.

    Thus, $t_1 \stackrel{a}{\longrightarrow} t_1'\in M_{CLL}$ and $t_3 \stackrel{a}{\longrightarrow} t_3' \in M_{CLL}$.
    Similar to Case 1, we get $t_2 \stackrel{a}{\longrightarrow}_{CLL} t_2'$ for some $t_2'$.
    Further, since $a\in A$, $t_2 \parallel_A t_3 \stackrel{a}{\longrightarrow} t_2'\parallel_A t_3' \in M_{CLL}$ follows from $\frac{t_2 \stackrel{a}{\longrightarrow} t_2',t_3 \stackrel{a}{\longrightarrow} t_3'}{t_2 \parallel_A t_3 \stackrel{a}{\longrightarrow} t_2' \parallel_A t_3'} \in Strip(R_{CLL}, M_{CLL})$. Thus, $a \in {\mathcal I}(t_2 \parallel_A t_3 )$.\\

    Similarly, we can show ${\mathcal I}(t_1 \parallel_A t_3 ) \supseteq {\mathcal I}(t_2 \parallel_A t_3 )$.
\end{proof}

\begin{lemma}[Monotonic w.r.t $\sqsubseteq_{RS}$]\label{L:CONGRUENCE_PRE}
For any $\odot \in \{\Box,\wedge,\parallel_A,\vee\}$, if $t_1 \sqsubseteq_{RS} t_2$ then $t_1 \odot t_3 \sqsubseteq_{RS} t_2 \odot t_3$ for each $t_3 \in T(\Sigma_{CLL})$.
\end{lemma}
\begin{proof}
  For $\odot = \wedge$, it immediately follows from Lemma~\ref{L:CON_ID_I}(2)(4).
  Next, we consider $\odot = \parallel_A$, the remaining parts raise no significantly different issues and omitted.
    Suppose  $t_1 \parallel_A t_3 \stackrel{\epsilon}{\Longrightarrow}_F|t_{13}$.
    It is enough to find $s$ such that $t_2 \parallel_A t_3 \stackrel{\epsilon}{\Longrightarrow}_F|s$ and $t_{13} \underset{\thicksim}{\sqsubset}_{RS}s$.
    By Lemma~\ref{L:TAU_I}(2), we have $t_1 \stackrel{\epsilon}{\Longrightarrow}_F|t_1'$, $t_3 \stackrel{\epsilon}{\Longrightarrow}_F|t_3'$ and $t_{13}\equiv t_1'\parallel_A t_3'$ for some $t_1',t_3'$.
     Moreover, it follows from $t_1 \sqsubseteq_{RS} t_2$ that $t_2 \stackrel{\epsilon}{\Longrightarrow}_F|t_2'$ and $t_1' \underset{\thicksim}{\sqsubset}_{RS} t_2'$ for some $t_2'$.
    We will check that $t_2' \parallel_A t_3'$ is exactly what we need.
    Clearly, by Lemma~\ref{L:CONGRUENCE_STABLE}(3), $t_1' \parallel_A t_3' \underset{\thicksim}{\sqsubset}_{RS} t_2'\parallel_A t_3'$.
    And, by Lemma~\ref{L:TAU_I}(1), we get $t_2 \parallel_A t_3 \stackrel{\epsilon}{\Longrightarrow}_F|t_2'\parallel_A t_3'$, as desired.
\end{proof}

We can now state an important result about ready simulation relation, which reveals that $\sqsubseteq_{RS}$ is substitutive under all our combinators.

\begin{proposition}[Precongruence]\label{L:CONGRUENCE}
  If $t_1 \sqsubseteq_{RS} t_2$ and $s_1 \sqsubseteq_{RS} s_2$, then
  \begin{enumerate}
    \item   $\alpha.t_1 \sqsubseteq_{RS} \alpha.t_2$ for each $\alpha \in Act_{\tau}$ and
    \item $t_1 \odot s_1 \sqsubseteq_{RS} t_2 \odot s_2$ for each $\odot \in \{\Box, \wedge, \parallel_A, \vee\}$.
  \end{enumerate}
\end{proposition}
\begin{proof}
\noindent \textbf{(1)} The argument splits into two cases depending on the kind of $\alpha$ (visible or invisible). The proof is straightforward and omitted.
\\

\noindent \textbf{(2)} By Prop.~\ref{S:DISJUNCTION}(1), \ref{L:EC2}(2), \ref{S:CONJUNCTION}(2)  and \ref{S:PARALLEL}(2), $\odot$ satisfies commutative law for each $\odot \in \{\Box,\wedge,\parallel_A, \vee\}$.
Then, it immediately follows from the transitivity of $\sqsubseteq_{RS}$ and Lemma~\ref{L:CONGRUENCE_PRE}.
\end{proof}

Hitherto we have only considered (in)equational laws involving one operator alone.
Next, we shall deal with a few of laws which refer to different operators in one (in)equation.
More such laws will be treated in the next section when establishing the soundness of $AX_{CLL}$.

\begin{proposition}[Distributive]\label{S:DISTRIBUTIVE}
       $t_1 \odot( t_2 \vee t_3) =_{RS} (t_1 \odot t_2) \vee (t_1 \odot t_3)$ for each  $\odot \in \{ \Box, \parallel_A,\wedge \}$.
\end{proposition}
\begin{proof}
    We consider the case where $\odot = \Box$, the others are similar.

    Firstly, we prove $t_1 \Box( t_2 \vee t_3) \sqsubseteq_{RS} (t_1 \Box t_2) \vee (t_1 \Box t_3)$.
    Suppose $t_1 \Box( t_2 \vee t_3) \stackrel{\epsilon}{\Longrightarrow}_F| t_4$.
    Then, by Lemma~\ref{L:DIS}, there exists $t_1'$ and a sequence $t_1 \Box ( t_2 \vee t_3) \equiv T_0 \stackrel{\tau}{\longrightarrow}_F,\dots,\stackrel{\tau}{\longrightarrow}_F T_n  \equiv t_4(n > 0)$ such that $t_1\stackrel{\epsilon}{\Longrightarrow}_F t_1'$, $T_j\equiv t_1' \Box( t_2 \vee t_3)$ and either $T_{j+1}\equiv t_1' \Box t_2 $ or $T_{j+1}\equiv t_1' \Box  t_3$ for some $j<n$.
    W.l.o.g, assume that $T_{j+1}\equiv t_1' \Box t_2$.
    So, by Lemma~\ref{L:STABILIZATION} and \ref{L:Basic_I}(2), $(t_1 \Box t_2) \vee (t_1 \Box t_3) \stackrel{\tau}{\longrightarrow}_{{CLL}} t_1 \Box t_2 \stackrel{\epsilon}{\Longrightarrow}_{{CLL}} t_1' \Box t_2 \stackrel{\epsilon}{\Longrightarrow}_F| t_4$.
    Further, by Lemma~\ref{L:F_NORMAL}(1)(3), $(t_1 \Box t_2) \vee (t_1 \Box t_3) \stackrel{\epsilon}{\Longrightarrow}_F| t_4$.
    Then, it follows from $t_4\underset{\thicksim}{\sqsubset}_{RS} t_4$ that $t_1 \Box( t_2 \vee t_3) \sqsubseteq_{RS} (t_1 \Box t_2) \vee (t_1 \Box t_3)$.

    Secondly, we prove $(t_1 \Box t_2) \vee (t_1 \Box t_3) \sqsubseteq_{RS} t_1 \Box(t_2 \vee t_3)$.
    By Prop.~\ref{S:DISJUNCTION}(1)(5), we have  $t_2 \sqsubseteq_{RS} t_2 \vee t_3$ and $t_3 \sqsubseteq_{RS} t_2 \vee t_3$.
    Further, by Prop.~\ref{L:CONGRUENCE}(2), we get $t_1 \Box t_2 \sqsubseteq_{RS} t_1 \Box (t_2 \vee t_3)$ and $t_1 \Box t_3 \sqsubseteq_{RS} t_1 \Box (t_2 \vee t_3)$. Then, $(t_1 \Box t_2 )\vee (t_1 \Box t_3) \sqsubseteq_{RS} t_1 \Box (t_2 \vee t_3)$ comes from Prop.~\ref{L:CONGRUENCE}(2) and \ref{S:DISJUNCTION}(3).
\end{proof}

\begin{proposition}\label{S:SPECIAL_I}
$\alpha.t_1 \Box \alpha.t_2 \sqsubseteq_{RS} \alpha.(t_1 \vee t_2) $ for each $\alpha \in Act_{\tau}$.
\end{proposition}
\begin{proof}
   $ t_1 \sqsubseteq_{RS} t_1 \vee t_2\;\text{and}\;t_2 \sqsubseteq_{RS} t_1 \vee t_2$  \qquad \;\;\;\; (by Prop.~\ref{S:DISJUNCTION}(1)(5))\\
 $\Rightarrow \alpha.t_1 \sqsubseteq_{RS} \alpha.(t_1 \vee t_2)\;\text{and}\;\alpha.t_2 \sqsubseteq_{RS} \alpha.(t_1 \vee t_2)$  (by Prop.~\ref{L:CONGRUENCE}(1))\\
 $\Rightarrow \alpha.t_1 \Box \alpha.t_2 \sqsubseteq_{RS} \alpha.(t_1 \vee t_2)$ \qquad \qquad \qquad \qquad(by Prop.~\ref{L:CONGRUENCE}(2) and \ref{L:EC2}(6)).
\end{proof}

A natural problem arises at this point, that is, whether the inequation below holds
\[\alpha.(t_1 \vee t_2) \sqsubseteq_{RS} \alpha.t_1 \Box \alpha.t_2. \tag{DS}\]
The answer is negative.
For instance, consider $t_1 \equiv \bot$ and $t_2 \equiv 0$.
By Lemma~\ref{L:F_NORMAL}, we have $a.(t_1\vee t_2)\equiv a.(\bot \vee 0)\notin F_{CLL}$ and $a.t_1 \Box a.t_2 \equiv a.\bot \Box a.0 \in F_{CLL}$.
Thus, it doesn't hold that $a.(\bot \vee 0) \sqsubseteq_{RS} a.\bot \Box a.0$.

We conclude this section with the next proposition, which establishes a necessary and sufficient condition for the inequation (DS) with $\alpha \in  Act$ to be true.
To this end, we introduce the notion below.

\begin{mydefn}[Uniform w.r.t $F_{CLL}$]
  Two process terms $t$ and $s$ are said to be uniform w.r.t $F_{CLL}$ if $t\in F_{CLL}$ iff $s \in F_{CLL}$.
\end{mydefn}

\begin{proposition}\label{S:SPECIAL}
For each $a \in Act$,
 $a.(t_1 \vee t_2) \sqsubseteq_{RS} a.t_1 \Box a.t_2 $ iff
 $t_1$ and $t_2$ are uniform w.r.t $F_{CLL}$.
\end{proposition}
\begin{proof}
\noindent (Left implies Right) Suppose $t_1$ and $t_2$ are not uniform w.r.t $F_{CLL}$. W.l.o.g, assume that $t_1 \in F_{CLL}$ and $t_2\notin F_{CLL}$. By Lemma~\ref{L:F_NORMAL}, we get $a.(t_1 \vee t_2) \notin F_{CLL}$ and $a.t_1 \Box a.t_2 \in F_{CLL}$. Hence, $a.(t_1 \vee t_2) \not\sqsubseteq_{RS} a.t_1 \Box a.t_2 $.\\

\noindent (Right implies Left)
Since  $a \in Act$, by Lemma~\ref{L:Basic_I}(1) and \ref{L:STABILIZATION}, both $a.t_1 \Box a.t_2 $ and $ a.(t_1 \vee t_2)$ are stable. So, it is enough to prove  $a.(t_1 \vee t_2) \underset{\thicksim}{\sqsubset}_{RS}  a.t_1 \Box a.t_2 $.
Put \[{\mathcal R}\triangleq\{(a.(t_1 \vee t_2), a.t_1 \Box a.t_2)\} \cup Id_S.\]
We need to show that $\mathcal R$ is a stable ready simulation relation.
It is trivial to check that (RS1)-(RS4) hold for each pair in $Id_S$.
In the following, we treat the pair $(a.(t_1 \vee t_2),a.t_1 \Box a.t_2)$.
By Lemma~\ref{L:Basic_I}(1), \ref{L:Basic_II}(1) and \ref{L:STABILIZATION}, this pair satisfies (RS1) and (RS4).

\textbf{(RS2)} Suppose $a.t_1 \Box a.t_2 \in F_{CLL}$. By Lemma~\ref{L:F_NORMAL}, it follows that $t_i \in F_{CLL}$ for some $i\in \{1,2\}$. Further, since $t_1$ and $t_2$ are uniform w.r.t $F_{CLL}$, we obtain $t_1 \in F_{CLL}$ and $t_2 \in F_{CLL}$. Hence, by Lemma~\ref{L:F_NORMAL} again, $a.(t_1 \vee t_2) \in F_{CLL}$.

\textbf{(RS3)} Suppose $a.(t_1 \vee t_2) \stackrel{b}{\Longrightarrow}_F|t'$. Since $(t',t')\in {\mathcal R}$, it suffices to prove that $a.t_1 \Box a.t_2 \stackrel{b}{\Longrightarrow}_F| t'$.
Since $a.(t_1 \vee t_2)$ is stable, $a.(t_1 \vee t_2) \stackrel{b}{\longrightarrow}_F t'' \stackrel{\epsilon}{\Longrightarrow}_F| t'$ for some $t''$.
So, by Lemma~\ref{L:Basic_I}(1), we get $b= a$ and $t'' \equiv t_1 \vee t_2$.
Further, by Lemma~\ref{L:Basic_I}(2), either $t_1 \vee t_2 \stackrel{\tau}{\longrightarrow}_F t_1 \stackrel{\epsilon}{\Longrightarrow}_F| t'$ or $t_1 \vee t_2 \stackrel{\tau}{\longrightarrow}_F t_2 \stackrel{\epsilon}{\Longrightarrow}_F| t'$.
W.l.o.g, we consider the first alternative.
Then, it immediately follows that $a.t_1 \Box a.t_2 \stackrel{a}{\longrightarrow}_{{CLL}} t_1  \stackrel{\epsilon}{\Longrightarrow}_F| t'$.
Moreover, since $a.(t_1 \vee t_2)\notin F_{CLL}$, by (RS2), we get $a.t_1 \Box a.t_2 \notin F_{CLL}$.
Hence, $a.t_1 \Box a.t_2 \stackrel{a}{\Longrightarrow}_F| t'$, as desired.
\end{proof}

Notice that the situation is different if $\alpha = \tau$.
In such case, the inequation (DS) does not always hold even if $t_1$ and $t_2$ are  uniform w.r.t $F_{CLL}$.
As a simple example, consider $t_1\equiv a.0$ and $t_2\equiv b.0$ with $a \neq b$.
Clearly, they are uniform w.r.t $F_{CLL}$.
Moreover, $\tau.(a.0\vee b.0)\stackrel{\epsilon}{\Longrightarrow}_F|a.0$,  and $a.0 \Box b.0$ is the unique process term such that $\tau.a.0\Box \tau.b.0 \stackrel{\epsilon}{\Longrightarrow}_F|a.0\Box b.0$.
Hence, it follows from $a.0 \notin F_{CLL}$ and ${\mathcal I}(a.0)\neq {\mathcal I}(a.0\Box b.0)$ that $a.0 \not \underset{\thicksim}{\sqsubset}_{RS} a.0\Box b.0$.
Thus, $\tau.(a.0\vee b.0) \not\sqsubseteq_{RS} \tau.a.0 \Box \tau.b.0$.

\section{Axiomatic System $AX_{CLL}$}
Section~5 has developed the behavioral theory of CLL on the level of semantics. This section will provide an algebraic and axiomatic approach to reason about behavior.
We will propose an axiomatic system to capture and investigate the operators in CLL through (in)equational laws, and establish its soundness and ground-completeness w.r.t ready simulation.

\subsection{$AX_{CLL}$}
In order to introduce the axiomatic system $AX_{CLL}$, a few preliminary definitions are given below.

\begin{mydefn}[Basic Process Term] \label{D:BPT}
The basic process terms are defined by BNF below
\[t::=0\;|\;(\alpha.t)\;|\;(t\vee t)\;|\;(t \Box t)\;|\;(t\parallel_A t)\]
where $\alpha \in Act_{\tau}$ and $A \subseteq Act$. We denote $T(\Sigma_B)$ as the set of all basic process terms.
\end{mydefn}

\begin{mydefn}\label{D:GEC}
Let $<t_0,t_1,\dots,t_{n-1}>$ be a finite sequence of process terms with $n \geq 0$. We define the general external choice $\underset{i<n}\square t_i$ by recursion:
\begin{enumerate}
  \item $\underset{i<0}\square t_i \triangleq 0$,
  \item $\underset{i<1}\square t_i \triangleq t_0$,
  \item $\underset{i<k+1}\square t_i \triangleq (\underset{i<k}\square t_i) \Box t_k$ for $k \geq 1$.
\end{enumerate}
\end{mydefn}

Moreover, given a finite sequence $<t_0,\dots,t_{n-1}>$ and $S \subseteq \{t_0,\dots,t_{n-1}\}$, the general external choice $\square S$ is defined as $\square S \triangleq \underset{j<|S|}\square t_j'$, where the sequence $<t_0',\dots,t_{|S|-1}'>$ is the restriction of $<t_0,\dots,t_{n-1}>$ to $S$.\\

In fact, up to $=_{RS}$ (or, =, see below), the order and grouping of terms in $\underset{i<n}{\square}t_i$ may be ignored by virtue of Prop.~\ref{L:EC2}(2)(4) (axioms $EC1$ and $EC2$ below, respectively).

\begin{mydefn}[Injective in Prefixes]
  A process term $\underset{i<n}{\square}\alpha_i.t_i$ is said to be injective in prefixes if $\alpha_i \neq \alpha_j$ for any $i \neq j < n$.
\end{mydefn}

We now can present the axiomatic system $AX_{CLL}$. As usual, $AX_{CLL}$ has two parts: axioms and inference rules.\\

\noindent \textbf{Axioms}

 Unless otherwise stated, we shall assume variables in axioms below to be in range of $T(\Sigma_{CLL})$. As usual, $t=t'$ means $t\leqslant t'$ and $t'\leqslant t$.\\

\noindent  \[  EC1   \;\;   x \Box y  = y \Box x \qquad\qquad\qquad\,
   DI1   \;\;     x \vee y  = y \vee x \qquad\qquad\qquad\qquad\qquad\qquad\;\;\]
  \[  EC2   \;\;   (x \Box y)\Box z  = x \Box (y \Box z) \qquad
   DI2   \;\;    x \vee (y \vee z ) = (x\vee y) \vee z \qquad\qquad\qquad\qquad\]
  \[  EC3   \;\;   x \Box x  = x \qquad\qquad\qquad\;\;\;\;\,
   DI3  \;\;     x \vee x  = x \qquad\qquad\qquad\qquad\qquad\qquad\qquad\;\]
  \[  EC4   \;\;   x \Box 0  = x \qquad\qquad\qquad\;\;\;\;\;
   DI4   \;\;     x \vee \bot  = x \qquad\qquad\qquad\qquad\qquad\qquad\qquad\]
  \[  EC5   \;\;   x \Box  \bot = \bot \qquad\qquad\qquad\;\;\;\mspace{1mu}
   DI5  \;\;     x  \leqslant x \vee y    \qquad\qquad\qquad\qquad\qquad\qquad\qquad\;\;\]
  \[  CO1   \;\;   x\wedge y  = y \wedge x \qquad\qquad\;\;\;\;\,
   DS1  \;\;     x \Box (y\vee z)  \leqslant (x \Box y) \vee (x \Box z)   \qquad\qquad\qquad\; \]
  \[  CO2   \;\;   (x\wedge y)\wedge z  = x \wedge (y \wedge z) \;\;\;\;
   DS2  \;\;     x \wedge (y\vee z)  \leqslant (x \wedge y) \vee (x \wedge z)  \qquad\qquad\qquad   \]
  \[  CO3   \;\;   x\wedge x  = x \qquad\qquad\qquad\;\;\;\;\;
   DS3  \;\;     x \parallel_A (y\vee z)  \leqslant (x \parallel_A y) \vee (x \parallel_A z)  \qquad\qquad\qquad    \]
  \[  CO4   \;\;   x\wedge \bot  = \bot \qquad\qquad\qquad\;\;
   DS4  \;\;     a.(x \vee y)  \leqslant a.x\Box a.y, \;\text{where}\;  x,y \in T(\Sigma_B).\;  \]
  \[ PR1  \;\;    a.\bot  = \bot  \qquad\qquad\qquad\;\;\;\;\;
   PA1   \;\;    x \parallel_A y  =y \parallel_A x \qquad\qquad\qquad\qquad\qquad\qquad\]
  \[ PR2  \;\;    \tau.x  = x  \qquad\qquad\qquad\;\;\;\;\;\;\:\:\:
   PA2  \;\;     x \parallel_A \bot  = \bot \qquad\qquad\qquad\qquad\qquad\qquad\;\;\;\;\;\]
  \noindent $ECC1$
   \[\underset{i< n}{\square}a_i.x_i\wedge \underset{j< m}{\square}b_j.y_j  = \bot \;\text{ if}\; \{a_i|i< n\}\neq \{b_j|j < m\}.\]

   \noindent $ECC2$
      \[\underset{i<n}{\square}a_i.(x_i \wedge y_i)  \leqslant \underset{i< n}{\square}a_i.x_i\wedge \underset{i< n}{\square}a_i.y_i\]

\noindent  $ECC3$
        \[\underset{i< n}{\square}a_i.x_i\wedge \underset{i< n}{\square}a_i.y_i  \leqslant
                \underset{i<n}{\square}a_i.(x_i \wedge y_i),
          \text{where}\;\underset{i<n}{\square}a_i.x_i\;\text{is injective in prefixes}.\]

\noindent $EXP1  $
    \begin{multline*}
      \underset{i< n}{\square}a_i.x_i \parallel_A \underset{j< m}{\square}b_j.y_j  \leqslant\\
      \left(\underset {\begin{subarray}
                   \;i< n,\\
                   a_i \notin A
                \end{subarray}}
                \square a_i.(x_i \parallel_A \underset{j< m}{\square}b_j.y_j) \Box
                \underset {\begin{subarray}
                   \;j< m,\\
                   b_j \notin A
                \end{subarray}}
                \square
                b_j.(\underset{i< n}{\square}a_i.x_i \parallel_A y_j)\right) \Box
        \underset {\begin{subarray}
                   \;i< n,j<m\\
                  a_i= b_j\in A
                \end{subarray}}
                \square
                a_i.(x_i \parallel_A y_j)
    \end{multline*}
     $EXP2  \qquad$
      \begin{multline*}
       \underset{i< n}{\square}a_i.x_i \parallel_A \underset{j< m}{\square}b_j.y_j \geqslant\\
      \left(\underset {\begin{subarray}
                   \;i< n,\\
                   a_i \notin A
                \end{subarray}}
                \square a_i.(x_i \parallel_A \underset{j< m}{\square}b_j.y_j) \Box
                \underset {\begin{subarray}
                   \;j< m,\\
                   b_j \notin A
                \end{subarray}}
                \square
                b_j.(\underset{i< n}{\square}a_i.x_i \parallel_A y_j)\right) \Box
       \underset {\begin{subarray}
                   \;i< n,j<m\\
                  a_i= b_j\in A
                \end{subarray}}
                \square
                a_i.(x_i \parallel_A y_j)
    \end{multline*}
    $\text{where}\;  x_i,y_j \in T(\Sigma_B)\; \text{for each}\; i< n\; \text{and}\; j< m$.\\

\noindent \textbf{Inference rules}
    \begin{align*}\hfill
    &\text{REF}  &       &\frac{-}{t\leqslant t}  \\
    &\text{TRANS}  &    &\frac{t\leqslant t',t'\leqslant t''}{t\leqslant t''}  \\
    &\text{CONTEXT}  &     & \text{for each n-ary function symbol } \;f \in \Sigma_{CLL}\\
    & & & \frac{t_1 \leqslant t_1',\dots, t_n \leqslant t_n'}{f(t_1, \dots, t_n) \leqslant f(t_1',\dots,t_n')}
    \end{align*}

Given the axioms and rules of inference, we assume that the resulting notions of proof, length of proof and theorem are already familiar to the reader. Following standard usage, $\vdash t \leqslant t'$ means that $t \leqslant t'$ is a theorem of $AX_{CLL}$.

\subsection{Soundness}
This subsection will establish the soundness of $AX_{CLL}$. To this end, a number of properties of general external choice $\underset{i<n}\square t_i$ are needed. First, a simpler result is given below:

\begin{lemma}\label{L:BIG_SQUARE}
Let $n\geq 0$ and $\{a_i | i<n \}\subseteq Act$.
\begin{enumerate}
  \item $ \underset{i<n}\square t_i \in F_{CLL}$ iff $t_k\in F_{CLL}$ for some $k<n$.
  \item  $\underset{i<n}\square a_i.t_i \stackrel{a_i}{\longrightarrow}_{CLL} t_i$ for each $i<n$.
  \item If $\underset{i<n}\square a_i.t_i \stackrel{\alpha}{\longrightarrow}_{CLL} s$ then $\alpha = a_k$ and $s \equiv t_k$ for some $k<n$.
\end{enumerate}
\end{lemma}
\begin{proof}
  Proceed by induction on $n$, omitted.
\end{proof}

\begin{proposition}\label{L:MULTIPLE_VI}
Let $a_i, b_j \in Act$ for each $i<n$ and $j<m$.
  \begin{enumerate}
    \item If $\{a_i|i< n\}\neq \{b_j|j< m\}$ then $\underset{i< n}{\square}a_i.t_i \wedge \underset{j< m }{\square} b_j.s_j =_{RS} \bot$.
    \item If $\{a_i|i<n\}=\{b_j|j<m\}$ then
     $\underset{\begin{subarray}
                   \;\;\;\;a_i= b_j,\\
                   i< n,j<m
                \end{subarray}} {\square}a_i.(t_i \wedge s_j) \sqsubseteq_{RS} \underset{i< n}{\square}a_i.t_i \wedge \underset{j< m}{\square}b_j.s_j$.
    \item $\underset{i< n} {\square}a_i.(t_i \wedge s_i) \sqsubseteq_{RS} \underset{i< n}{\square}a_i.t_i \wedge \underset{i< n}{\square}a_i.s_i$.
  \end{enumerate}
\end{proposition}
\begin{proof}
  \noindent \textbf{(1)} Since $\{a_i|i< n\}\neq \{b_j|j< m\}$, w.l.o.g, we may assume that $a_k \notin \{b_j|j< m\}$ for some $k< n$.
  Since $a_i,b_j\in Act$ for each $i< n$ and $j< m$, by Lemma~\ref{L:BIG_SQUARE}(2)(3), we have $\underset{i< n}{\square}a_i.t_i \stackrel{a_k}{\longrightarrow}_{{CLL}} t_k $, $\underset{i< n}{\square}a_i.t_i \not\stackrel{\tau}{\longrightarrow}_{{CLL}}$, $\underset{j< m}{\square}b_j.s_j \not\stackrel{a_k}{\longrightarrow}_{{CLL}}$ and $\underset{j< m}{\square}b_j.s_j \not\stackrel{\tau}{\longrightarrow}_{{CLL}}$.
  Further, by Lemma~\ref{L:STABILIZATION}, we have $\underset{i< n} {\square}a_i.t_i \wedge \underset{j< m}{\square}b_j.s_j \not\stackrel{\tau}{\longrightarrow}_{{CLL}}$.
    So, $\frac{\underset{i< n}{\square}a_i.t_i \stackrel{a_k}{\longrightarrow} t_k }{\underset{i< n}{\square}a_i.t_i \wedge \underset{j< m}{\square}b_j.s_jF }\in Strip({\mathcal R}_{CLL},M_{{CLL}})$.
    Hence,  we have $\underset{i< n}{\square}a_i.t_i \wedge \underset{j< m}{\square}b_j.s_jF \in M_{{CLL}}$.
Then, $\underset{i< n}{\square}a_i.t_i \wedge \underset{j< m}{\square}b_j.s_j =_{RS} \bot$ holds trivially.\\

\noindent \textbf{(2)}
If $n=0$ then $m=0$ and it is trivial to check $ 0 \sqsubseteq_{RS} 0 \wedge 0$.
In the following, we consider the case where $n>0$.
For each $i<n$, by Lemma~\ref{L:CON_ID_I}(2) and Prop.~\ref{L:CONGRUENCE}(1), we have
\[a_i.(t_i \wedge s_j) \sqsubseteq_{RS} a_i.t_i\;\text{for each}\; j<m\;\text{such that}\;a_i=b_j.\]
Therefore, by Prop.~\ref{L:EC2}(2)(4)(6) and \ref{L:CONGRUENCE}(2), it follows that
\[\underset{\begin{subarray}
                   \;\;\;\;a_i= b_j,\\
                   i< n,j<m
                \end{subarray}} {\square}a_i.(t_i \wedge s_j) \sqsubseteq_{RS} \underset{i< n}{\square}a_i.t_i. \tag{\ref{L:MULTIPLE_VI}.1}\]
Similarly, we have
\[\underset{\begin{subarray}
                   \;\;\;\;a_i= b_j,\\
                   i< n,j<m
                \end{subarray}} {\square}b_j.(t_i \wedge s_j) \sqsubseteq_{RS} \underset{j< m}{\square}b_j.s_j. \tag{\ref{L:MULTIPLE_VI}.2}\]
Since $\{a_i|i<n\}=\{b_j|j<m\}$, by Lemma~\ref{L:EC2}(2)(4), we have
\[\underset{\begin{subarray}
                   \;\;\;\;a_i= b_j,\\
                   i< n,j<m
                \end{subarray}} {\square}a_i.(t_i \wedge s_j) =_{RS}
                \underset{\begin{subarray}
                   \;\;\;\;a_i= b_j,\\
                   i< n,j<m
                \end{subarray}} {\square}b_j.(t_i \wedge s_j). \tag{\ref{L:MULTIPLE_VI}.3}\]
Thus, by Lemma~\ref{L:CON_ID_I}(4), it follows from (\ref{L:MULTIPLE_VI}.1), (\ref{L:MULTIPLE_VI}.2) and (\ref{L:MULTIPLE_VI}.3) that
 \[\underset{\begin{subarray}
                   \;\;\;\;a_i= b_j,\\
                   i< n,j<m
                \end{subarray}} {\square}a_i.(t_i \wedge s_j) \sqsubseteq_{RS} \underset{i< n}{\square}a_i.t_i \wedge \underset{j< m}{\square}b_j.s_j. \]

\noindent \textbf{(3)} Immediately follows from (2).
\end{proof}

In the following, we provide an example to illustrate that it does not always hold that $\underset{i< n}{\square}a_i.t_i\wedge \underset{i< n}{\square}a_i.s_i \sqsubseteq_{RS} \underset{i< n} {\square}a_i.(t_i \wedge s_i)$.

\begin{example}
Consider process terms
$a_0.t_0 \triangleq a.b.0$, $a_1.t_1 \triangleq a.c.0$, $a_0.s_0 \triangleq a.b.0$ and $a_1.s_1 \triangleq a.b.0$ where $c \neq b$.
Then, $\underset{i<2}\square a_i.t_i \equiv a.b.0 \Box a.c.0$,
$\underset{i<2}\square a_i.s_i \equiv a.b.0 \Box a.b.0$
and $\underset{i<2}\square a_i.(t_i \wedge s_i) \equiv a.(b.0 \wedge b.0) \Box a.(c.0 \wedge b.0)$.
Assume that $\underset{i<2}\square a_i.t_i \wedge \underset{i<2}\square a_i.s_i \sqsubseteq_{RS} \underset{i<2}\square a_i.(t_i \wedge s_i) $.
Since these process terms are stable, we have $\underset{i<2}\square a_i.t_i \wedge \underset{i<2}\square a_i.s_i \underset{\thicksim}{\sqsubset}_{RS} \underset{i<2}\square a_i.(t_i \wedge s_i) $.
It follows from $c.0 \wedge b.0 \not\stackrel{\tau}{\longrightarrow}_{CLL}$ and $b.0 \not\stackrel{c}{\longrightarrow}_{CLL}$ that $\frac{c.0 \stackrel{c}{\longrightarrow}0}{c.0 \wedge b.0F} \in Strip(R_{CLL},M_{CLL})$.
So, $c.0 \wedge b.0 \in F_{CLL}$ follows from $c.0 \stackrel{c}{\longrightarrow}_{CLL}0$.
Further, by Lemma~\ref{L:F_NORMAL}, we get $\underset{i<2}\square a_i.(t_i \wedge s_i)  \in F_{CLL}$.
Thus, it follows from $\underset{i<2}\square a_i.t_i \wedge \underset{i<2}\square a_i.s_i \underset{\thicksim}{\sqsubset}_{RS} \underset{i<2}\square a_i.(t_i \wedge s_i) $ that $\underset{i<2}\square a_i.t_i \wedge \underset{i<2}\square a_i.s_i \in F_{CLL}$.
Since $\underset{i<2}\square a_i.t_i \notin F_{CLL}$, $ \underset{i<2}\square a_i.s_i \notin F_{CLL}$ and ${\mathcal I}(\underset{i<2}\square a_i.t_i) = {\mathcal I}(\underset{i<2}\square a_i.s_i)$, the last rule applied in the proof tree of $Strip({\mathcal P}_{CLL},M_{CLL})\vdash \underset{i<2}\square a_i.t_i \wedge \underset{i<2}\square a_i.s_iF$ is of the form \[\frac{\underset{i<2}\square a_i.t_i \wedge \underset{i<2}\square a_i.s_i \stackrel{a}{\longrightarrow}s}{\underset{i<2}\square a_i.t_i \wedge \underset{i<2}\square a_i.s_iF}\;\text{with}\; M_{CLL}\models \underset{i<2}\square a_i.t_i \wedge \underset{i<2}\square a_i.s_i \neg \overline{F}_{a}.\]
Then, by Lemma~\ref{L:CON_LLTS}, $t \in F_{CLL}$ for each $t$ such that $\underset{i<2}\square a_i.t_i \wedge \underset{i<2}\square a_i.s_i \stackrel{a}{\longrightarrow}_{CLL}t$.
But, it is easy to see that $\underset{i<2}\square a_i.t_i \wedge \underset{i<2}\square a_i.s_i \stackrel{a}{\longrightarrow}_{CLL}b.0 \wedge b.0$ and $b.0 \wedge b.0 \notin F_{CLL}$, a contradiction.
$\qquad\qquad\qquad\qquad\qquad\qquad\qquad\qquad\qquad\qquad\qquad\qquad\qquad\;\;\square$
\end{example}

However, for any general external choice $\underset{i<n}\square a_i.t_i$ with distinct prefixes, we have

\begin{proposition}\label{L:MULTIPLE_II}
Let $a_i \in Act$ for each $i<n$.
  If $\underset{i< n}{\square}a_i.t_i$ is injective in prefixes then $\underset{i< n}{\square}a_i.t_i\wedge \underset{i< n}{\square}a_i.s_i \sqsubseteq_{RS} \underset{i< n} {\square}a_i.(t_i \wedge s_i)$.
\end{proposition}
\begin{proof}
We consider the non-trivial case where $n>0$.
Since $a_i \in Act$ for each $i< n$, by Lemma~\ref{L:BIG_SQUARE}(3) and \ref{L:STABILIZATION}, it is easy to see that $\underset{i< n}{\square}a_i.t_i \wedge \underset{i < n}{\square}a_i.s_i$ and $\underset{i< n}{\square}a_i.(t_i \wedge s_i)$ are stable.
Thus, it suffices to prove $ \underset{i< n}{\square}a_i.t_i \wedge \underset{i< n}{\square}a_i.s_i \underset{\thicksim}{\sqsubset}_{RS} \underset{i< n} {\square}a_i.(t_i \wedge s_i)$. Put
\[ {\mathcal R}=\{ (\underset{i< n}{\square}a_i.t_i \wedge \underset{i< n}{\square}a_i.s_i , \underset{i< n} {\square}a_i.(t_i \wedge s_i)) \} \cup Id_S.\]
We need to show that the pair $(\underset{i< n}{\square}a_i.t_i \wedge \underset{i< n}{\square}a_i.s_i , \underset{i< n} {\square}a_i.(t_i \wedge s_i))$ satisfies (RS1)-(RS4).
It is easy to check that (RS1) and (RS4) hold. We deal with (RS2) and (RS3) as follows.

\textbf{(RS2)} Suppose $\underset{i< n} {\square}a_i.(t_i \wedge s_i) \in F_{CLL}$.
By Lemma~\ref{L:BIG_SQUARE}(1) and \ref{L:F_NORMAL}(2), there exists $k<n$ such that $t_k \wedge s_k \in F_{CLL}$.
On the other hand, by Lemma~\ref{L:BIG_SQUARE}(2), we obtain $ \underset{i< n}{\square}a_i.t_i \stackrel{a_k}{\longrightarrow}_{CLL} t_k$ and $ \underset{i< n}{\square}a_i.s_i \stackrel{a_k}{\longrightarrow}_{CLL} s_k$.
Further, it follows from Lemma~\ref{L:Basic_II}(2) that $\underset{i< n}{\square}a_i.t_i \wedge \underset{i< n}{\square}a_i.s_i \stackrel{a_k}{\longrightarrow}_{CLL} t_k \wedge s_k$.
Moreover, since both $\underset{i<n}{\square}a_i.t_i$ and $\underset{i<n}{\square}a_i.s_i$ are injective in prefixes, by Lemma~\ref{L:Basic_II}(2) and \ref{L:BIG_SQUARE}(3), it follows that $t_k \wedge s_k$ is the unique $a_k$-derivative of $ \underset{i< n}{\square}a_i.t_i\wedge \underset{i< n}{\square}a_i.s_i$. Therefore, by Lemma~\ref{L:LLTS_I}, $ \underset{i< n}{\square}a_i.t_i\wedge \underset{i< n}{\square}a_i.s_i \in F_{CLL}$ comes from $t_k \wedge s_k \in F_{CLL}$, as desired.

\textbf{(RS3)} Suppose $\underset{i< n}{\square}a_i.t_i \wedge \underset{i< n}{\square}a_i.s_i \stackrel{a}{\Longrightarrow}_F |t'$.
Since $\underset{i< n}{\square}a_i.t_i \wedge \underset{i< n}{\square}a_i.s_i$ is stable, we have $\underset{i< n}{\square}a_i.t_i \wedge \underset{i< n}{\square}a_i.s_i \stackrel{a}{\longrightarrow}_F t''  \stackrel{\epsilon}{\Longrightarrow}_F |t'$ for some $t''$.
Since $\underset{i<n}\square a_i.t_i$ and $\underset{i<n} \square a_i.s_i$ are injective in prefixes, by Lemma~\ref{L:Basic_II}(2) and \ref{L:BIG_SQUARE}(3), we get $\underset{i< n}{\square}a_i.t_i \stackrel{a_k}{\longrightarrow}_{CLL} t_k$, $\underset{i< n}{\square}a_i.s_i \stackrel{a_k}{\longrightarrow}_{CLL} s_k$, $a = a_k$ and $t'' \equiv t_k \wedge s_k$ for some $k<n$.
On the other hand, by Lemma~\ref{L:BIG_SQUARE}(2), we get $\underset{i< n} {\square}a_i.(t_i \wedge s_i) \stackrel{a_k}{\longrightarrow}_{CLL} t_k \wedge s_k$.
Moreover, since $\underset{i< n}{\square}a_i.t_i \wedge \underset{i< n}{\square}a_i.s_i \notin F_{CLL}$, by (RS2), $\underset{i< n} {\square}a_i.(t_i \wedge s_i) \notin F_{CLL}$.
Hence, we obtain $\underset{i< n} {\square}a_i.(t_i \wedge s_i) \stackrel{a}{\longrightarrow}_F t_k\wedge s_k \equiv t''  \stackrel{\epsilon}{\Longrightarrow}_F |t' $ and $(t',t')\in {\mathcal R}$.
\end{proof}

The next two propositions state the properties of the interaction of general external choice and parallel operator, which are analogous to the expansion law in usual process calculus.

\begin{proposition}\label{L:MULTIPLE_I}
Let $n \geq 0$, $m \geq 0$, $A \subseteq Act$ and
$a_i,b_j\in Act$ for each $i<n$ and $j<m$.Then
\[\underset{i< n}{\square}a_i.t_i \parallel_A \underset{j< m}{\square}b_j.s_j
  \sqsubseteq_{RS} ((\square \Omega_1) \Box (\square \Omega_2)) \Box (\square \Omega_3),\]
where
$\Omega_1 = \{a_i.(t_i \parallel_A \underset{j< m}{\square}b_j.s_j)|i<n\;\text{and}\;a_i \notin A \}$,
$\Omega_2 =  \{b_j.(\underset{i< n}{\square}a_i.t_i \parallel_A s_j)|j<m\;\text{and}\;b_j \notin A\}$ and
$\Omega_3 =  \{a_i.(t_i \parallel_A s_j)|a_i = b_j \in A,i<n\;\text{and}\;j<m\}$.
\end{proposition}
\begin{proof}
  Set $N \triangleq \underset{i< n}{\square}a_i.t_i \parallel_A \underset{j< m}{\square}b_j.s_j$ and $M\triangleq  ((\square \Omega_1) \Box (\square \Omega_2)) \Box (\square \Omega_3)$.
  Since $a_i,b_j\in Act$ for each $i< n$ and $j< m$, by Lemma~\ref{L:BIG_SQUARE}(3) and \ref{L:STABILIZATION}, both $N$ and $M$ are stable.
  It is enough to prove $N\underset{\thicksim}{\sqsubset}_{RS} M$. Put
  \[{\mathcal R}\triangleq\{(N,M)\}\cup Id_S.\]
  It suffices to show that $\mathcal R$ is a stable ready simulation relation. We will check that the pair $(N, M)$ satisfies (RS2), (RS3) and (RS4) one by one.

  \textbf{(RS2)} Suppose $M\in F_{CLL}$.
  Then, by Lemma~\ref{L:BIG_SQUARE}(1) and \ref{L:F_NORMAL}(3), we have $t \in F_{CLL}$ for some $t \in \Omega_1 \cup \Omega_2 \cup \Omega_3$.
  We shall consider the case where $t \in \Omega_1$, the others may be treated similarly and omitted.
  In such case, we may assume that $t \equiv a_{i_0}.(t_{i_0} \parallel_A \underset{j< m}{\square}b_j.s_j)$ with $i_0<n$ and $a_{i_0}\notin A$.
  So, by Lemma~\ref{L:F_NORMAL}(2)(3), either $t_{i_0} \in F_{CLL}$ or $\underset{j< m}{\square}b_j.s_j  \in F_{CLL}$.
  Then, by Lemma~\ref{L:BIG_SQUARE}(1) and \ref{L:F_NORMAL}(2)(3), it is easy to see that each of them implies   $\underset{i< n}{\square}a_i.t_i \parallel_A \underset{j< m}{\square}b_j.s_j \in F_{CLL}$, as desired.

  \textbf{(RS3)} Suppose $N \stackrel{a}{\Longrightarrow}_F|t'$.
  So, $N \notin F_{CLL}$, further, by (RS2), we get $ M \notin F_{CLL}$.
  Since $N$ is stable, $N \stackrel{a}{\longrightarrow}_F t'' \stackrel{\epsilon}{\Longrightarrow}_F|t'$ for some $t''$.
  We consider three cases based on the form of the last rule applied in the proof tree of  $Strip({\mathcal P}_{CLL},M_{CLL}) \vdash \underset{i< n}{\square}a_i.t_i \parallel_A \underset{j< m}{\square}b_j.s_j \stackrel{a}{\longrightarrow}t''$.\\

\noindent  Case 1 $\frac{\underset{i< n}{\square}a_i.t_i \stackrel{a}{\longrightarrow} r}{\underset{i< n}{\square}a_i.t_i \parallel_A \underset{j< m}{\square}b_j.s_j \stackrel{a}{\longrightarrow}r\parallel_A \underset{j< m}{\square}b_j.s_j}$ with $M_{CLL}\models \underset{j< m}{\square}b_j.s_j \not\stackrel{\tau}{\longrightarrow}$ and $a\notin A$.

  So, $\underset{i< n}{\square}a_i.t_i \stackrel{a}{\longrightarrow} r \in M_{CLL}$ and $t'' \equiv r\parallel_A \underset{j< m}{\square}b_j.s_j$.
  By Lemma~\ref{L:BIG_SQUARE}(3), we have $a= a_{i_0}$ and $r \equiv t_{i_0}$ for some $i_0 < n$.
  Since $a_{i_0}=a \notin A$, $a_{i_0}.(t_{i_0} \parallel_A \underset{j< m}{\square}b_j.s_j) \in \Omega_1$.
  So, by Lemma~\ref{L:BIG_SQUARE}(2), $\square \Omega_1 \stackrel{a_{i_0}}{\longrightarrow}_{CLL}t_{i_0} \parallel_A \underset{j< m}{\square}b_j.s_j$.
  Further, by Lemma~\ref{L:BIG_SQUARE}(3), it follows from $\{a_i,b_j|i<n\;\text{and}\;j<m\}\subseteq Act$ that
  \[\square \Omega_2 \not\stackrel{\tau}{\longrightarrow}_{CLL}\;\text{and}\;\square \Omega_3 \not\stackrel{\tau}{\longrightarrow}_{CLL}.\]
  Then, by Lemma~\ref{L:Basic_II}(1), we get  $M \stackrel{a_{i_0}}{\longrightarrow}_{CLL}t_{i_0} \parallel_A \underset{j< m}{\square}b_j.s_j$.
  Hence, $ M \stackrel{a}{\Longrightarrow}_F|t'$ and $(t',t')\in {\mathcal R}$.\\

\noindent  Case 2 $\frac{\underset{j< m}{\square}b_j.s_j \stackrel{a}{\longrightarrow} r}{\underset{i< n}{\square}a_i.t_i \parallel_A \underset{j< m}{\square}b_j.s_j \stackrel{a}{\longrightarrow} \underset{i< n}{\square}a_i.t_i \parallel_A r}$ with $M_{CLL}\models \underset{i< n}{\square}a_i.t_i \not\stackrel{\tau}{\longrightarrow}$ and $a\notin A$.

  Similar to Case 1.\\

\noindent  Case 3 $\frac{\underset{i< n}{\square}a_i.t_i \stackrel{a}{\longrightarrow} r,\underset{j< m}{\square}b_j.s_j \stackrel{a}{\longrightarrow} s}{\underset{i< n}{\square}a_i.t_i \parallel_A \underset{j< m}{\square}b_j.s_j \stackrel{a}{\longrightarrow}r\parallel_A s}$ with $a\in A$.

  So, $\underset{i< n}{\square}a_i.t_i \stackrel{a}{\longrightarrow} r \in M_{CLL}$, $ \underset{j< m}{\square}b_j.s_j \stackrel{a}{\longrightarrow} s \in M_{CLL}$ and $t'' \equiv r\parallel_A s$.
  By Lemma~\ref{L:BIG_SQUARE}(3), we have $a= a_{i_0}$, $r \equiv t_{i_0}$ for some $i_0 < n$ and $a = b_{j_0}$, $s \equiv s_{j_0}$ for some $j_0 < m$.
  Since $a_{i_0} = b_{j_0}=a \in A$, $a_{i_0}.(t_{i_0} \parallel_A s_{j_0}) \in \Omega_3$.
  So, by Lemma~\ref{L:BIG_SQUARE}(2), $\square \Omega_3 \stackrel{a_{i_0}}{\longrightarrow}_{CLL}t_{i_0} \parallel_A s_{j_0}$.
  Moreover, by Lemma~\ref{L:BIG_SQUARE}(3), it follows from $\{a_i,b_j|i<n\;\text{and}\;j<m\}\subseteq Act$ that $\square \Omega_1 \not\stackrel{\tau}{\longrightarrow}_{CLL}$ and $\square \Omega_2 \not\stackrel{\tau}{\longrightarrow}_{CLL}$. Then, by Lemma~\ref{L:Basic_II}(1), we get  $M \stackrel{a_{i_0}}{\longrightarrow}_{CLL}t_{i_0} \parallel_A s_{j_0}$.
  Hence, $ M \stackrel{a}{\Longrightarrow}_F|t'$ and $(t',t')\in {\mathcal R}$.

  \textbf{(RS4)} We just prove ${\mathcal I}( M) \subseteq{\mathcal I}(N)$, the proof of ${\mathcal I}(M) \supseteq{\mathcal I}(N)$ is similar and omitted.
   Assume that $a \in {\mathcal I}(M)$.
   So, $ M \stackrel{a}{\longrightarrow}_{CLL} t'$ for some $t'$.
   By Lemma~\ref{L:Basic_II}(1), $s \stackrel{a}{\longrightarrow}_{CLL}t'$ for some $s \in \underset{1 \leq i \leq 3}\cup \Omega_i$.
     We shall consider the case where $s\in \Omega_1$, the others are similar and omitted.
     In such case, we may assume $s \equiv a_{i_0}.(t_{i_0}\parallel_A \underset{j<m}{\square}b_j.s_j)$ with $i_0<n$ and $a_{i_0}\notin A$.
     Then, we get $a = a_{i_0}$ and $t' \equiv t_{i_0}\parallel_A \underset{j<m}{\square}b_j.s_j$.
   Since $\{b_j|j<m\}\subseteq Act$, by Lemma~\ref{L:BIG_SQUARE}(3), $M_{CLL}\models \underset{j< m}{\square}b_j.s_j \not\stackrel{\tau}{\longrightarrow}$.
   Further, it follows from $a_{i_0}\notin A$ that
   \[ \frac{\underset{i< n}{\square}a_i.t_i \stackrel{a_{i_0}}{\longrightarrow} t_{i_0}}{\underset{i< n}{\square}a_i.t_i \parallel_A \underset{j< m}{\square}b_j.s_j \stackrel{a_{i_0}}{\longrightarrow} t_{i_0}\parallel_A \underset{j< m}{\square}b_j.s_j }\in Strip(R_{CLL},M_{CLL}).\tag{\ref{L:MULTIPLE_I}.1.1}\]
   On the other hand, by Lemma~\ref{L:BIG_SQUARE}(2), we have $\underset{i\leq n}{\square}a_i.t_i \stackrel{a_{i_0}}{\longrightarrow}_{CLL} t_{i_0}$.
    Therefore, it follows from (\ref{L:MULTIPLE_I}.1.1) that $\underset{i< n}{\square}a_i.t_i \parallel_A \underset{j< m}{\square}b_j.s_j \stackrel{a_{i_0}}{\longrightarrow} t_{i_0}\parallel_A \underset{j< m}{\square}b_j.s_j \in M_{CLL}$.
    So, $a = a_{i_0} \in {\mathcal I}(N)$, as desired.
\end{proof}

Compare to usual expansion law in process calculus, e.g., Prop. 3.3.5 in \cite{Milner89}, we expect that the inequation below holds, where $\Omega_i$ ($1\leq i\leq 3$) is same as ones in Prop.~\ref{L:MULTIPLE_I}.
\[((\square \Omega_1) \Box (\square \Omega_2)) \Box (\square \Omega_3) \sqsubseteq_{RS}    \underset{i< n}{\square}a_i.t_i \parallel_A \underset{j< m}{\square}b_j.s_j. \tag{EXP}\]
Unfortunately, it isn't valid.
For instance, consider the process terms $a_0.t_0 \triangleq a.\bot$, $a_1.t_1 \triangleq c.0$ and $b_0.s_0 \triangleq b.0$ with $a\neq b$, $c\neq a$ and $c \neq b$.
Let $A = \{a,b\}$.
Clearly, the set $\Omega_i(1 \leq i \leq 3)$ corresponding to ones in the above proposition are: $\Omega_1 =\{c.(0\parallel_{\{a,b\}}b.0)\}$ and $\Omega_2 = \Omega_3 = \emptyset$.
Then, \[(( \square \Omega_1) \Box (\square \Omega_2)) \Box ( \square \Omega_3 ) \equiv (c.(0 \parallel_{\{a,b\}} b.0)\Box 0)\Box 0.\]
By Lemma~\ref{L:F_NORMAL}, $(a.\bot \Box c.0) \parallel_{\{a,b\}} b.0 \in F_{CLL}$ and  $(c.(0\parallel_{\{a,b\}}b.0) \Box 0)\Box 0 \notin F_{CLL}$.
Then, it is easy to see that $(c.(0\parallel_{\{a,b\}}b.0) \Box 0)\Box 0 \not \sqsubseteq_{RS} (a.\bot \Box c.0) \parallel_{\{a,b\}} b.0 $.

However, the inequation (EXP) holds for the process terms satisfying a moderate condition. Formally, we have the result below.

\begin{proposition}\label{L:MULTIPLE_IV}
Let $n,m\geq0$, $A \subseteq Act$, and let $t_i,s_j \in T(\Sigma_{CLL})$ and
$a_i,b_j\in Act$ for each $i<n$ and $j<m$.
Assume that $(\{t_i|a_i\in A\;\text{and}\; a_i\neq b_j\;\text{for each}\; j<m\}\cup \{s_j|b_j\in A\;\text{and}\; b_j\neq a_i\;\text{for each}\;i<n\}) \cap F_{CLL} =\emptyset$,
 then \[((\square \Omega_1) \Box (\square \Omega_2)) \Box (\square \Omega_3) \sqsubseteq_{RS}    \underset{i< n}{\square}a_i.t_i \parallel_A \underset{j< m}{\square}b_j.s_j,\]
where $\Omega_i$ ($1\leq i\leq 3$) is same as ones in Prop.~\ref{L:MULTIPLE_I}.
\end{proposition}
\begin{proof}
  Set $M \triangleq \underset{i< n}{\square}a_i.t_i \parallel_A \underset{j< m}{\square}b_j.s_j$ and $N \triangleq  ((\square \Omega_1) \Box (\square \Omega_2)) \Box (\square \Omega_3)$.
  Similar to Prop.~\ref{L:MULTIPLE_I}, we shall prove $N\underset{\thicksim}{\sqsubset}_{RS} M$. Put
  \[{\mathcal R}\triangleq\{(N,M)\}\cup Id_S.\]
  It suffices to show that $\mathcal R$ is a stable ready simulation relation. We will check that the pair $(N, M)$ satisfies (RS2), the remainder is similar to Prop.~\ref{L:MULTIPLE_I}.

  \textbf{(RS2)} Suppose $M \in F_{CLL}$. It follows from Lemma~\ref{L:F_NORMAL} and \ref{L:BIG_SQUARE}(1) that
  \[\text{either}\;t_{i_0} \in F_{CLL}\;\text{for some}\;i_0<n\;\text{or}\; s_{j_0} \in F_{CLL}\;\text{for some}\; j_0<m. \]
  W.l.o.g, we consider the first alternative.
  Then, by the assumption, we get either $a_{i_0} \notin A$ or $a_{i_0} = b_{j_0}$ for some $j_0 < m$.
  Consequently, either $a_{i_0}.(t_{i_0} \parallel_A \underset{j< m}{\square}b_j.s_j) \in \Omega_1$ or $a_{i_0}.(t_{i_0} \parallel_A s_{j_0}) \in \Omega_3$.
  So, by Lemma~\ref{L:F_NORMAL} and \ref{L:BIG_SQUARE}(1), $N\in F_{CLL}$ follows, as desired.
\end{proof}



By Lemma~\ref{L:F_NORMAL}, it is easy to see that the operators $\alpha.()$, $\vee$, $\Box$ and $\parallel_A$ preserve consistency. Thus, an immediate consequence of Lemma~\ref{L:F_NORMAL} is

\begin{lemma}\label{L:BPT}
  $T(\Sigma_B)\cap F_{CLL}=\emptyset$.
\end{lemma}
\begin{proof}
  Induction on the structure of basic process terms (see, Def.~\ref{D:BPT}).
\end{proof}

Therefore, as a corollary of Prop.~\ref{S:SPECIAL} and \ref{L:MULTIPLE_IV}, we have

\begin{proposition}\label{L:MULTIPLE_V}
Let $r_1,r_2\in T(\Sigma_B)$ and $t_i,s_j \in T(\Sigma_B)$ for each $i<n$ and $j<m$. Then
  \begin{enumerate}
    \item $a.(r_1 \vee r_2) \sqsubseteq_{RS} a.r_1 \Box a.r_2 $ for each $a \in Act$.
    \item $((\square \Omega_1) \Box (\square \Omega_2)) \Box (\square \Omega_3) \sqsubseteq_{RS}    \underset{i< n}{\square}a_i.t_i \parallel_A \underset{j< m}{\square}b_j.s_j,$
        where $\Omega_i(1\leq i \leq 3)$ is same as ones in Prop.~\ref{L:MULTIPLE_I}.
  \end{enumerate}
\end{proposition}
\begin{proof}
    Immediately follows from Lemma~\ref{L:BPT} and Prop.~\ref{S:SPECIAL} and \ref{L:MULTIPLE_IV}.
\end{proof}

We now have all of the properties we require to prove the soundness of the axiomatic system $AX_{CLL}$.

\begin{mydefn}
  For any $t,s\in T(\Sigma_{CLL})$, the inequation $t\leqslant s$ is said to be valid in $M_{CLL}$, in symbols $M_{CLL}\models t\leqslant s$, if and only if $t\sqsubseteq_{RS}s$.
\end{mydefn}

\begin{theorem}[Soundness]\label{T:SOUNDNESS}
If $\vdash t\leqslant s$ then $M_{CLL}\models t\leqslant s$ for any $t,s \in T(\Sigma_{CLL})$.
\end{theorem}
\begin{proof}
As usual, it is enough to show that
\begin{enumerate}
  \item all ground instances of axioms are valid in $M_{CLL}$, and
  \item all inference rules preserve validity.
\end{enumerate}
Item (1) is implied by Prop.~\ref{S:PREFIX}, \ref{S:DISJUNCTION}, \ref{L:EC2}, \ref{S:CONJUNCTION}, \ref{S:PARALLEL}, \ref{S:DISTRIBUTIVE}, \ref{L:MULTIPLE_VI}, \ref{L:MULTIPLE_II}, \ref{L:MULTIPLE_I} and \ref{L:MULTIPLE_V}.
Item (2) is implied by Prop.~\ref{L:CONGRUENCE} and the fact that $\sqsubseteq_{RS}$ is reflexive and transitive.
\end{proof}

\subsection{Normal Form and Ground-Completeness}
This subsection will establish the ground-completeness of $AX_{CLL}$. We begin by giving two useful notations.\\

\noindent \textbf{Notation}
\begin{enumerate}
  \item $Prefix(\underset{i<n}{\square}a_i.t_i)\triangleq \{a_i|i<n\}$.
  \item Let $<t_0,t_1,\dots,t_{n-1}>$ be a finite sequence of process terms with $n > 0$. We define the general disjunction $\underset{i<n}\bigvee t_i$ by recursion:
    \begin{enumerate}
    \item $\underset{i<1}\bigvee t_i \triangleq t_0$,
    \item $\underset{i<k+1}\bigvee t_i \triangleq (\underset{i<k}\bigvee t_i) \vee t_k$ for $k\geq 1$.
    \end{enumerate}
\end{enumerate}

Similar to general external choice, up to =, the order and grouping of terms in $\underset{i<n}\bigvee t_i$  may be ignored by virtue of axioms $DI1$ and $DI2$.

To prove the ground-completeness of $AX_{CLL}$, we use a general technique involving normal forms.
The idea is to isolate a particular subclass of terms, called normal forms, such that the proof of the completeness is straightforward for it.
The completeness for arbitrary terms will follow if we can show that each term can be reduced to normal form using the equations in $AX_{CLL}$.
We define normal form as follows.

\begin{mydefn}[Normal Form]\label{D:NORMAL_FORM}
    The set $NF_B$ is the least subset of $T(\Sigma_{CLL})$ such that $\underset{i< n}{\bigvee}t_i\in NF_B$ if $n > 0$ and for each $i<n$, $t_i$ has the format $\underset{j< m_i}{\square}a_{ij}.t_{ij}$ with $m_i \geq 0$ such that

        \noindent (N)\;\;\;  $t_{ij}\in NF_B$ for each $j<m_i$,

        \noindent (D) \;\; $\underset{j< m_i}{\square}a_{ij}.t_{ij}$ is injective in prefixes, and

        \noindent (N-$\tau$)  $a_{ij} \in Act$ for each $j<m_i$.

We put \[NF \triangleq\{ \bot \} \cup NF_B.\]
Each process term in $NF$ is said to be  in normal form.
\end{mydefn}

Notice that $NF_B \subseteq T(\Sigma_B)$, and $0\in NF_B$ by taking $n=1$ and $m_0 = 0$ in $ \underset{i<n}{\bigvee}\underset{j<m_i}{\square}t_{ij}$.
In the following, we will show that each process term can be transformed using the equations into a normal form, which is a crucial step in establishing the ground-completeness of $AX_{CLL}$.
To this end, the next five lemmas are firstly proved.

\begin{lemma}\label{L:DIS_INEQUATION}
  $\vdash (x\odot y)\vee(x\odot z)\leqslant x \odot (y \vee z)$ for each $\odot \in \{\Box, \wedge, \parallel_A\}$.
\end{lemma}
\begin{proof}
  $\vdash y \leqslant y \vee z$ and $\vdash z \leqslant y \vee z$  \qquad\qquad\qquad(by $DI1$, $DI5$ and TRANS)

\noindent   $\Rightarrow \vdash x \odot y \leqslant x \odot (y \vee z)$ and $\vdash x\odot z \leqslant x\odot (y \vee z)$ (by CONTEXT and REF)

\noindent  $\Rightarrow \vdash (x \odot y) \vee (x\odot z) \leqslant x \odot (y \vee z)$ \qquad \qquad (by $DI3$, CONTEXT and TRANS)
\end{proof}

\begin{lemma}\label{L:COMP_CONJ}
  If $t,s\in NF_B$ then $\vdash t \wedge s = r$ for some $r \in NF$.
\end{lemma}
\begin{proof}
We prove it by induction on the number $|t|+|s|$.
  Since $t,s \in NF_B$, we may assume that $t \equiv \underset{i< n}{\bigvee}t_{i}$ and $s \equiv \underset{i'< n'}{\bigvee}s_{i'}$.
   By $DI1$, $DI2$, $CO1$, $DS2$ and Lemma~\ref{L:DIS_INEQUATION}, we get
   \[ \vdash t \wedge s
    = \bigvee_{i< n,i'< n'}(t_{i}\wedge s_{i'}) \tag{\ref{L:COMP_CONJ}.1}.\]
    Let $i< n$ and $i'< n'$. We will show that $\vdash t_{i} \wedge s_{i'}=r_{ii'}$ for some $r_{ii'} \in NF$. Clearly, we may assume that $t_{i}\equiv \underset{j< m_{i}}{\square}a_{ij}.t_{ij}$ and $s_{i'}\equiv \underset{j'< m_{i'}'}{\square}b_{i'j'}.s_{i'j'}$ satisfying (N), (D) and (N-$\tau$) in Def.~\ref{D:NORMAL_FORM}. We consider two cases below.\\

    \noindent Case 1 $Prefix(t_{i})\neq Prefix(s_{i'})$.

          By  $ECC1$, we have $\vdash t_{i} \wedge s_{i'}  = \bot $.\\

    \noindent Case 2  $Prefix(t_{i})= Prefix(s_{i'})$.

        Thus, by the item (D) in Def.~\ref{D:NORMAL_FORM}, we have $m_i = m_i'$.
        If $m_i=0$ then, by the definition of general external choice, we get $t_i\equiv s_{i'} \equiv 0$.
        Moreover, $\vdash t_i \wedge s_{i'}=0$ follows from $CO3$.
        In the following, we consider the nontrivial case where $m_i>0$.
        By $EC1$, $EC2$, $ECC2$ and $ECC3$, it follows that
        \[\vdash t_{i} \wedge s_{i'}  = \underset{\begin{subarray}
                   \; j,j'<m_i,\\
                   a_{ij}= b_{i'j'}
                \end{subarray}}\square a_{ij}.(t_{ij} \wedge s_{i'j'}). \]
         For each pair $j,j'< m_i$  with $a_{ij}= b_{i'j'}$,
         since $t_{ij},s_{i'j'} \in NF_B$ and $|t|+|s|>|t_{ij}|+|s_{i'j'}|$, by IH, we have $\vdash t_{ij} \wedge s_{i'j'} = t_{iji'j'}$ for some  $t_{iji'j'} \in NF$.
         Set \[S \triangleq \underset{\begin{subarray}
                   \; j,j'<m_i,\\
                   a_{ij}= b_{i'j'}
                \end{subarray}}\square a_{ij}.t_{iji'j'}.\]
        Consequently, by  CONTEXT and TRANS, we have
            \[\vdash t_{i} \wedge s_{i'} = S.\]
        Clearly, if $t_{iji'j'}\in NF_{B}$ for each pair $j,j'<m_i$ with  $a_{ij}= b_{i'j'}$, then $S \in NF_{B}$.
        Otherwise, we have $t_{ij_0i'j_0'} \equiv \bot$ for some $j_0,j_0'<m_i$, then it follows from $PR1$ that \[\vdash  a_{ij_0}.t_{ij_0i'j_0'}  =\bot.\]
        Further, by $EC5$, CONTEXT and TRANS, we get $\vdash S =\bot$.

     In summary, it follows from the discussion above that, for each $i<n$ and $i'<n'$,
     \[\text{either} \vdash t_{i} \wedge s_{i'}= r_{ii'}\;\text{for some}\;r_{ii'}\in NF_{B}\;\text{or}\; \vdash t_{i} \wedge s_{i'}= \bot.\]
     Then, by $DI1$, $DI4$ and (\ref{L:COMP_CONJ}.1), we get either $\vdash t \wedge s = r$ for some $r \in NF_{B}$ or $\vdash t \wedge s= \bot$.
\end{proof}

In the above proof, we do not  explicitly show the proof for the induction basis where $t \equiv s \equiv 0$, as it is an instance of the proof of the induction step.

\begin{lemma}\label{L:SP3}
      $\vdash a.x \Box a.y  \leqslant a.(x \vee y) $.
\end{lemma}
\begin{proof}
$ \vdash x  \leqslant x \vee y$ and $\vdash y  \leqslant x \vee y$   \qquad\;\;\;\;\; (by  $DI1$ ,  $DI5$ and TRANS )

\noindent $\Rightarrow   \vdash a.x  \leqslant a.(x \vee y)$ and $\vdash a.y  \leqslant a.(x \vee y)$  \;\qquad\qquad\qquad\;  (by CONTEXT)

\noindent $\Rightarrow   \vdash a.x \Box a.y \leqslant a.(x \vee y)$ \qquad\qquad\qquad (by CONTEXT,  $EC3$ and TRANS)
\end{proof}

\begin{lemma}\label{L:BIG_SQUARE_EC}
  If $t \equiv \underset{i<n}{\square}a_i.t_i \in NF_B$  and $s \equiv \underset{j<m}{\square}b_j.s_j \in NF_B$, then $\vdash t \Box s = \underset{i<k}\square c_i.r_i$ for some $\underset{i<k}\square c_i.r_i \in NF_B$.
\end{lemma}
\begin{proof}
If $n=0$ or $m=0$ then it immediately follows from $EC1$, $EC4$ and Def.~\ref{D:GEC}.
In the following, we consider the non-trivial case where $n>0$ and $m>0$.
We distinguish two cases below.\\

\noindent Case 1 $Prefix(t) \cap Prefix(s)=\emptyset$.

          Set
            \[p_k \triangleq \begin{cases}
                a_k.t_k  & k<n,\\
                b_{k-n}.s_{k-n} &  n \leq k < m+n.\\
                \end{cases}\]
        Then, it is trivial to check that $\underset{k<m+n}\square p_k $ satisfies (N), (D) and (N-$\tau$) in Def.~\ref{D:NORMAL_FORM}, that is, $\underset{k<m+n}\square p_k \in NF_B$.
        Moreover, by $EC2$ and TRANS, it immediately follows that $\vdash t \Box s = \underset{k<m+n}\square p_k$.\\

\noindent Case 2 $Prefix(t) \cap Prefix(s) \neq \emptyset $.

  Let $i_0<n$ and $j_0<m$ with $a_{i_0} = b_{j_0}$, since $NF_B \subseteq T(\Sigma_B)$, by Lemma~\ref{L:SP3} and $DS4$, we get $\vdash a_{i_0}.t_{i_0} \Box b_{j_0}.s_{j_0} = a_{i_0}.(t_{i_0} \vee s_{j_0})$.
  Further, by Def.~\ref{D:NORMAL_FORM}, $DI1$, $DI2$, CONTEXT and TRANS, it follows from $t_{i_0},s_{j_0}\in NF_B$ that
  \[\vdash a_{i_0}.t_{i_0} \Box b_{j_0}.s_{j_0} = a_{i_0}.p\;\text{for some}\;p\in NF_B.\]
  Thus, for each $i<n$ and $j<m$ with $a_i = b_j$, we can fix a process term $p_{ij}\in NF_B$ such that
  \[\vdash a_i.t_i \Box b_j.s_j = a_i.p_{ij}.\]
  Put
    \begin{enumerate}
      \item  $S_1 \triangleq \underset{ \begin{subarray} \;a_i \notin Prefix(s),\\\;\;\;\;\;\;\;i<n\end{subarray}}{\square}a_i.t_i$,
      \item  $ S_2 \triangleq \underset{ \begin{subarray} \;b_j \notin Prefix(t),\\\;\;\;\;\;\;\;j<m\end{subarray}}{\square}b_j.s_j,$
      \item $S_3 \triangleq \underset{\begin{subarray} \;a_i \in Prefix(t)\cap Prefix(s),\\\;\;\;\;\;a_i = b_j,i<n,j<m\end{subarray}
       }{\square}a_i.p_{ij}.$
    \end{enumerate}
    Then, by $EC1$, $EC2$, TRANS and CONTEXT, we obtain $\vdash t \Box s = (S_1 \Box S_2)\Box S_3$.
    Clearly, both $S_1$ and $S_2$ are in $NF_B$.
    Moreover, since $t$ and $s$ are injective in prefixes, so is $S_3$.
    Hence, $S_3$ is also in $NF_B$.
    Further, since $Prefix(S_i)\cap Prefix(S_j) = \emptyset$ for $1 \leq i \neq j \leq 3$, similar to Case 1, we have
    $\vdash (S_1 \Box S_2)\Box S_3 = \underset{i<k}\square c_i.r_i$ for some $\underset{i<k}\square c_i.r_i \in NF_B$.
\end{proof}

\begin{lemma}\label{L:COMP_PARALLEL}
  If $t,s \in NF_B$ then $\vdash  t\parallel_A s  = r$ for some $r \in NF_B$.
\end{lemma}
\begin{proof}
We prove it by induction on the number $|t|+|s|$.
Since $t,s \in NF_B$, we may assume that $t\equiv \underset{i< n}{\bigvee}t_{i}$ and $s\equiv \underset{i'< n'}{\bigvee}s_{i'}$.
 By axioms $DI1$, $DI2$, $PA1$, $DS3$ and Lemma~\ref{L:DIS_INEQUATION}, we get
    \[\vdash  t \parallel_A s
    = \underset{i< n,i'< n'}{\bigvee}(t_{i}\parallel_A s_{i'}). \tag{\ref{L:COMP_PARALLEL}.1}\]
    We shall show that for each $i< n$ and $i'< n'$,
    \[\vdash t_{i}\parallel_A s_{i'} =r_{ii'}\;\text{for some}\;r_{ii'} \in NF_B.\]
     Let $i< n$ and $i'< n'$.
     We may assume that $t_{i}\equiv \underset{j< m_{i}}{\square}a_{ij}.t_{ij}$ and $s_{i'}\equiv \underset{j'< m_{i'}'}{\square}b_{i'j'}.s_{i'j'}$ satisfying (N), (D) and (N-$\tau$) in Def.~\ref{D:NORMAL_FORM}.
     By $EXP1$ and $EXP2$, we have
            \begin{multline*}
                \vdash t_{i}\parallel_A s_{i'} =\\
                (\underset {\begin{subarray}
                   \;j< m_i,\\
                   a_{ij} \notin A
                \end{subarray}}
                \square a_{ij}.(t_{ij} \parallel_A s_{i'}) \Box
                \underset {\begin{subarray}
                   \;j'< m_{i'}',\\
                   b_{i'j'} \notin A
                \end{subarray}}
                \square
                b_{i'j'}.(t_{i} \parallel_A s_{i'j'})) \Box
             \underset {\begin{subarray}
                   \;j< m_i,j'<m_{i'}',\\
                  a_{ij}= b_{i'j'}\in A
                \end{subarray}}
                \square
                a_{ij}.(t_{ij} \parallel_A s_{i'j'}). \tag{\ref{L:COMP_PARALLEL}.2}
            \end{multline*}
    We consider two cases.\\

\noindent    Case 1 $m_i=0$ or $m_{i'}'=0$.

    W.l.o.g, assume that $m_i=0$. Then, by (\ref{L:COMP_PARALLEL}.2), $EC1$, $EC4$, CONTEXT and TRANS, we get
    \[ \vdash t_i \parallel_A s_{i'} = \underset {\begin{subarray}
                   \;j'< m_{i'}',\\
                   b_{i'j'} \notin A
                \end{subarray}}
                \square
                b_{i'j'}.( 0 \parallel_A s_{i'j'}). \tag{\ref{L:COMP_PARALLEL}.3}\]
    If $\{b_{i'j'} \notin A|j'<m_{i'}'\}=\emptyset$ then $\vdash t_i \parallel_A s_{i'} = 0$.
    Next, we consider the case where $\{b_{i'j'}\notin A|j'<m_{i'}'\} \not= \emptyset$.
    For each $j' < m_{i'}'$ with $b_{i'j'}\notin A$, we have $s_{i'j'}\in NF_B$, moreover, $|t|+|s|>|0|+|s_{i'j'}|$.
    Then, by IH, we get $\vdash 0\parallel_A s_{i'j'} = p_{j'}$ for some $p_{j'} \in NF_B$.
    Therefore, by CONTEXT, TRANS and (\ref{L:COMP_PARALLEL}.3), it is easy to see that $\vdash t_i \parallel_A s_{i'} =  r_{ii'}$ for some $r_{ii'} \in NF_B$.\\

\noindent    Case 2 $m_i > 0$ and  $m_{i'}' > 0$.

    In such case, for each $j<m_i$ and $j'<m_{i'}'$, we have $|t|+|s|>|t_{ij}|+|s_{i'}|$, $|t|+|s|>|t_{i}|+|s_{i'j'}|$ and $|t|+|s|>|t_{ij}|+|s_{i'j'}|$.
      Moreover, $t_{ij},s_{i'},t_{i},s_{i'j'}\in NF_B$.
      Then, by IH, there exist $t_{iji'},t_{ii'j'},t_{iji'j'}\in NF_B$ such that $ \vdash  t_{ij} \parallel_A s_{i'} = t_{iji'}$, $\vdash t_{i} \parallel_A s_{i'j'} = t_{ii'j'}$  and $ \vdash  t_{ij} \parallel_A s_{i'j'} =t_{iji'j'}$.
            Set
    \begin{enumerate}
      \item  $S_1 \triangleq \underset{ \begin{subarray} \;a_{ij} \notin A,\\j<m_i\end{subarray}}{\square}a_{ij}.t_{iji'}$,
      \item  $ S_2 \triangleq \underset{ \begin{subarray} \;b_{i'j'} \notin A,\\j'<m_{i'}'\end{subarray}}{\square}b_{i'j'}.t_{ii'j'}$,
      \item $S_3 \triangleq \underset{ \begin{subarray} \;a_{ij} = b_{i'j'}\in A,\\j'<m_{i'}',j<m_i\end{subarray}}{\square}a_{ij}.t_{iji'j'}$.
    \end{enumerate}
    Clearly, $S_1,S_2,S_3\in NF_B$ and $\vdash t_{i}\parallel_A s_{i'}=(S_1 \Box S_2) \Box S_3$.
    Further, by Lemma~\ref{L:BIG_SQUARE_EC}, we get $\vdash t_{i}\parallel_A s_{i'}=r_{ii'}$ for some $r_{ii'} \in NF_B$, as desired.

     In summary, by the discussion above, we conclude that, for each $i<n$ and $i'<n'$, $\vdash t_{i}\parallel_A s_{i'} =r_{ii'}$ for some $r_{ii'}\in NF_B$.
     Then, by Def.~\ref{D:NORMAL_FORM} and (\ref{L:COMP_PARALLEL}.1), it immediately follows that $\vdash t \parallel_A s = r $ for some $r \in NF_B$, as desired.
\end{proof}

Now, we can prove that each process term is normalizable. That is,  we have the result below.

\begin{theorem}[Normal Form Theorem]\label{T:NORMALFORM}
  For each $t \in T(\Sigma_{CLL})$, $\vdash  t = s$   for some $s \in NF$.
\end{theorem}
\begin{proof}
  We prove it by induction on the structure of  $t$.

\noindent $\bullet$ $t \equiv 0$ or $t \equiv \bot$.

    Trivially.

\noindent $\bullet$ $t \equiv \alpha.t_1$.

        By IH and CONTEXT, we get $\vdash t = \alpha.t_1'$ for some $t_1' \in NF$.
        If $t_1'\not\equiv \bot$ and $\alpha\in Act$, then $\alpha.t_1'\in NF_B$.
        If $t_1' \equiv \bot$, by  $PR1$, $PR2$ and TRANS, we obtain $\vdash t = \bot$.
        If $\alpha = \tau$, by $PR2$ and TRANS, we have  $\vdash t = t_1'$.

\noindent $\bullet$ $t \equiv t_1 \odot t_2$ with $\odot \in \{\vee,\Box,\wedge,\parallel_A\}$.

        For $i=1,2$, by IH, we have $\vdash t_i = t_i'$ for some $t_i'\in NF$. We distinguish four cases based on $\odot$.\\

\noindent Case 1 $\odot = \vee$.

        If  $t_1' \not\equiv \bot$  and  $t_2' \not\equiv \bot$ (i.e., $t_1',t_2'\in NF_B$), then it immediately follows from $DI1$, $DI2$, CONTEXT and TRANS that $\vdash t = s$ for some $s \in NF_B$.
        Otherwise, w.l.o.g, assume that $t_1' \equiv \bot$.
        Then, by $DI1$, $DI4$ and TRANS, we get $\vdash  t = t_2'$.\\

\noindent Case 2 $\odot = \Box$.

        If either $t_1' \equiv \bot$ or $t_2' \equiv \bot$,  then it follows from $EC1$ and $EC5$ that $\vdash  t = \bot$.
        In the following, we consider the case where $t_1' \not\equiv \bot$ and $t_2' \not\equiv \bot$.
        In such case, we get $t_1',t_2'\in NF_B$.
        So, we may assume that $t_1' \equiv \underset{i< n}{\bigvee}\underset{j<m_i}\square a_{ij}.s_{ij}$ and $t_2' \equiv \underset{i'< n'}{\bigvee}\underset{j'<m_{i'}'}\square b_{i'j'}.r_{i'j'}$ with $\underset{j<m_i}\square a_{ij}.s_{ij},\underset{j'<m_{i'}'}\square b_{i'j'}.r_{i'j'}\in NF_B$ for each $i<n$ and $i'<n'$.
        Thus, by $DI1$, $DI2$, CONTEXT, TRANS, $DS1$ and Lemma~\ref{L:DIS_INEQUATION},  we obtain
        \[\vdash t_1 \Box t_2 =
        \underset{i < n, i' < n'}{\bigvee}(\underset{j<m_i}\square a_{ij}.s_{ij} \Box \underset{j'<m_{i'}'}\square b_{i'j'}.r_{i'j'}).\]
        Further, by Lemma~\ref{L:BIG_SQUARE_EC} and Def.~\ref{D:NORMAL_FORM}, it immediately follows that $\vdash t_1 \Box t_2 = t_3$ for some $t_3\in NF_B$.\\

\noindent Case 3 $\odot = \wedge$.

         If $t_i'\in NF_B$ for $i=1,2$ then, by Lemma~\ref{L:COMP_CONJ}, we have $\vdash t = t_3$ for some $t_3 \in NF$, otherwise, by $CO1$ and $CO4$, we get $\vdash t = \bot$.\\

\noindent Case 4 $\odot = \parallel_A$.

        If either $t_1' \equiv \bot$ or $t_2' \equiv \bot$ then, by $PA1$ and $PA2$, we get $\vdash t = \bot$.
        Otherwise,  we have $t_1',t_2' \in NF_B$, so, by Lemma~\ref{L:COMP_PARALLEL}, we obtain $\vdash  t= s$ for some $s \in NF_B$.
\end{proof}

We now turn our attention to the ground-completeness of $AX_{CLL}$.
First, we state a trivial result about general disjunction.

\begin{lemma}\label{L:SUBSTRUCTURE}
  Let $n>0$ and $t_i$ be stable for each $i< n$.
  \begin{enumerate}
    \item If $\underset{i< n}{\bigvee}t_i \notin F_{{CLL}}$ then $\underset{i< n}{\bigvee}t_i \stackrel{\epsilon}{\Longrightarrow}_F| t_{i}$ for each $i < n$.
    \item If $\underset{i< n}{\bigvee}t_i \stackrel{\epsilon}{\Longrightarrow}_{CLL}| t'$ then $t'\equiv t_{i_0}$ for some $i_0 < n$.
  \end{enumerate}
\end{lemma}
\begin{proof}
Proceed by induction on $n$, omitted.
%
%
%
\end{proof}

An important step in proving the ground-completeness is:

\begin{lemma}\label{L:COMPLETENESS}
  If $t_1,t_2 \in NF$ and $t_1 \underset{\thicksim}{\sqsubset}_{RS} t_2$, then $\vdash  t_1 \leqslant t_2$
\end{lemma}
\begin{proof}
    We prove it by induction on the degree of $t_1$.
    Since $t_1 \underset{\thicksim}{\sqsubset}_{RS} t_2$, both $t_1$ and $t_2$ are stable. Further, since $t_1,t_2\in NF$, we get, for $i=1,2$

    \[t_i \equiv 0\;\text{or}\; t_i \equiv \bot\;\text{or}\; t_i \equiv \underset{j<n_i}{\square}a_{ij}.t_{ij}\in NF_B\;\text{with}\; n_i>0.  \tag{\ref{L:COMPLETENESS}.1}\]

    Therefore, the argument splits into three cases below.\\

\noindent Case 1 $t_1 \equiv \bot$.

        Then, by $DI1$, $DI4$, $DI5$ and TRANS, we have $\vdash  t_1 \leqslant t_2$.\\

\noindent Case 2 $t_1 \equiv 0$.

       By Lemma~\ref{L:F_NORMAL}(5) and \ref{L:Basic_I}(3), $t_1 \notin F_{{CLL}}$ and ${\mathcal I}(t_1)=\emptyset$. Further, since $t_1 \underset{\thicksim}{\sqsubset}_{RS} t_2$, we get $t_2 \notin F_{{CLL}}$ and ${\mathcal I}(t_1)={\mathcal I}(t_2)$. Thus, by (\ref{L:COMPLETENESS}.1), we have $t_2 \equiv 0$. Then, $\vdash  t_1 \leqslant t_2$ follows from REF.\\

\noindent Case 3 $t_1 \equiv \underset{i< n}{\square}a_i.t_{1i} $ with $n>0$.

        Since $t_1 \in NF_B\subseteq T(\Sigma_B)$, by Lemma~\ref{L:BPT}, we have $t_1 \notin F_{{CLL}}$.
        Hence, by $t_1 \underset{\thicksim}{\sqsubset}_{RS} t_2$,  we get $t_2 \notin F_{{CLL}}$ and ${\mathcal I}(t_2)={\mathcal I}(t_1)=\{a_i|i < n\} \not= \emptyset$.
        Further, it follows from (\ref{L:COMPLETENESS}.1) and the condition (D) in Def.~\ref{D:NORMAL_FORM} that
        there exist $t_{2i}\in NF_B$ and $a_i'\in Act$($i<n$) such that
         \[t_2 \equiv \underset{i < n}{\square}a_i'.t_{2i}\in NF_B\;\text{and}\; \{a_i|i<n\}=\{a_i'|i<n\}.\]
       By CONTEXT, it is easy to know that, in order to complete the proof, it is enough to show that
       \[\forall i<n \exists i'<n(\vdash a_i.t_{1i}\leqslant a_{i'}'.t_{2i'}).\]
    Let $i_0< n$. We have $a_{i_0}=a_{i_0'}'$ for some $i_0'<n$.
    Since $t_{1i_0},t_{2i_0'}\in NF_B$, by Def.~\ref{D:NORMAL_FORM}, there exist $m,m'>0$, $s_j(j<m)$ and $s_{j'}'(j'<m')$ such that
    \begin{enumerate}
      \item $t_{1i_0} \equiv  \underset{j< m}{\bigvee}s_{j}$ and $t_{2i_0'}\equiv  \underset{j'< m'}{\bigvee} s_{j'}'$,
      \item $s_j$ and $s_{j'}'$ are stable for each $j<m$ and $j'<m'$,
      \item $s_j,s_{j'}' \in NF_B$ for each $j<m$ and $j'<m'$.
    \end{enumerate}
    In the following, we want to show that $\vdash  s_{j}\leqslant t_{2i_0'}$ for each $j< m$.
     Let $j_0< m$.
     Since $NF_B \subseteq T(\Sigma_B)$, by Lemma~\ref{L:BPT} and \ref{L:SUBSTRUCTURE}(1), it immediately follows that $t_{1i_0}\stackrel{\epsilon}{\Longrightarrow}_F|s_{j_0}$.
    Thus, $t_1 \stackrel{a_{i_0}}{\longrightarrow}_F t_{1i_0} \stackrel{\epsilon}{\Longrightarrow}_F| s_{j_0}$.
    Then, it follows from  $t_1 \underset{\thicksim}{\sqsubset}_{RS} t_2$ that
    \[t_2 \stackrel{a_{i_0}}{\Longrightarrow}_F|t_2' \;\text{and}\; s_{j_0}\underset{\thicksim}{\sqsubset}_{RS} t_2'\;\text{for some}\; t_2'.\tag{\ref{L:COMPLETENESS}.2}\]
    Further, since $t_2$ is injective in prefixes and $t_2$ is stable, we get $t_2 \stackrel{a_{i_0}}{\longrightarrow}_F t_{2i_0'}\stackrel{\epsilon}{\Longrightarrow}_F|t_2'$.
    Then, by Lemma~\ref{L:SUBSTRUCTURE}(2), we obtain
    \[t_2'\equiv s_{j_0'}' \;\text{for some}\; j_0'< m'.\tag{\ref{L:COMPLETENESS}.3}\]
    Since $|t_1|>|s_{j_0}|$, by (\ref{L:COMPLETENESS}.2), (\ref{L:COMPLETENESS}.3) and IH, we get $\vdash  s_{j_0}\leqslant s_{j_0'}'$.
    Further, by $DI1$, $DI2$, $DI5$ and TRANS, we have $\vdash  s_{j_0}\leqslant t_{2i_0'}$, as desired.

    So far, we have obtained
    \[\vdash s_{j} \leqslant t_{2i_0'}\;\text{for each}\;j<m.\]
    Then, by $DI1$, $DI2$, $DI3$, CONTEXT and TRANS, we get $\vdash \underset{j<m}{\bigvee}s_{j} \leqslant t_{2i_0'}$, that is, $\vdash t_{1i_0}\leqslant t_{2i_0'}$.
    So, by CONTEXT, it follows that $\vdash  a_{i_0}.t_{1i_0} \leqslant a_{i_0'}'.t_{2i_0'}$.
\end{proof}

We are now ready to prove the following, the main result of this section.

\begin{theorem}[Ground-Completeness]
  For any $t_1,t_2 \in T(\Sigma_{CLL})$, $M_{CLL} \models t_1 \leqslant t_2$ implies $ \vdash t_1 \leqslant t_2$.
\end{theorem}
\begin{proof}
  Assume that $M_{CLL} \models t_1 \leqslant t_2$. Thus, $t_1 \sqsubseteq_{RS} t_2$.
  By Theorem~\ref{T:NORMALFORM}, $\vdash  t_1 = t_1^*$ and $\vdash  t_2 = t_2^*$ for some $t_1^*,t_2^* \in NF$.
  It suffices to prove that $\vdash  t_1^* \leqslant t_2^*$.
  By Theorem~\ref{T:SOUNDNESS}, we have $t_1 =_{RS} t_1^*$ and $t_2 =_{RS} t_2^*$. So, $t_1^* \sqsubseteq_{RS} t_2^*$.

   If $t_1^* \equiv \bot$ then it follows from $DI1$, $DI4$, $DI5$ and TRANS that $\vdash  t_1^* \leqslant t_2^*$.
   Next, we consider the case $t_1^* \not\equiv \bot$. Then, $t_1^* \in NF_B$.
   We may assume $t_1^* \equiv \underset{i< n}{\bigvee}t_{1i}$ with $n>0$ and for each $i<n$, $t_{1i} \equiv \underset{j<m_i}\square a_{ij}.r_{ij}\in NF_B$ with $m_i \geq 0$.
    In order to complete the proof, it is enough to show that
    \[\vdash t_{1i}\leqslant t_2^*\;\text{for each}\;i< n.\]
    Let $i_0 < n$.
    Since $NF_B \subseteq T(\Sigma_B)$, by Lemma~\ref{L:BPT} and \ref{L:SUBSTRUCTURE}(1), we have $t_1^* \stackrel{\epsilon}{\Longrightarrow}_F|t_{1i_0}$.
    Then, it follows from $t_1^* \sqsubseteq_{RS} t_2^*$ that $t_2^* \stackrel{\epsilon}{\Longrightarrow}_F| t_2'$ and $t_{1i_0}\underset{\thicksim}{\sqsubset}_{RS} t_2'$ for some $t_2'$.
    So, $t_2^* \notin F_{CLL}$, that is, $t_2^* \not\equiv \bot$.
    Thus, $t_2^* \in NF_B$ and we may assume that  $t_2^* \equiv \underset{i< k}{\bigvee}t_{2i}$ with $k>0$ and for each $i<k$, $t_{2i} \equiv \underset{j<m_i'}\square b_{ij}.s_{ij}\in NF_B$ for some $m_i' \geq 0$.
    Thus, $t_{2i}$ is stable for each $i<k$.
    Then, by Lemma~\ref{L:SUBSTRUCTURE}(2), it follows from $t_2^* \stackrel{\epsilon}{\Longrightarrow}_F| t_2'$ that $t_2' \equiv t_{2i_0'}$ for some $i_0'< k$.
    Further, by Lemma~\ref{L:COMPLETENESS}, $\vdash t_{1i_0}\leqslant t_{2i_0'}$ follows from $t_{1i_0}\underset{\thicksim}{\sqsubset}_{RS} t_2' \equiv t_{2i_0'}$.
    Finally, by $DI1$, $DI2$, $DI5$ and TRANS, we obtain $\vdash t_{1i_0}\leqslant t_2^*$, as desired.
\end{proof}

\section{Conclusions and Further Work}

Inspired by L{\"u}ttgen and Vogler's work in \cite{Luttgen10}, this paper considers a process calculus with logical operators. Two different views of the language CLL are explored in detail:

\noindent --- a behavioral view, ready simulation,

\noindent --- a proof-theoretic view, the axiomatic system $AX_{CLL}$.

The soundness and completeness of $AX_{CLL}$ reveal that the above two views are equivalent, that is,
\[t\sqsubseteq_{RS} s\;\text{iff}\;\vdash t\leqslant s\;\text{for any}\;t,s \in T(\Sigma_{CLL}).\]

CLL is designed as a process calculus for Logic LTSs, which rephrases L\"{u}ttgen and Vogler's setting in a process-algebraic style.
The constructors prefix, conjunction, disjunction, external choice and parallel over Logic LTSs are captured by corresponding operators in CLL.
Moreover, a number of properties concerning these constructors are re-established in behavioral theory of CLL by very different method.
Similar to usual process algebras, this paper develops behavioral theory based on the SOS rules which specify the behavior of process terms, while L{\"u}ttgen and Vogler establish these properties depending on the constructions of Logic LTSs, and do not refer to any syntactical element.
Compared with their work, the main contribution of this paper is to present a sound and ground-complete axiomatic system of ready simulation in the presence of logic operators.

It is well known that, in addition to behavior and proof-theoretic views, the language of process algebra may be interpreted in denotational view (see, e.g. \cite{Hennessy88}).
The denotational method aims at defining a denotational function which associates semantic objects to process terms.
Such function is often given recursively by induction on the structure of process terms.
It is easy to see that constructors explored by L{\"u}ttgen and Vogler in \cite{Luttgen10} is useful when considering denotational semantics of CLL.
In fact, we can show the result below\\

\noindent \textbf{Observation} $(G(t) \odot^* G(s))\downarrow (t\odot^* s) =_{RS} G(t\odot s)$ for $t,s \in T(\Sigma_{CLL})$ and $\odot \in \{\wedge,\vee,\Box,\parallel_A\}$.\\

Here, $G(t)$ denotes the sub-LTS of $LTS(CLL)$ generated by the process term $t$, $\odot^*$ is the constructor over Logic LTSs corresponding to $\odot$, and $(G(t) \odot^* G(s))\downarrow (t\odot^* s)$ is the sub-LTS of $G(t) \odot^* G(s)$ generated by the state $t\odot^* s$\footnote{Since $t$ and $s$ are states in $G(t)$ and $G(s)$ respectively, according to the construction in \cite{Luttgen10}, $G(t) \odot^* G(s)$ contains the state labelled by $t\odot^* s$.}.
This result suggests to us that, based on the constructors in \cite{Luttgen10}, it seems not difficult to provide a denotational semantics for CLL, which is fully abstract with respect to operational semantics presented in this paper.
We leave it as further work.

This paper adopts the predicate $F$ to describe unimplementable processes.
Follows \cite{Luttgen10}, this predicate is involved in the notion of ready simulation.
In this sense, the value of $F$ is regarded as an observable signal of processes.
The process algebra with observations of propositional formulae have been considered by Baeten and Bergstra in \cite{Baeten97}.
The motivation behind their work lies in providing a framework to deal with conditional process expression (e.g., $x\triangleleft \phi \triangleright y$) based on the view that the visible part of the state of a process is a proposition.
Clearly, this is another method to incorporate logical components with process algebras.
An interesting research is to compare it with the framework adopted in \cite{Luttgen10} and this paper, in which logical operators over processes are introduced directly.

This paper focuses on exploring logical constructors of Logic LTSs in process algebraic style, some useful operators occurring in usual process algebras, such as hiding, recursion et.al, are not involved in  CLL.
Extending CLL by incorporating these operators is worth further investigation.\\


\end{document}